\documentclass[letterpaper]{article}

\usepackage{times}  
\usepackage{helvet} 
\usepackage{courier}  
\usepackage[hyphens]{url}  
\usepackage{graphicx} 
\usepackage[ bottom=1.25in, top=1.25in, left=0.85in, right=0.85in]{geometry}
\usepackage{float}

\usepackage{caption} 
\usepackage{authblk}
\usepackage{hyperref}
\title{On the Problem of Underranking in Group-Fair Ranking}
\date{}

\author[$\star$]{Sruthi Gorantla}
\author[$\dagger$]{Amit Deshpande}
\author[$\star$]{Anand Louis}

\affil[$\star$]{\small{Indian Institute of Science, Bangalore, India. \texttt{\{gorantlas, anandl\}@iisc.ac.in}}}
\affil[$\dagger$]{\small{Microsoft Research, Bangalore, India. \texttt{amitdesh@microsoft.com}}}

\usepackage{amsthm}

\usepackage{amssymb}

\usepackage{amsmath}

\usepackage[ruled,linesnumbered]{algorithm2e}

\usepackage{xcolor}

\usepackage{soul}
\usepackage{cleveref}

\newtheorem{theorem}{Theorem}[section]
\newtheorem*{theorem*}{Theorem}

\newtheorem*{claim*}{Claim}

\newtheorem*{proposition*}{Proposition}
\newtheorem{lemma}[theorem]{Lemma}
\newtheorem*{lemma*}{Lemma}
\newtheorem{corollary}[theorem]{Corollary}

\newtheorem*{conjecture*}{Conjecture}

\newtheorem*{fact*}{Fact}

\newtheorem*{hypothesis*}{Hypothesis}

\theoremstyle{definition}
\newtheorem{definition}[theorem]{Definition}

\usepackage{enumitem}

\renewcommand{\leq}{\leqslant}
\renewcommand{\le}{\leqslant}
\renewcommand{\geq}{\geqslant}

\newcommand{\defeq}{\stackrel{\textup{def}}{=}}

\newcommand{\paren}[1]{\left(#1 \right )}

\newcommand{\set}[1]{\left\{#1\right\}}

\newcommand{\abs}[1]{\left\lvert#1\right\rvert}

\newcommand{\ceil}[1]{\left\lceil #1 \right\rceil}
\newcommand{\floor}[1]{\left\lfloor #1 \right\rfloor}

\newcommand{\Z}{{\mathbb Z}}

\newcommand{\R}{\mathbb R}
\newcommand{\Q}{\mathbb Q}

\DeclareMathOperator*{\argmin}{argmin}

\usepackage{amsthm, amssymb, amsmath}
\usepackage{ soul, cleveref}

\newcommand\numberthis{\addtocounter{equation}{1}\tag{\theequation}}
\usepackage{caption}
\usepackage{subcaption}

\newcommand{\pak}{{\sc precision}$@K$}

\newcommand{\amin}{\alpha_{\text{min}}}
\newcommand{\amax}{\alpha_{\text{max}}}

\newcommand{\asum}{\sum_{l \in [\ell]}\alpha_l }
\newcommand{\bsum}{\sum_{l \in [\ell]}\beta_l }

\begin{document}

\maketitle

\begin{abstract}
	Search and recommendation systems, such as search engines, recruiting tools, online marketplaces, news, and social media, output ranked lists of content, products, and sometimes, people. Credit ratings, standardized tests, risk assessments output only a score, but are also used implicitly for ranking. Bias in such ranking systems, especially among the top ranks, can worsen social and economic inequalities, polarize opinions, and reinforce stereotypes. On the other hand, a bias correction for minority groups can cause more harm if perceived as favoring group-fair outcomes over meritocracy.
	
	In this paper, we formulate the problem of underranking in group-fair rankings, which was not addressed in previous work. Most group-fair ranking algorithms post-process a given ranking and output a group-fair ranking. We define underranking based on how close the group-fair rank of each item is to its original rank, and prove a lower bound on the trade-off achievable for simultaneous underranking and group fairness in ranking. We give a fair ranking algorithm that takes any given ranking and outputs another ranking with simultaneous underranking and group fairness guarantees comparable to the lower bound we prove. Our algorithm works with group fairness constraints for any number of groups. Our experimental results confirm the theoretical trade-off between underranking and group fairness, and also show that our algorithm achieves the best of both when compared to the state-of-the-art baselines.
\end{abstract}

\section{Introduction}
Search and recommendation systems have revolutionized the way we consume an overwhelming amount of data and find relevant information quickly \cite{BrinPage1998,Adomavicius2005}. 
They help us find relevant documents, news, media, people, places, products and rank them based on our interests and intent \cite{Kofler2016,Pei2019}. 
Information presented through ranked lists influences our worldview \cite{Pariser2011,Tavani2020}. 
Biased ranking of news, people, products raises ethical concerns and can potentially cause long-term economic and societal harm to demographics and businesses \cite{Noble2018}. 
Many state-of-the-art rankings that maximize utility or relevance reflect existing societal biases and are often oblivious to the societal harm they may cause by amplifying such biases. 
When these systems amplify societal biases observed in their training data, they worsen social and economic inequalities, polarize opinions, and reinforce stereotypes \cite{ONeil2016}. 
In addition to ethical concerns, there are also legal obligations to remove bias. Disparate impact laws prohibit even unintentional but biased outcomes in employment, housing, and many other areas if one group of people belonging to a protected group is adversely affected compared to another \cite{Barocas2016}. 

\paragraph{Fairness in ranking.} 
Fairness in ranking has three broad requirements: \emph{sufficient presence} of items belonging to different groups, \emph{consistent treatment} of similar individuals, and \emph{proper representation} to avoid representational harm to members of protected groups \cite{Castillo2019survey}. 
The first and the third requirements are about fairness to groups, whereas the second requirement is about fairness to individuals. 
For example, diversity alone in top ranks satisfies sufficient presence but need not provide consistent treatment and proper representation in the way the items are ranked. 
Fair ranking algorithms can be divided into two categories. 
Re-ranking algorithms that modify a given ranking of high utility to incorporate fairness constraints while trying to preserve the original utility, and learning-to-rank algorithms that incorporate fairness and utility objectives together into learning a ranker from training data. 
Re-ranking can be used to post-process the prediction as well as to pre-process the training data of any given ranker. 

Most of the fair ranking algroithms are designed to output group-fair ranking.
Group fairness in machine learning literature has focused on outcome-based or proportion-based definitions of fairness (e.g., demographic parity, equality of opportunity) \cite{Hardt2016, Barocas2019}.
Although group-fairness is a desirable goal, affirmative action to achieve group-fairness is often misinterpreted as non-meritocratic by individuals and requires a deeper understanding \cite{Crosby2004}. In this context, we argue that it is important to measure the worst-case effect of group-fair ranking on the individuals, which is not addressed by previous work.

\paragraph{Our contributions.}
To the best of our knowledge, our paper is the first to study how group-fair re-ranking affects individual ranks in the worst case. 
Previous work has looked at re-ranking or learning-to-rank with group-fairness and aggregate individual-fairness (or consistency) constraints.
Our group-fairness definition ensures sufficient presence of all groups, similar to previous work, but we give a new, natural definition called \emph{underranking} to study how re-ranking affects the merit-based ranks of individual items in the worst case. Our main contributions can be summarized as follows.
\begin{itemize}[leftmargin=*] 
	\item We define underranking based on the worst-case deviation from the true, merit-based (or color-blind) rank of any individual item (see \Cref{def:ifair}). This directly captures the loss of visibility suffered by individual items of high merit in the process to achieve high group-fairness. We prove a lower bound on the trade-off achievable between underranking and group-fairness simultaneously.
	\item We propose a re-ranking algorithm that takes a given merit-based (or color-blind) ranking and outputs another ranking with simultaneous underranking and group-fairness guarantees, comparable to the lower bound mentioned above. Our algorithm is fast, flexible, and can accommodate multiple groups, each with a different constraint on their group-wise representation.
	\item We do extensive experiments to show that our algorithm achieves the best of both underranking and group-fairness compared to the baselines on standard real-world datasets such as COMPAS recidivism and German credit risk. Moreover, our algorithm runs significantly faster than the baselines.
\end{itemize}

\paragraph{Related work.}
\label{sec:related_work}
The two most important baselines related to our work are the group-fair re-ranking algorithms \cite{ranking_with_fairness_constraints,fa_ir_a_fair_top_k_ranking_algorithm}. Fair ranking to maximize ranking utility subject to upper and lower bounds on group-wise representation in the top $k$ ranks, for all values of $k$, can be framed as an integer optimization problem \cite{ranking_with_fairness_constraints}. 
The authors propose an exact dynamic programming (DP)-based algorithm, and a greedy approximation algorithm to achieve fairness in ranking for intersectional subgroups. 
The \emph{fair top-$k$ ranking problem} gives another formulation for fair re-ranking of a given true or \emph{color-blind} ranking based on numerical quality values and a given $k$, so that the top-$k$ re-ranking maximizes certain selection and ordering utilities subject to group-wise representation constraints \cite{fa_ir_a_fair_top_k_ranking_algorithm}. 
The authors give an efficient algorithm called FA*IR to solve the fair top-$k$ ranking problem.

Fair ranking has also been studied in the learning-to-rank (LTR) setting, where the output ranking is probabilistic, so the fairness and utility guarantees are often on average instead of the worst case.
Given a query-document pair, the probability of each document being ranked at the top rank is called its \textit{exposure}. 
While the traditional ListNet framework simply minimizes a loss function based on the items' true and predicted exposure \cite{learning_to_rank_from_pairwise_approach_to_listwise_approach},
an extension of this, DELTR \cite{reducing_disparate_exposure_in_ranking}, learns fair ranking via a multi-objective optimization that maximizes utility and minimizes disparate exposure for different groups of items for group-fairness or different items for individual-fairness.
This general learning-to-rank framework facilitates optimizing multiple utility metrics while satisfying equal exposure, and Fair-PG-LTR \cite{policy_learning_for_fairness_in_ranking} learns a ranking that satisfies fairness of exposure. 
Aggregate or average-case guarantees in ranking are more suited to the applications in search and recommendation, whereas the worst-case guarantees are more suited to the applications in recruitment, school admissions, healthcare etc. where the worst-case fairness to individuals could be critical.

There is related work on defining and maximizing various group-fairness metrics over each prefix of the top $k$ ranks \cite{measuring_fairness_in_ranked_outputs}, for a given $k$, using an optimization algorithm to learn fair representations \cite{learning_fair_representations}. There are also other measures of group-fairness in ranking based on pairwise comparisons \cite{pairwise_fairness_for_ranking_and_regression, fairness_in_recommendation_ranking_through_pairwise_comparisons}. Recent work has also studied fairness-aware ranking in search and recommendations for real-world recruitment tools using fairness metrics based on skew in the top $k$ and Normalized Discounted KL-divergence (NDKL) divergence \cite{Geyik2019}. 
Fairness and ranking utility trade-offs have also been studied via counterfactually fair rankings \cite{causal_intersectionality_for_fair_ranking}.

\section{Preliminaries}
Let $M, N \in \Z^+$ and $N \leq M$. 
Then, a ranking is an assignment of $M$ ranks to $N$ items such that each rank (denoted by a number in $[M]$) is assigned to at most one item (denoted by a number in $[N]$) and each item is assigned exactly one rank.
Whenever a rank is not assigned to any item, we call it an \emph{empty rank}.
Note the whenever $N = M$, there are no empty ranks in the ranking.
We say that rank $i$ is \emph{lower} than rank $j$ if $i < j$, and rank $i$ is \emph{higher} than rank $j$ if $i > j$.
In a ranking the \emph{top $m$ ranks} refer to the ranks $(1,2,\ldots, m)$, \emph{each prefix of the top $m$ ranks} is the set $\set{(1, 2, \ldots, i) | i \in \set{1,\ldots,m}}$, \emph{every $k$ consecutive ranks in the top $m$ ranks} is the set $\set{(i+1, i+2, \ldots, i+k) | i \in \set{0,\ldots,m-k}}$, \emph{$i$th block of size $k$} is the ranks $((i-1) k+1, (i-1) k+2, \ldots, (i-1)k+k)$, and hence, \emph{top $d$ blocks of size $k$} is the set $\{((i-1) k+1, (i-1) k+2, \ldots, (i-1)k+k) | i \in \set{1,...,d}\}$.
A \emph{true ranking} is the ranking of the items based on a measure of merit of the items.
A re-ranking algorithm rearranges the items in the true ranking and outputs another ranking of the items with some desired properties.
We note that a true ranking is not always available for the real-world datasets. 
In our experiments, we use some natural substitutes for the true ranking; see \Cref{sec:experiments} for  details.
An item's \emph{true rank} is its rank in the true ranking.
The set of $N$ items is partitioned into $\ell$ groups based on the sensitive attributes of the items.
We denote each group by the subscript $l$ whenever we refer to a group $l \in[\ell]$. 
In all that follows, $\alpha_l$, $\beta_l \in [0,1]$ and $\gamma \geq 1$.
The fairness constraints are based on the representation (number of items) from each group in the ranking, and are denoted by $\boldsymbol{\alpha} = (\alpha_1, \alpha_2, \ldots, \alpha_{\ell})$ and $\boldsymbol{\beta} = (\beta_1, \beta_2, \ldots, \beta_{\ell})$, where $\alpha_l, \beta_l$ represent the fairness constraints for group $l$.
And for any $c\in\R$, $c\boldsymbol{\alpha} = (c\alpha_1, c\alpha_2, \ldots, c\alpha_{\ell})$, and similary $c\boldsymbol{\beta} = (c\beta_1, c\beta_2, \ldots, c\beta_{\ell})$.

We now formally define the notion of {\em group fairness} of a ranking of $N$ items.
\begin{definition}[Group Fairness]
	\label{def:gfair}
	A ranking is said to satisfy $\paren{\boldsymbol{\alpha}, \boldsymbol{\beta}, k}$ group fairness if every $k$ consecutive ranks have at most $\alpha_l k$ and at least $\beta_l k$ items from group $l$, for every group $l \in [\ell]$.
\end{definition}
That is, each element, ranks $(i+1, ..., i+k)$, in the set of every $k$ consecutive ranks in top $N$ ranks is such that for each group $l \in [\ell]$, at most $\alpha_l k$ and at least $\beta_l k$ of these ranks are assigned to the group $l$.
The set of top $\floor{N/k}$ blocks of size $k$ is a strict subset of the set of every $k$ consecutive ranks in top $N$ ranks. 
Therefore, any ranking that satisfies group fairness constraints for every $k$ consecutive ranks in top $N$ ranks also satisfies group fairness constraints for of the top $\floor{N/k}$  blocks of size $k$. 

Using the notion of {\em underranking}, we would like to capture how much an item has been displaced from its true rank during re-ranking for group fairness.
\begin{definition}[Underranking]
	\label{def:ifair}
	A ranking satisfies $\gamma$ underranking if the rank of each item is at most $\gamma$ times its true rank. 
\end{definition}
We remark that unless the true ranking satisfies the group fairness conditions, some items with high merit must suffer a loss of visibility during the process of re-ranking for group fairness.
That is, the output group fair ranking has underranking strictly greater than 1.
This manifests the trade-off between the group fairness and the underranking in ranking.

Closely related to underranking is the well studied notion of \pak~of ranking \cite{ir_evaluation_methods, introduction_to_information_retrieval, reducing_disparate_exposure_in_ranking}.
For a given ranking, \pak~is defined as the number of items in the top $K$ ranks of the true ranking
which are still in the top $K$ ranks after re-ranking. 
We get the following relation between underranking and \pak.
\begin{corollary}
	\label{cor:precision}
	A ranking satisfying $\gamma$ underranking also has \pak~at least $\floor{K/\gamma}$, $\forall K \in \Z^+$.
\end{corollary}

\begin{proof}
	Fix a ranking having $\gamma$ underranking. 
	By definition, the top $\floor{K/\gamma}$ items in the true ranking get displaced at most to the rank $\floor{ K/\gamma} \gamma \leq K $. 
	Hence, at least the top $\floor{K/\gamma}$ items in the true ranking are also in the top $K$ ranks in ranking with $\gamma$  underranking.
	Therefore, \pak~is at least $\floor{ K/\gamma}$.
\end{proof}

\begin{figure}
	\centering
	\includegraphics[width=0.5\textwidth]{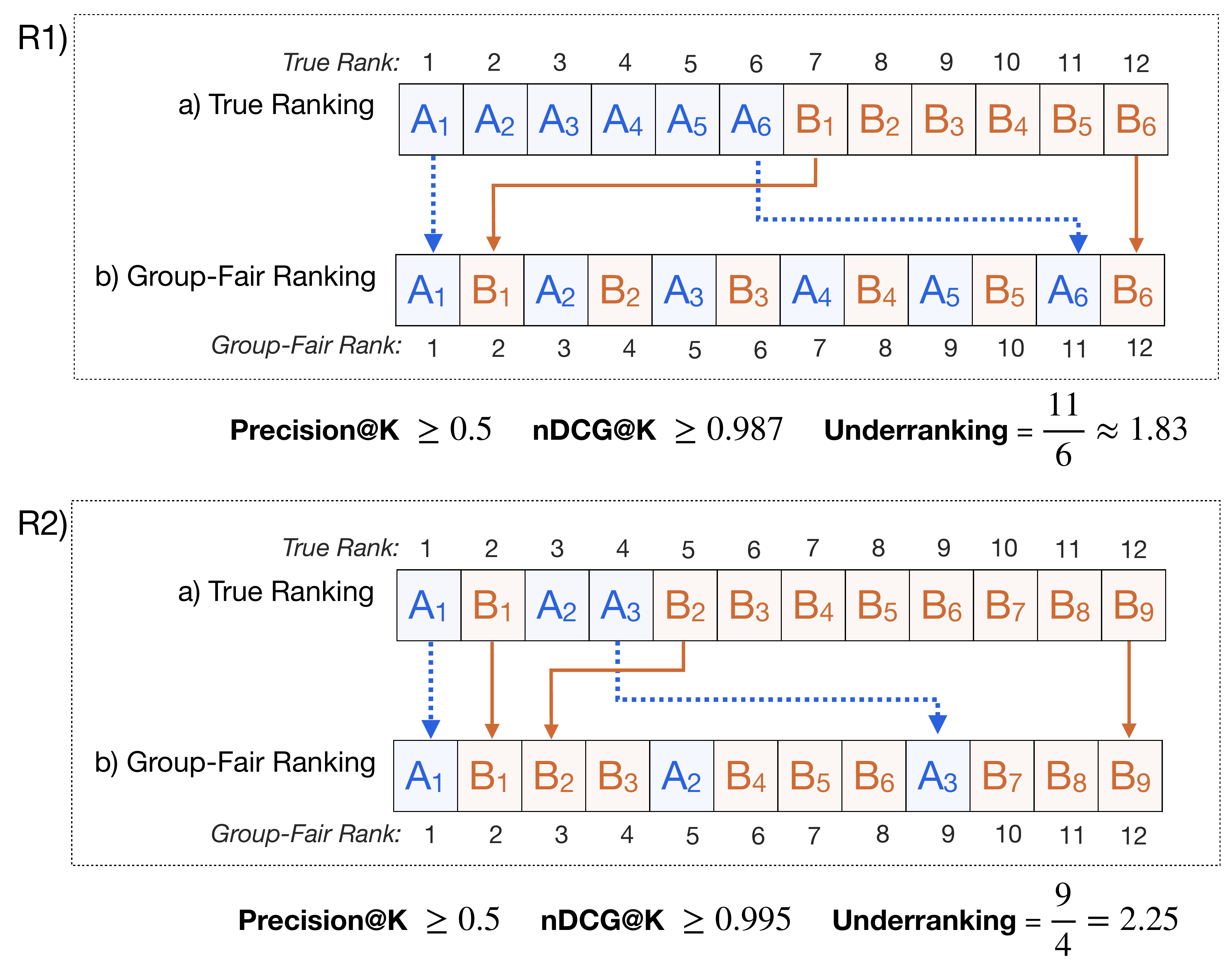}  	
	\caption{High ranking utility does not imply better (lower) underranking in group-fair rankings.}
	\label{fig:example}
\end{figure}

We note that our definition of group fairness in ranking is slightly different from previous definitions in \cite{fa_ir_a_fair_top_k_ranking_algorithm, ranking_with_fairness_constraints, Castillo2019survey}.
There, the group fairness constraints are at every prefix of the top $k$ ranking, whereas, in \Cref{def:gfair} group fairness constraints are for every $k$ consecutive ranks.
Our notion of group fairness has the desirable property that even if items from top few ranks are removed from the ranking, the remaining ranking still satisfies the group fairness constraints. 
Using this notion of group fairness, we propose \Cref{alg:ranking} that achieves simultaneous group fairness and underranking guarantees.
Such theoretical guarantees are not available for the re-ranking algorithms with prefix group fairness constraints. 
We will also see in \Cref{sec:experiments} that the algorithm proposed in this paper achieves better representation of the protected groups in the prefixes of the top $k$ ranking as well.

We also note here that, even though low (better) underranking in the top $K$ ranks implies high \pak, the converse need not be true.
Consider two pairs of a true ranking and its corresponding group-fair ranking shown in \Cref{fig:example}, R1(a), R1(b) and R2(a), R2(b) with items from groups A and B.
Both R1(b) and R2(b) satisfy proportional representation (50\% of each group in R1 and 25\% A's and 75\% B's in R2) in every prefix of the ranking (ignoring the rounding errors), as well as in every $k$ consecutive ranks, for a reasonable choice of $k$, which is $4$ or $8$.
If we assume that in both R1(a) and R2(a), the merit of the item ranked at rank $i$ is $1 - 0.01*(i-1)$, then the nDCG (see \Cref{sec:experiments} for the formula), which is also a ranking utility metric like \pak, is more for R2 than R1.
We also note that, in both R1 and R2, the \pak~in any prefix of the top 12 ranks is at least $0.5$.
From these two examples we observe that even though the utility of the group-fair ranking of R2 is greater than or equal to that of R1, the underranking in R2 is higher (worse) than that of R1.
Hence, high ranking utility may not imply better underranking.
This is also observed in our experiments on the real-world datasets.

\section{Trade-off between Underranking and Group Fairness}
\label{sec:algorithm_and_results}

Our first main result is a lower bound on the underranking when satisfying group fairness in blocks of size $k$. 

\begin{theorem}[Lower bound]
	\label{thm:igtradeoff}
	Fix $\ell \in \Z^+$. 
	For each group $ l \in [\ell]$, fix $\alpha_l, \beta_l \in (0, 1] \cap \Q $ such that $\alpha_l \geq \beta_l$, $\sum_{l \in [\ell]} \alpha_l\geq 1$, and $\sum_{l \in [\ell]} \beta_l\leq 1$.
	Fix $k \in \Z^+$. 
	For every $n_0 \in \Z^+$, there exists an $n$ such that $n \geq n_0$, and there exists a true ranking of the $N = \ell n$ items grouped into $\ell$ groups of $n$ items each, such that the following holds. 
	Any ranking satisfying $\gamma$ underranking (w.r.t. the true ranking) and 
	$\paren{\boldsymbol{\alpha}, \boldsymbol{\beta}, k}$ group fairness in the top $\frac{\gamma n}{k}$ blocks of size $k$ must have $\gamma \geq \frac{1}{\min\{\amin, 1-\sum_{l \neq l_*}\beta_l\}}$, where  $\amin = \min_l \alpha_l$ and $l_* = \argmin_l \beta_l$.
\end{theorem}
\begin{proof}
	Let $\gamma = a/b$ where $a,b \in \Z^{+}$. 
	Set $n$ to be any integer multiple of $bk$ such that $n \geq n_0$.
	Let $\hat{l} = \argmin_{l \in [\ell]}  \min\{\alpha_l, 1-\sum_{j \neq l}\beta_j\}$.
	Consider a true ranking where all the $n$ items from group $\hat{l}$ are placed in top $n$ ranks followed by the items from all the other groups. 
	Observe that by our choice of parameters, $\gamma n$ is an integer. 
	Let $\gamma n = ck$ for some $c \in \Z^+$.
	Now, consider any ranking of these items satisfying $\gamma$ underranking in the top $\gamma n$ ranks and $\paren{\boldsymbol{\alpha}, \boldsymbol{\beta}, k}$ group fairness in the top $\frac{\gamma n}{ k}$ blocks of size $k$. 
	By the definition of underranking, we get that the top $\gamma n$ ranks must contain all the $n$ items from group $\hat{l}$. 
	Since the ranking satisfies all the upper bound constraints, any of the top $c$ blocks must have at most $\alpha_{\hat{l}} k$ items from group $\hat{l}$.
	Similarly, since the ranking also satisfies all the lower bounds, any of the top $c$ blocks must have at least $\beta_l k$ items from group $l$, $\forall l \in [\ell]$. 
	Hence, there could be at most $\paren{1-\sum_{l \neq \hat{l}}\beta_l}\cdot k$ items from group $\hat{l}$. 
	This implies that, the top $ck$ ranks have at most $\paren{\min\set{\alpha_{\hat{l}}, 1-\sum_{l \neq \hat{l}}\beta_l}} \cdot ck$ items from group $\hat{l}$. 
	Therefore, the top $ck = \gamma n$ ranks contain at most $\paren{\min\set{\alpha_{\hat{l}}, 1-\sum_{l \neq \hat{l}}\beta_l}} \cdot \gamma n$ items from group $\hat{l}$.
	If $ \gamma < 1/\min\set{\alpha_{\hat{l}, 1-\sum_{l \neq \hat{l}}\beta_l}}$, then top $\gamma n$ ranks contain strictly less than $n$ elements from group $\hat{l}$, which is a contradiction.
	Therefore, we must have $\gamma \geq\frac{1}{\min_{l \in [\ell]}\min\{\alpha_l, 1-\sum_{j \neq l}\beta_j\}} = \frac{1}{\min\{\amin, 1-\sum_{l \neq l_*}\beta_l\}}$, where  $\amin = \min_l \alpha_l$ and $l_* = \argmin_l \beta_l$.
\end{proof}

Our next main result is a fair ranking algorithm that takes a true ranking and outputs another ranking with underranking and group fairness guarantees in any $k$ consecutive ranks. 

\begin{theorem}[Trade-off 1]
	\label{thm:igfair}
	Given a true ranking of $N$ items grouped into $\ell$ disjoint groups, with each group having at least $n$ items, and fairness parameters $k \in \Z^+$ and $\alpha_l, \beta_l, \forall l \in [\ell]$, where $\alpha_l, \beta_l$ define the upper and lower group fairness constraints for the group $l$ respectively such that $0 \le \beta_l \le \alpha_l \le 1$, $\sum_{l\in[\ell]}^{} \alpha_l > 1$, $\sum_{l\in[\ell]}^{} \beta_l < 1 $. 
	Let $\amin := \min_l \alpha_l$, $\amax := \max_l \alpha_l$, and $l_* = \argmin_l \beta_l$. \\
	Let $\epsilon := \frac{2}{k}\cdot\max\bigg\{\paren{1+\frac{\ell}{ \asum - 1}}, \paren{1+\frac{\ell}{ 1 - \bsum}}, \max_{l \in [\ell]} \paren{1+\frac{2}{ \alpha_l - \beta_l}} \bigg\}$. \\
	Then there exists a polynomial time algorithm to compute a ranking satisfying both of the following simultaneously,
	\begin{enumerate}
		\item $\frac{1}{\min\set{\amin - \frac{1}{\floor{\epsilon k/2}}, \paren{1-\sum_{l \neq l_*} \beta_l} - \frac{\ell - 1}{\floor{\epsilon k/2}}}}$ underranking,
		\item $\paren{(1+\epsilon)\boldsymbol{\alpha}, (1-\epsilon)\boldsymbol{\beta}, k}$ group fairness in the top $\floor{\frac{n}{\amax}} - \floor{\epsilon k/2 }$ ranks.
	\end{enumerate}
\end{theorem}
We note that $\epsilon$ need not be smaller than $1$.

We also obtain slightly stronger guarantees if we only need group fairness in blocks of size $k$ instead of group fairness guarantees for any $k$ consecutive ranks. 
That is, in each of the top $\floor{n/(\amax k)}$ blocks of size $k$, the output ranking has to be such that, for each group $l \in [\ell]$, at least $\beta_l k$ and at most $\alpha_l k$ ranks are assigned to items from group $l$. 
\begin{theorem}[Trade-off 2]
	\label{thm:igblockfair}
	Given a true ranking of $N$ items grouped into $\ell$ disjoint groups, with each group having at least $n$ items, and fairness parameters $k \in \Z^+$, and $\boldsymbol{\alpha}, \boldsymbol{\beta}$, where $\alpha_l, \beta_l$ define the upper and lower group fairness constraints for the group $l$ respectively such that $0 < \beta_l \le \alpha_l \le 1$, $\sum_{l\in[\ell]}^{} \alpha_l > 1 $ and $\sum_{l\in[\ell]}^{} \beta_l < 1 $, if the fairness parameters are also such that $\alpha_l k \in \Z^+$ and $\beta_l k \in \Z^+, \forall l \in [\ell]$, then there exists a polynomial time algorithm to compute a ranking satisfying,
	\begin{enumerate}
		\item $\frac{1}{\min\set{\amin, 1 - \sum_{l \neq l_*} \beta_l}} $ underranking,
		\item $\paren{\boldsymbol{\alpha}, \boldsymbol{\beta}, k}$ group fairness in each of the top $\floor{\frac{n}{\amax k}}$ blocks of size $k$,
	\end{enumerate}
	where  $\amin = \min_l \alpha_l$, $\amax = \max_l \alpha_l$, and $l_* = \argmin_l \beta_l$.
\end{theorem}

\begin{algorithm}[t]
	\SetAlgoLined
	\KwIn{A true ranking of the $N$ items and parameters $\alpha_l, \beta_l, \forall l \in [\ell]$, and $k$ satisfying the conditions in \Cref{thm:igfair}.}
	
	Set $\epsilon, \amin, l_*$ as in \Cref{thm:igfair}, set $B := \floor{\frac{\epsilon k}{2}}$
	
	Set $b := \min\left\{\floor{\amin B}, B - \sum_{l \neq l_*} \ceil{\beta_lB } \right\}$
	
	Set $M := \ceil{NB/b}$
	
	\For{$i := \ceil{N/b}$ down to $1$} 
	{	
		\For{$j := 1$ to  $\min\{b, N - (i-1)b\}$}
		{
			Move item at rank $(i-1)b + j$ to rank $(i-1)B + j$ 
			
		}
		
	}\label{step:move-down}
	\For{each rank $j$ in $1$ to $M$}
	{ 
		
		\If{rank $j$ is empty}
		{ 
			
			Set $i := \ceil{j/B}$ 
			
			\For{$j' := j+1$ to $M$ such that rank $j'$ is not empty}
			{
				
				Set $l := $ group the item ranked at $j'$ belongs to
				
				\If{ the lower bound for group $l$ in the block $i$ is not satisfied $\lor$ (lower bounds of all the groups are satisfied $\land$ upper bound for group $l$ would not be violated)     }
				{	\label{step:condition}
					
					Move the item at rank $j'$ to rank $j$
					
					Break the loop
					
				}
				
			}
			
		}
	}\label{step:move-up}
	\For{$j := 1$ to $N$}
	{
		\If{rank $j$ is empty}
		{
			Move to rank $j$, the first item at rank higher than $j$\;
		}
	} \label{step:fill-tail}
	Output final ranking from rank $1$ to rank $N$\;
	\caption{ALG}
	\label{alg:ranking}
\end{algorithm}

\paragraph{Overview of the Algorithm~\ref{alg:ranking} and proof outline.}
We defer all the proofs to the supplementary material.
Here we present an overview of \Cref{alg:ranking}. 
We use Algorithm \ref{alg:ranking} to prove Theorem \ref{thm:igfair}. 
We invoke Algorithm \ref{alg:ranking} with a different value of $\epsilon$ to prove 
Theorem \ref{thm:igblockfair}.
Let $\epsilon, b$, and $B$ be as set in the algorithm.
The $i$th \emph{block} of size $\floor{\epsilon k/2}$ refers to the ranks $(i-1) \floor{\epsilon k/2} + 1$ to $i \floor{\epsilon k/2}$.
We are given a true ranking of $N$ items. Hence we first assume that the length of the true ranking is $M = \ceil{NB/b}$ such that the ranks $N+1$ to $M$ are empty in the beginning of the algorithm.
We first move the items to a rank higher than their true ranks in a fashion such that at the end of Step \ref{step:move-down} the underranking of this intermediate ranking consisting of $M$ ranks is bounded (see \Cref{lem:ifair}).
By our choice of parameters, this also guarantees that none of the blocks have more than $\min\big\{\floor{\amin \floor{\epsilon k/2}}, \floor{\epsilon k/2} - \sum_{l \neq l_*} \ceil{\beta_l\floor{\epsilon k/2}}\big\}\cdot \floor{\epsilon k/2}$ items from any group.  
Next, we greedily fill up the empty ranks starting from the rank $1$ while ensuring that the group fairness is not violated, until there are items available from each group.
We use the fact that there are at least $n$ items from each group to show that if there is any empty rank in the top $\floor{n/\amax} - \floor{\epsilon k/2} $ ranks, then there will be at least one higher ranked item available from each group which can be assigned to the empty rank without violating the condition in Step \ref{step:condition}.
Therefore, top $\floor{n/\amax} - \floor{\epsilon k/2} $ ranks will be unchanged after Step \ref{step:move-up}.

Then we fill the remaining empty ranks till $N$ while ensuring that the underranking does not get worse.
It is easy to show that after Step \ref{step:fill-tail}, each of the top $\floor{\frac{n}{\floor{\amax  \floor{\epsilon k/2} } } } $ blocks have at most $\floor{\alpha_l \floor{\epsilon k/2}}$ items and at least $\ceil{\beta_l \floor{\epsilon k/2}}$ items from group $l$ for each $l$.
Finally we output the top $N$ ranks.
Observe that any $k$ consecutive ranks must include some number of blocks completely, and will intersect at most 
two blocks partially. 
Therefore, in the worst case, the number of items from a group $l$ in any $k$ consecutive ranks of the top $\floor{n/\amax} - \floor{\epsilon k/2}$ ranks will be at most $\alpha _l k + 2 \alpha_l \floor{\epsilon k /2} \leq \alpha_l(1 + \epsilon)k$, and at least $\beta _l k - 2 \beta_l \floor{\epsilon k /2} \geq \beta_l(1 - \epsilon)k$.
This gives us our group fairness guarantee in the top  $\floor{n/\amax} - \floor{\epsilon k/2}$ ranks.

\subsection{Proof of \Cref{thm:igfair} and \Cref{thm:igblockfair}}

Before proving the underranking and group fairness guarantess of \Cref{alg:ranking}, we prove the following useful lemma,
\begin{lemma}
	\label{lem:blockresults}
	In a block of size $\floor{\epsilon k/2}$, the following always hold,
	\begin{enumerate}
		\item If the block contains any empty ranks, then there is at least one group that can be assinged to this rank without violating the upper bound constraint.
		\item If $\forall l\in[\ell]$ there are $\ceil{\beta_l \floor{\epsilon k/2}}$ items available from the group $l$, then we can always assign $\ceil{\beta_l \floor{\epsilon k/2}}$ ranks to group $l$, $\forall l \in [\ell]$, such that all the lower bound constraints in the block are satisfied.
		\item For each group $l \in [\ell]$, the upper bound on the number of ranks to be assigned to the group is always greater than the lower bound.
	\end{enumerate}
\end{lemma}
\begin{proof}
	By our choice of parameters, we have,
	\begin{align*}
	\epsilon  \geq \frac{2}{k}\cdot \max\Bigg\{\floor{1+\frac{\ell}{ \asum - 1}}, \floor{1+\frac{\ell}{ 1 - \bsum}}, \max_{l \in [\ell]} \floor{1+\frac{2}{ \alpha_l - \beta_l}} \Bigg\} \geq \frac{2}{k}.
	\end{align*}
	Therefore, the size of each block $\floor{\epsilon k/2}$ is at least 1.
	Moreover,
	\begin{align*}
	&\epsilon \geq \frac{2}{k}\paren{1+\frac{\ell}{\asum - 1}} \\
	\implies &\paren{\epsilon k/2 - 1} \geq \frac{\ell}{\asum - 1}\\
	\implies &\floor{\epsilon k/2}  > \frac{\ell}{\asum - 1}\\
	\implies &\sum_{l \in [\ell]}\paren{\alpha_l \floor{\epsilon k/2} - 1} > \floor{\epsilon k/2}\\
	\implies &\sum_{l \in [\ell]}\floor{\alpha_l  \floor{\epsilon k/2} } > \floor{\epsilon k/2}. 
	\end{align*}
	Therefore, if the block has an empty rank, there is at least one group that can be assinged to this empty rank without violating the upper bound constraints. Hence, the statement 1 is true. We also have,
	\begin{align*}
	&\epsilon \geq \frac{2}{k}\paren{1+\frac{\ell}{1 - \bsum }} \\
	\implies &\epsilon k/2 - 1 \geq   \frac{\ell}{1 - \bsum }
	\implies \floor{\epsilon k/2} > \frac{\ell}{1 - \bsum }\\
	\implies &\floor{\epsilon k/2} \paren{1-\bsum} > \ell  \numberthis \label{eq:st2}\\
	\implies &\sum_{l\in[\ell]} \paren{\beta_l \floor{\epsilon k/2}+1} < \floor{\epsilon k/2}
	\implies \sum_{l \in [\ell]} \ceil{\beta_l \floor{\epsilon k/2}} < \floor{\epsilon k/2}.
	\end{align*}
	Hence, if $\forall l \in [\ell]$, there are $\ceil{\beta_l \floor{\epsilon k/2}}$ items available from the group $l$, we can assign $\ceil{\beta_l \floor{\epsilon k/2}}$ ranks in the block to group $l$ since the sum of the minimum number of ranks needed to be assigned to satisfy the lower bound constraints is strictly less than the number of empty ranks in the block. Therefore statement 2 is also true. Finally, for all $l \in [\ell]$,
	\begin{align*}
	&\epsilon \geq \frac{2}{k}\paren{1+\frac{2}{\alpha_l - \beta_l}} \\
	\implies &\epsilon k/2 - 1 \geq \frac{2}{\alpha_l - \beta_l} \\
	\implies &\floor{\epsilon k/2} >  \frac{2}{\alpha_l - \beta_l} \\
	\implies &\beta_l \floor{\epsilon k/2} + 1 < \alpha_l \floor{\epsilon k/2} -1 \numberthis \label{eq:st3}\\
	\implies &\beta_l \floor{\epsilon k/2} + 1 < \floor{\alpha_l \floor{\epsilon k/2} }\\
	\implies &\ceil{\beta_l \floor{\epsilon k/2} } < \floor{\alpha_l \floor{\epsilon k/2} }.
	\end{align*} 
	Hence, the statement 3 is true.
\end{proof}

Throughout the results shown below, let $\amax = \max_l \alpha_l$, $\amin = \min_l \alpha_l$, $l_* = \argmin_l \beta_l$,
\begin{lemma}
	\label{lem:ifair}
	The underranking of the ranking output by \Cref{alg:ranking} is
	\begin{align*}
	\gamma = \frac{1}{\min\set{\amin - \frac{1}{\floor{\epsilon k/2}}, \paren{1-\sum_{l \neq l_*} \beta_l} - \frac{\ell - 1}{\floor{\epsilon k/2}}}}
	.\end{align*}
\end{lemma}
\begin{proof}
	Fix an item having true rank $j \in [N]$.
	
	\textbf{Case 1:} $\floor{\amin \floor{\epsilon k/2}} \leq \floor{\epsilon k/2} - \sum_{l \neq l_*} \ceil{\beta_l\floor{\epsilon k/2}}$.
	
	Let $\eta = \floor{\amin \floor{\epsilon k/2}} $. 
	At the end of step~\ref{step:move-down}, its rank is
	\begin{align*}
	&\paren{\ceil{\frac{j}{\eta } } -1} \floor{\frac{\epsilon k}{2}} + \paren{j - \paren{\ceil{\frac{j}{\eta}} -1}\eta}  \\
	&\leq \frac{j}{\eta } \paren{\floor{\frac{\epsilon k}{2}} - \eta} + j 
	= j \frac{\floor{\frac{\epsilon k}{2}}}{\eta }\\
	&= j \frac{\floor{\epsilon k/2}}{\floor{\amin \floor{\epsilon k/2}}} \numberthis \label{eq:underrankingcase1}\\
	&< \frac{j}{\frac{\amin \floor{\epsilon k/2}-1}{\floor{\epsilon k/2}}} 
	= \frac{j}{\amin - \frac{1}{\floor{\epsilon k/2}}}
	.\end{align*}
	From \Cref{eq:st3} we have,
	\begin{align*}
	& \alpha_l \floor{\epsilon k/2} -1 > \beta_l \floor{\epsilon k/2} + 1\\
	\implies & \alpha_l - \frac{1}{ \floor{\epsilon k/2}} > \beta_l  + \frac{1}{\floor{\epsilon k/2}} > 0.
	\end{align*}
	Since this holds true even for the group corresponding to $\amin$, the underranking is positive.
	
	\textbf{Case 2:} $\floor{\amin \floor{\epsilon k/2}}> \floor{\epsilon k/2} - \sum_{l \neq l_*} \ceil{\beta_l\floor{\epsilon k/2}}$.
	
	Let $\eta = \floor{\epsilon k/2} - \sum_{l \neq l_*} \ceil{\beta_l\floor{\epsilon k/2}}$. 
	At the end of step~\ref{step:move-down}, similar to case 1, its rank is
	\begin{align*}
	&\paren{\ceil{\frac{j}{\eta } } -1} \floor{\frac{\epsilon k}{2}} + \paren{j - \paren{\ceil{\frac{j}{\eta}} -1}\eta}  \\
	&\leq \frac{j}{\eta } \paren{\floor{\frac{\epsilon k}{2}} - \eta} + j 
	= j \frac{\floor{\frac{\epsilon k}{2}}}{\eta }\\
	&= j \frac{\floor{\frac{\epsilon k}{2}}}{\floor{\frac{\epsilon k}{2}}- \sum_{l \neq l_*}^{} \ceil{\beta_l\floor{\frac{\epsilon k}{2}} } } \numberthis \label{eq:underrankingcase2}\\
	&\leq j \frac{\floor{\frac{\epsilon k}{2}}}{\floor{\frac{\epsilon k}{2}}- \sum_{l \neq l_*}^{} \paren{\beta_l\floor{\frac{\epsilon k}{2}}  + 1 } }\\
	&= \frac{j}{\paren{1- \sum_{l \neq l_*}^{} \beta_l} - \frac{\ell - 1}{\floor{\epsilon k/2}} }
	.\end{align*}
	From \Cref{eq:st2} we have,
	\begin{align*}
	&\floor{\epsilon k/2} \paren{1-\bsum} > \ell \\
	\implies & \paren{1-\bsum} - \frac{\ell - 1}{\floor{\epsilon k/2}} > \frac{1}{\floor{\epsilon k/2}}\\
	\implies & \paren{1-\sum_{l \neq l_*}^{} \beta_l} - \frac{\ell - 1}{\floor{\epsilon k/2}} > \beta_{l_*} + \frac{1}{\floor{\epsilon k/2}} > 0.
	\end{align*}
	Therefore, the underranking in this case is also positive.
\end{proof}

\begin{lemma}
	\label{lem:nonemptyranks}
	At the end of step~\ref{step:move-up},
	none of the top $\floor{\frac{n}{\amax}}  - \floor{\epsilon k/2} $ ranks will be empty.
\end{lemma}

\begin{proof}
	Consider step~\ref{step:move-up} of the algorithm.
	A rank $j$ will be left unassigned if either 
	(i) fairness constraints for each group are satisfied and assigning it to any item will only violate the fairness constraints, or 
	(ii) there is no item ranked higher than $j$ that can be assigned the rank $j$ and still satisfies the fairness constraints.
	
	Let $i$ be the block the rank $j$ belongs to, $i = \ceil{j/\floor{\epsilon k/2}}$.
	From statement 2 in \Cref{lem:blockresults} we know that, if the ranks in block $i$ are assigned such that $\forall l\in[\ell]$, $\ceil{\beta_l \floor{\epsilon k/2}}$ ranks are assigned to group $l$, there could still be empty ranks in the block $i$. Also from statement 3 in \Cref{lem:blockresults} we know that for each group $l$, the number of ranks to be assigned to in order to satisfy the upper bound constraints is always greater than the number of ranks to be assigned to in order to satisfy the lower bound constraints. Therefore we can still add items to the block $i$ until the upper bound constraints are not violated. Now, from statement 1 of \Cref{lem:blockresults} we also know that if there are empty ranks in the block, then there is at least one group which can be assinged to the empty rank and not violate the upper bound constraints. Hence, case (i) can not happen.
	
	We know that $\floor{\amax\floor{\epsilon k/2}} \geq \floor{\alpha_l \floor{\epsilon k/2}}$ by definition and $\floor{\amax\floor{\epsilon k/2}} \geq \ceil{\beta_l \floor{\epsilon k/2}}$ by \Cref{lem:blockresults} for any $l$.
	Since there are at least $n$ items from each group, we have that as long as $i$ satisfies $i \floor{\amax\floor{\epsilon k/2}} \leq n $, there will be at least one item available from each group to move into an empty rank in the top $i$ blocks without violating the fairness constraints for any of the top $i$ blocks. Thus, the top $i$ blocks will be filled at the end of step~\ref{step:move-up}. Therefore, the number of ranks filled is at least 
	
	\begin{align*}
	i \floor{\epsilon k/2} &= \floor{\frac{n}{\floor{\amax  \floor{\epsilon k/2} } } } \floor{\epsilon k/2} \\
	&> \paren{\frac{n}{\floor{\amax  \floor{\epsilon k/2 } } } - 1} \floor{\epsilon k/2} \\
	& \geq \paren{\frac{n}{\amax  \floor{\epsilon k/2} }-1} \floor{\epsilon k/2}  \\
	&=  \frac{n}{\amax } - \floor{\epsilon k/2} \geq \floor{\frac{n}{\amax}} -  \floor{\epsilon k/2} .
	\end{align*}	
	
	Therefore, case (ii) will not happen for the top $\floor{\frac{n}{\amax}}  - \floor{\epsilon k/2}$ ranks.
	Thus, at the end of step~\ref{step:move-up},
	none of the top $\floor{\frac{n}{\amax}}  - \floor{\epsilon k/2} $ ranks will be empty.
\end{proof}

\begin{lemma}
	\label{lem:blocks_ubounds}
	At the end of step~\ref{step:move-up}, in each of the top $M/\floor{\epsilon k/2}$ blocks of size $\floor{\epsilon k/2}$, for each group $l \in [\ell]$ we have that the number of items from that group is at most $\floor{\alpha_l \floor{\epsilon k/2}}$, where $M$ is the length of the intermediate ranking as in \Cref{alg:ranking}.
\end{lemma}

\begin{proof}
	For any block $i$, we observe that at the end of step~\ref{step:move-down}, block $i$ of size $\floor{\epsilon k/2}$ 
	has at most \\
	$\min\{\floor{\amin \floor{\epsilon k/2}}, \floor{\epsilon k/2} - \sum_{l \neq l_*} \ceil{\beta_l\floor{\epsilon k/2}}\}$ non-empty ranks and therefore 
	has at most $\floor{\amin \floor{\epsilon k/2}}$ items from any particular group.
	Step~\ref{step:condition} ensures that when the algorithm terminates, each block 
	has at most $\floor{\alpha_l \floor{\epsilon k/2}}$ items from group $l$ for all the groups.
	Since the length of the ranking is $M$, and each block is of size $\floor{\epsilon k/2}$, the statement follows. 
\end{proof}

\begin{lemma}
	\label{lem:blocks_lbounds}
	At the end of step~\ref{step:move-up}, each of the top $\floor{\frac{n}{\floor{\amax  \floor{\epsilon k/2} } } } $ blocks have at least $\ceil{\beta_l \floor{\epsilon k/2}}$ items from group $l$, forall $l \in [\ell]$.
\end{lemma}

\begin{proof}
	
	For any block $i$, we observe that at the end of step~\ref{step:move-down}, block $i$ of size $\floor{\epsilon k/2}$ 
	has exactly\\
	$\min\{\floor{\amin \floor{\epsilon k/2}}, \floor{\epsilon k/2} - \sum_{l \neq l_*} \ceil{\beta_l\floor{\epsilon k/2}}\}$ ranks non empty.
	Let $n_l$ be the number of items from group $l$ assigned ranks in block $i$ after step~\ref{step:move-down}. 
	We know that $\floor{\amin \floor{\epsilon k/2}} \geq \ceil{\beta_{l_*} \floor{\epsilon k/2}}$ and from \Cref{lem:blockresults} we also have that $\floor{\epsilon k/2} - \sum_{l \neq l_*} \ceil{\beta_l\floor{\epsilon k/2}} \geq \ceil{\beta_{l_*}\floor{\epsilon k/2}}$.
	Therefore, 
	\begin{align*}
	\sum_l n_l &= \min\{\floor{\amin \floor{\epsilon k/2}}, \floor{\epsilon k/2} - \sum_{l \neq l_*} \ceil{\beta_l\floor{\epsilon k/2}}\} \\
	&\geq  \ceil{\beta_{l_*}\floor{\epsilon k/2}}
	.\end{align*}
	
	To satisfy all the lower bounds, there have to be at least $\sum_l \paren{\ceil{\beta_l \floor{\epsilon k/2}} - n_l}$ empty ranks in block $i$.
	
	The number of empty ranks in block $i$ after step~\ref{step:move-down} are,
	\begin{align*}
	\floor{\epsilon k/2} - &\min\{\floor{\amin \floor{\epsilon k/2} }, \floor{\epsilon k/2} - \sum_{l \neq l_*} \ceil{\beta_l\floor{\epsilon k/2}} \} \\
	&\geq \floor{\epsilon k/2} - \paren{ \floor{\epsilon k/2} - \sum_{l \neq l_*} \ceil{\beta_l\floor{\epsilon k/2}}}\\
	&= \sum_{l \neq l_*} \ceil{\beta_l\floor{\epsilon k/2}}\\
	&= \sum_{l} \ceil{\beta_l\floor{\epsilon k/2}} - \ceil{\beta_{l_*}\floor{\epsilon k/2}} \\
	&\geq  \sum_{l} \ceil{\beta_l\floor{\epsilon k/2}}  - \sum_l n_l\\
	&= \sum_{l} \paren{ \ceil{\beta_l\floor{\epsilon k/2}} - n_l}
	.\end{align*}

	Therefore, after Step~\ref{step:move-down}, there are enough empty ranks left in the block to satisfy the lower bounds of all the groups.
	
	Step~\ref{step:condition} ensures that as long as the items are available, we always satisfy these lower bounds. Since from \Cref{lem:nonemptyranks} the top $\floor{\frac{n}{\floor{\amax  \floor{\epsilon k/2} } } } $ blocks have no empty ranks, there are at least $\ceil{\beta_l \floor{\epsilon k/2}}$ items from group $l$ in each of these blocks.
\end{proof}

\begin{lemma}
	\label{lem:gfair_bounds}
	The ranking output by \Cref{alg:ranking} satisfies $\paren{(1+\epsilon)\boldsymbol{\alpha}, (1-\epsilon)\boldsymbol{\beta}, k}$ group fairness in every $k$ consecutive ranks in the top $\floor{\frac{n}{\amax}}  - \floor{\epsilon k/2} $ ranks.
\end{lemma}

\begin{proof}
	\Cref{lem:nonemptyranks} shows that none of the top $\floor{\frac{n}{\amax}}  - \floor{\epsilon k/2} $  ranks will
	be empty at the end of step~\ref{step:move-up}; therefore, these ranks will remain unchanged in the steps after
	step~\ref{step:fill-tail}.
	
	Consider any $k$ consecutive ranks $j, \ldots, j+k-1$. Let $i_1 \defeq \ceil{j/\floor{\epsilon k/2}}$ and $i_2 \defeq \ceil{(j+k-1)/\floor{\epsilon k/2}}$. By construction, the blocks $i_1 + 1, \ldots, i_2 - 1$ are fully contained in the ranks $\set{j,j+1, \ldots, j+k-1}$. For any $l \in [\ell]$, the number of items from group $l$ in ranks $j$ to $j+k-1$ is at most the number of items from group $l$ in blocks $i_1$ to $i_2$. Using \Cref{lem:blocks_ubounds} we get that this is at most
	\begin{gather*}
	\floor{\alpha_l k }+ 2 \floor{\alpha_l\floor{\epsilon k/2}} \leq \alpha_l(1+ \epsilon) k. 
	\end{gather*}
	We note that this bound also holds for cases when $i_2 = i_1 + 1$ or $i_2 = i_1$.
	
	Let $m$ be the number of blocks fully contained in these $k$ ranks. Then $k = m\floor{\epsilon k/2} + (x+y)$ where $0 \leq x,y < \floor{\epsilon k/2}$. For any $l \in [\ell]$, the number of items from group $l$ in $k$ consecutive ranks is at least the number of items from group $l$ in $m$ blocks. Using \Cref{lem:blocks_lbounds} we get that this is at least
	\begin{gather*}
	m\ceil{\beta_l \floor{\epsilon k/2}}
	= \paren{\frac{k-(x+y)}{\floor{\epsilon k/2}} }\ceil{\beta_l \floor{\epsilon k/2}}\\
	\geq \paren{\frac{k-(x+y)}{\floor{\epsilon k/2}} }\beta_l \floor{\epsilon k/2}
	= \beta_l k - (x+y)\beta_l \\
	> \beta_l k - 2\beta_l \floor{\epsilon k/2} 
	> \beta_l (1-\epsilon)k.
	\end{gather*}
\end{proof}

\begin{proof}[Proof of \Cref{thm:igfair}]
	Follows from the choice of $\epsilon$ and from \Cref{lem:ifair}, \Cref{lem:nonemptyranks}, \Cref{lem:gfair_bounds}.
\end{proof}

\begin{proof}[Proof of \Cref{thm:igblockfair}]
	We use \Cref{alg:ranking} with $\epsilon := 2$.
	Now, the $i$th ``block'' is of size $\floor{\frac{\epsilon k}{2}} = k$. 
	
	Fix an item $j \in [N]$ in the true ranking.
	If $\amin k \leq k - \sum_{l \neq l_*} \beta_l k$, from \Cref{eq:underrankingcase1} we have that its final rank will be at most 
	\[ \frac{j\floor{\epsilon k/2}}{\floor{\amin\floor{\epsilon k/2}}} = \frac{j}{\amin} . \]
	Here, the equality follows from our choice of $\epsilon = 2$.
	
	Otherwise, from \Cref{eq:underrankingcase2} we have that its final rank will be at most
	\begin{align*}
	&j \frac{\floor{\frac{\epsilon k}{2}}}{\floor{\frac{\epsilon k}{2}}- \sum_{l \neq l_*}^{} \ceil{\beta_l\floor{\frac{\epsilon k}{2}} } } \\
	&= \frac{jk}{k - \sum_{l \neq l_*} \ceil{\beta_l k}} & (\text{substituting } \epsilon = 2)\\
	&= \frac{j}{1 - \sum_{l \neq l_*} \frac{\ceil{\beta_l k}}{k}} = \frac{j}{1-\sum_{l \neq l_*}\beta_l}.
	\end{align*}
	
	Hence, the ranking output by \Cref{alg:ranking} with $\epsilon=2$ satisfies $\frac{1}{\min\set{\amin,1 - \sum_{l \neq l_*} \beta_l} }$ underranking.

	\Cref{lem:blocks_ubounds} shows that at the end of step~\ref{step:fill-tail}, each block has at 
	most $\floor{\alpha_l k}$ items from group $l$.
	Then we have, 
	\[
	\sum_{l \in [\ell]} \floor{\alpha_l k} = k \paren{\asum} \geq k.
	\]
	Therefore, if a block contains  $\floor{\alpha _l k}$ items from group $l \in [\ell]$, it can not have any empty ranks. 
	Consequently, as long as $i\floor{\amax k}$ items fom each group are available, blocks 1 to $i$ will not contain empty ranks and also have required number of items to satisfy the lower bounds since $\alpha_l k \geq \beta_l k$.
	Since there are at least $n$ items from each group, no rank in the top $i$ blocks will be empty, where $i$ satisfies $i \floor{\amax k} \leq n$.
	Hence, for $i\in \Z^+$ such that $i \leq \frac{n}{\floor{\amax k}} = \frac{n}{\amax k}$, blocks $1$ to $i$ each of size $k$ contain at most $\amax k$ items from each group.
	That is, the top $\floor{\frac{n}{\amax k}}$ blocks each of size $k$ satisfy $\paren{\boldsymbol{\alpha}, \boldsymbol{\beta}, k}$  group fairness.
\end{proof}

\section{Experimental Validation}
\label{sec:experiments}

In this section, we give empirical observations about three broad questions 
\begin{enumerate}
	\item Is there a trade-off between underranking and group fairness in the real-world datasets? 
	\item How effective is underranking in choosing a group-fair ranking?
	\item Does ALG achieve best trade-off between group fairness and underranking?
\end{enumerate}

\subsection{Datasets} 
We experiment on two real-world datasets.
(1) \emph{German Credit Risk} dataset consists of credit risk scoring of 1000 adult German residents \cite{german_credit} along with their demographic information such as personal status, gender, age, etc. as well as financial status such as credit history, property, housing, job etc. 
Schufa scores of these individuals is used to get a global ranking on the dataset similar to \cite{fa_ir_a_fair_top_k_ranking_algorithm}. 
\cite{Castillo2019survey} observed that Schufa scoring is biased against young adults.
Hence, we divide the dataset into protected and non-protected groups based on age.
We consider two such cases (i) \textit{age} $< 25$ as protected group, and (ii) \textit{age} $< 35$ as protected group similar to \cite{fa_ir_a_fair_top_k_ranking_algorithm}.
(2) \emph{COMPAS\footnote{Correctional Offender Management Profiling for Alternative Sanctions} recidivism} dataset consists of violent recidivism assessment of nearly 7000 criminal defendants based on a questionnaire.~\cite{compas_propublica} have analysed this tool and pointed out the biases in the recidivism score against African Americans and females.
We consider ranking based on the recidivism score (individual with the least negative recidivism score is ranked at the top) with (i) \textit{gender} (=female)\footnote{Non-binary genders were not annotated in any of the datasets used in this paper.} and (ii) \textit{race} (=African American) as protected groups similar to \cite{fa_ir_a_fair_top_k_ranking_algorithm}. 
We use the processed subsets of German credit risk and COMPAS recidivism datasets\footnote{https://github.com/DataResponsibly/FairRank/tree/master/datasets}.
The implementation of the algorithm proposed in this paper is also made public\footnote{\href{https://github.com/sruthigorantla/Underranking_and_group_fairness}{Implementation of ALG}}.

\subsection{Baselines}
The baselines considered in this paper are described below,
\begin{enumerate}
	\item \textbf{\cite{ranking_with_fairness_constraints}'s DP algorithm:} In \cite{ranking_with_fairness_constraints}, $W_{ij}$ represents the utility of the item $i$ placed at rank $j$.
	Since we only have the scores (or relavance) of the item when placed at the top rank, we construct $W_{ij}$ using positional discounting as described in the appendix of \cite{ranking_with_fairness_constraints}.
	We first sort the items in the decreasing order of their scores, which gives the true ranking. 
	Wlog, let this ordering also represent the indices of the items, i.e., the item with highest score is indexed as item $1$. 
	Let $y_i$ be the score of item $i$. 
	Then,
	\begin{align*}
	W_{ij}  = \frac{y_i}{\log_2(j+1)}.
	\end{align*}
	Then $W$ satisfies the Monge and monotonicity properties required by the DP algorithm in \cite{ranking_with_fairness_constraints}.
	Let there be $\ell$ groups the items can belong to. 
	Let $P_l$ contain the indices of the items that belong to group $l$, for each $l \in [\ell]$.
	Since we have $N$ items, let $x \in \set{0,1}^{N \times N}$ be a ranking (or assignment) whose $j$-th column contains a one in the $i$-th position if item $i$ is assigned to rank $j$.
	Note that each rank can be assigned to exactly one item and each item is assigned exactly one rank.
	Then, the fairness constraints at every prefix $k' \in [k]$ of the top $k$ ranking are in the form of the following cardinality constraints,
	\begin{align*}
	L_{l, k'} \le \sum_{1 \le k' \le k}\sum_{i\in P_l} x_{ik'} \le U_{l, k'}.
	\end{align*}
	The set of rankings that satisfy the above fairness constraints is represented with $\mathcal{B}$.
	Then the fair ranking problem posed as the following integer program,
	\begin{align*}
	\max_{x \in \mathcal{B}} \sum_{j \in [N]}\sum_{i \in [N]} x_{ij}W_{ij}.
	\end{align*}
	The DP algorithm solves the above integer program exactly.
	In the expeirments with only one protected group represented by $l = 1$ and one non-protected group represented by $l = 2$, we use the lower bounds on the representation of the protected group, $L_{1, k'} = \ceil{pk'}, \forall k' \in [k]$ such that every prefix of the top $k$ ranks has minimum $p$ proportion of the items from the protected group. All other constraints are removed, i.e., $U_{1, k'} = k', U_{2, k'} = k', L_{2, k'} = 0$. In case of experiments with upper bound on the protected group (\Cref{fig:german_25_rev} to \Cref{fig:compas_race_rev_block}), we use $U_{1, k'} = \floor{pk'}, U_{2, k'} = k', L_{1, k'} = 0, L_{2, k'} = 0$, since we want a maximum proportion of $p$ of the protected group in each prefix of the top $k$ ranks. In our experiments, we show the trade-offs between representation and underranking by varying the parameter $p$. We use $p = p_l^* + \delta$ where $p_l^*$ is the proportion of group $l$ in the dataset and $\delta \in \R$.
	\item \textbf{\cite{fa_ir_a_fair_top_k_ranking_algorithm}}: This greedy algorithm solves top-$k$ ranking problem such that the proportion of the protected group stays significantly above the given minimum $p$, in every prefix $k' \in [k]$. Since, in our experiments with just one protected and one non-protected group, we want all the algorithms to achieve a minimum representation of $p^* + \delta$ of the protected group with true representation $p^*$ and small number $\delta \in \R$, we run FA*IR by choosing the parameter $p = p^* + \delta$. Note that FA*IR can not handle the upper bound constraints on the representation of the groups. Hence in \Cref{fig:german_25_rev} to \Cref{fig:compas_race_rev_block} we do not consider comparison with FA*IR.
\end{enumerate}
In ALG and both the baselines, we choose $k = 100$.

\begin{figure*}[t]
	\centering
	\begin{subfigure}[b]{\linewidth}
		\centering
		\includegraphics[scale=0.2]{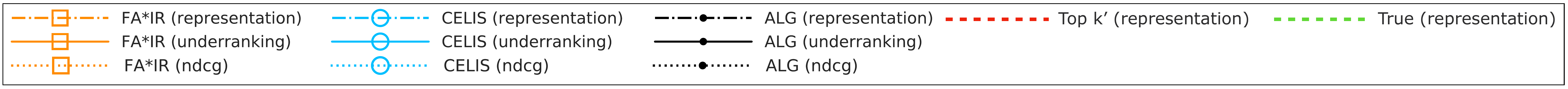} 
	\end{subfigure}
	
	\begin{subfigure}[b]{0.33\linewidth}
		\centering
		\includegraphics[scale=0.149]{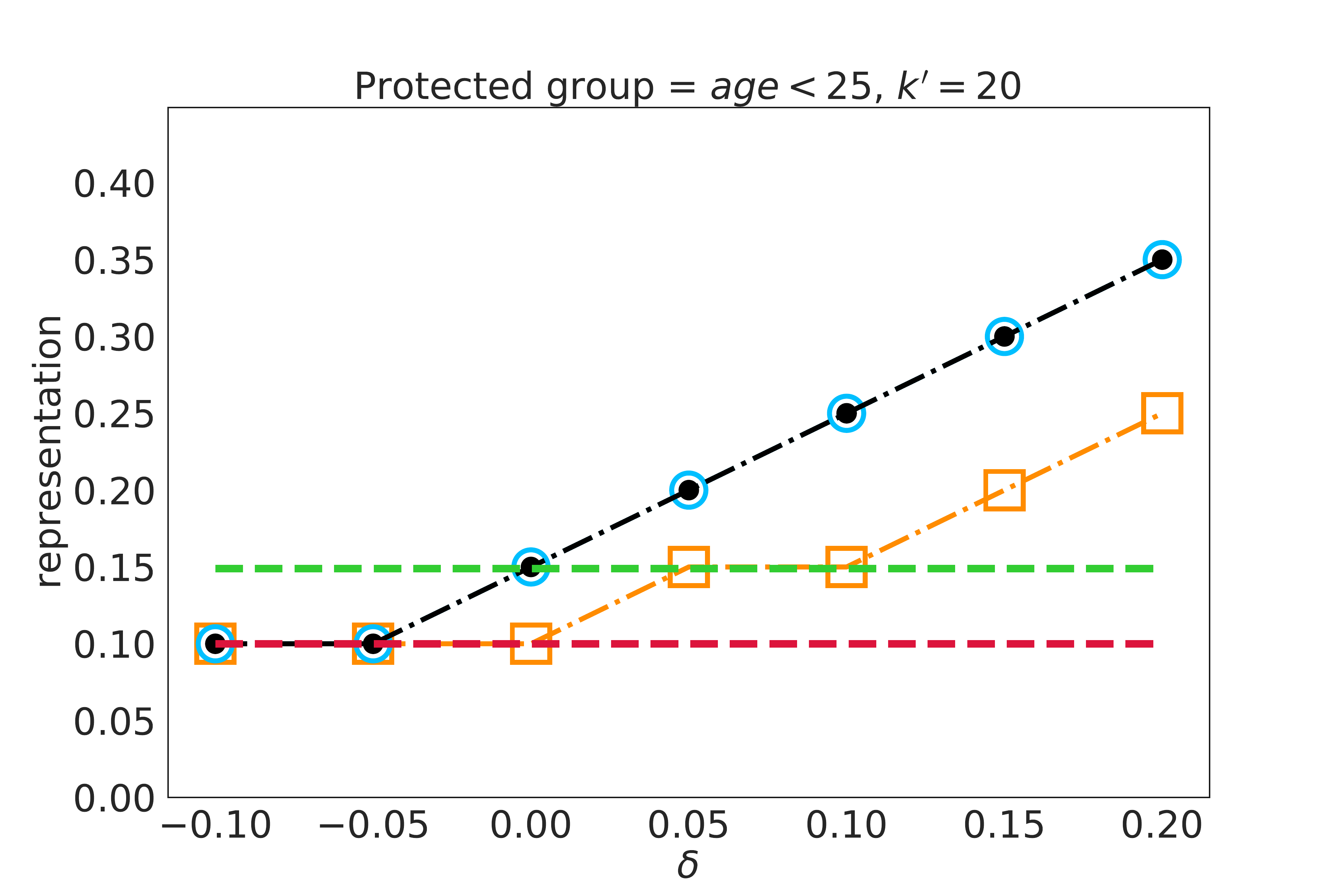} 
		\caption{Representation at top $20$ ranks}
		\label{fig:german_25_a}
	\end{subfigure}
	\begin{subfigure}[b]{0.33\linewidth}
		\centering
		\includegraphics[scale=0.149]{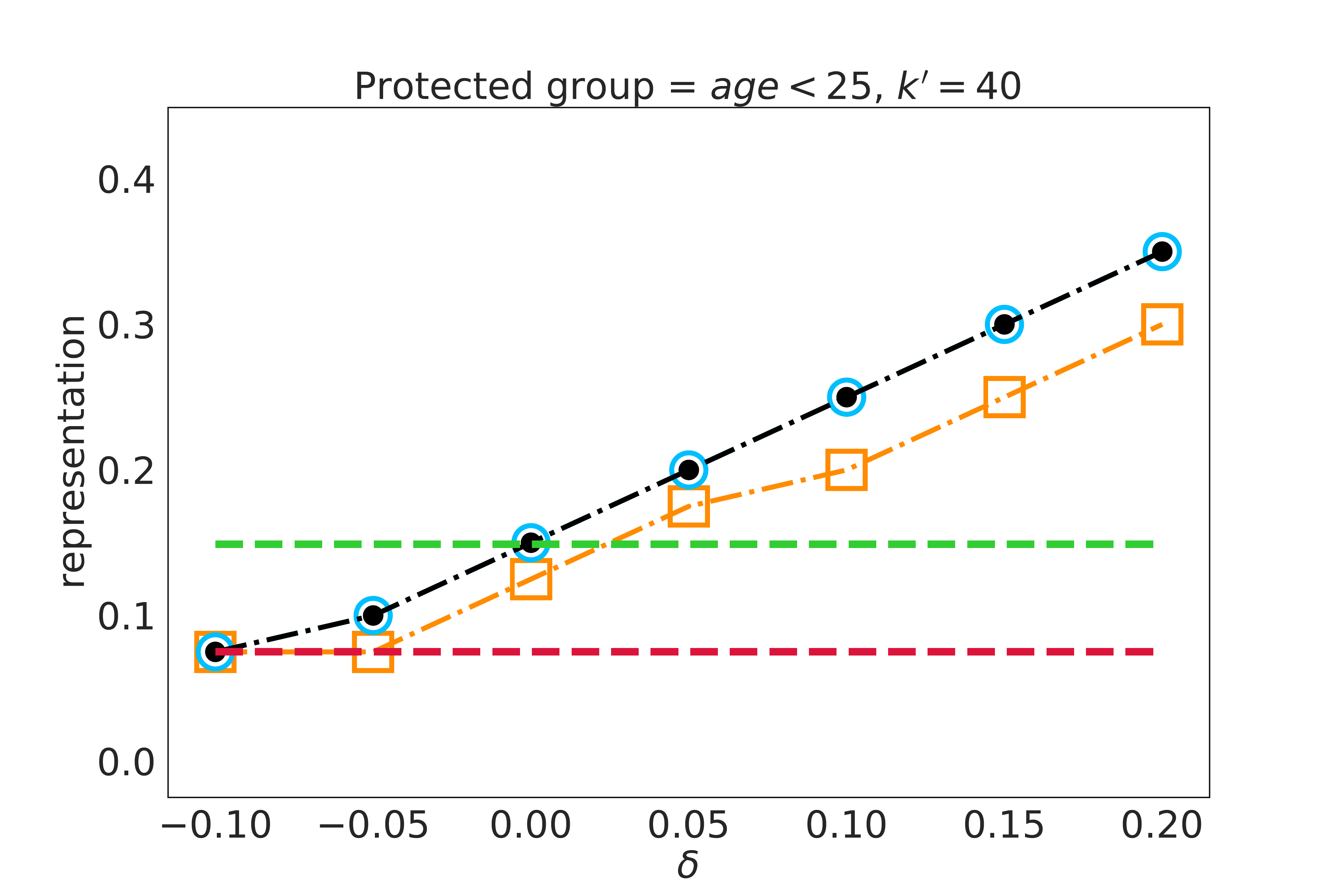} 
		\caption{Representation at top $40$ ranks}
		\label{fig:german_25_b}
	\end{subfigure}
	\begin{subfigure}[b]{0.33\linewidth}
		\centering
		\includegraphics[scale=0.149]{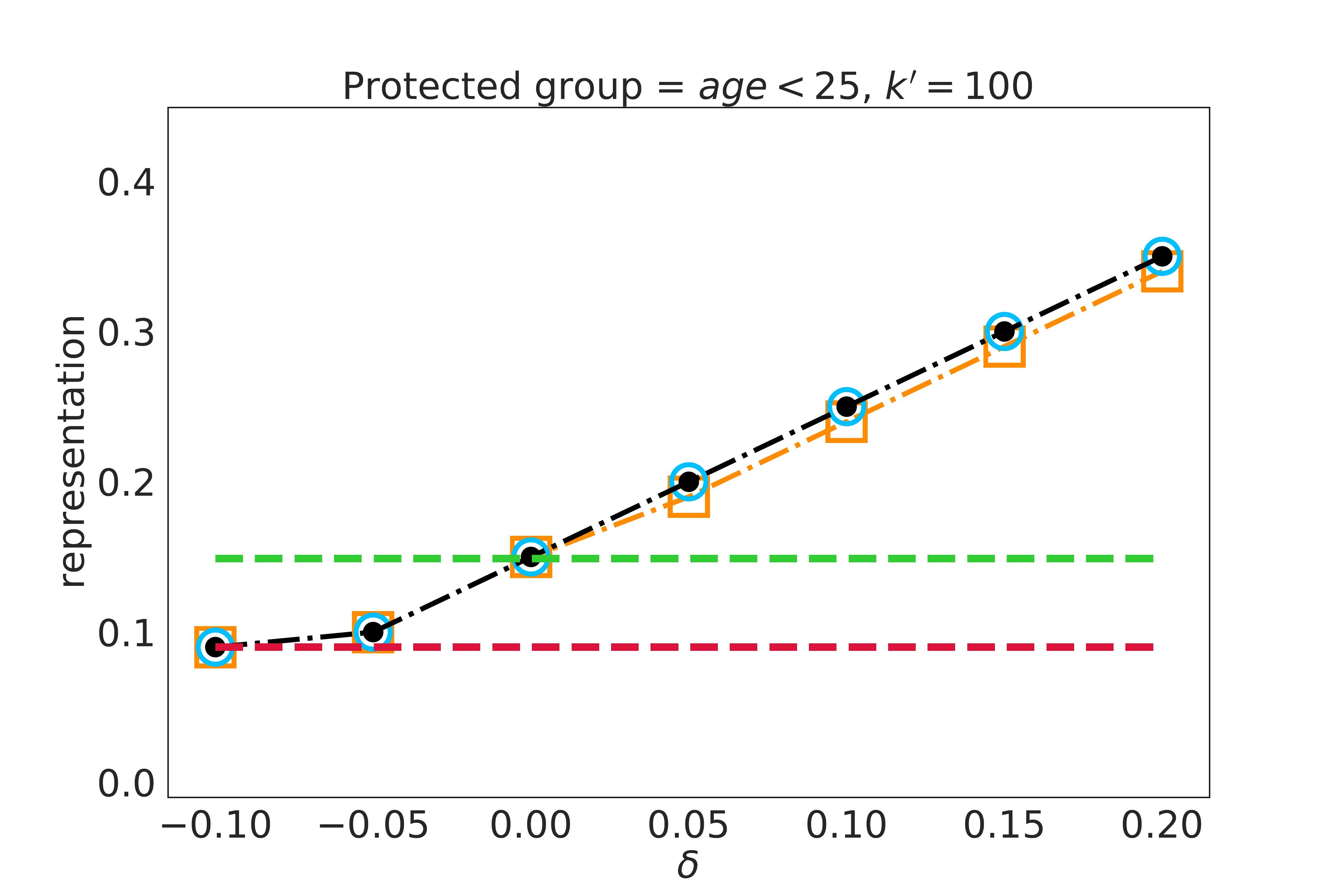} 
		\caption{Representation at top $100$ ranks}
		\label{fig:german_25_c}
	\end{subfigure}
	
	\begin{subfigure}[b]{0.33\linewidth}
		\centering
		\includegraphics[scale=0.149]{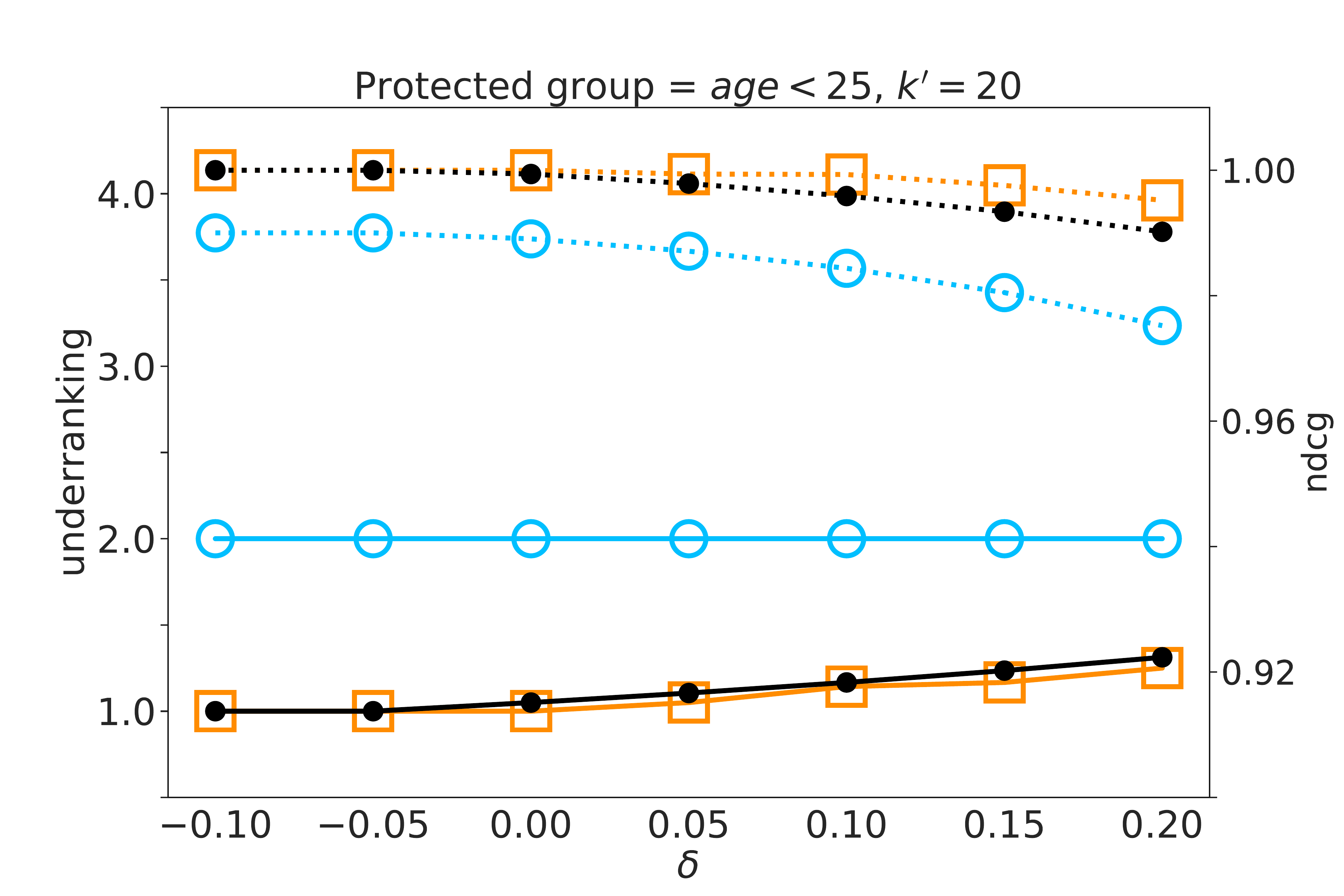} 
		\caption{Underranking, nDCG at top $20$ ranks}
		\label{fig:german_25_d}
	\end{subfigure}
	\begin{subfigure}[b]{0.33\linewidth}
		\centering
		\includegraphics[scale=0.149]{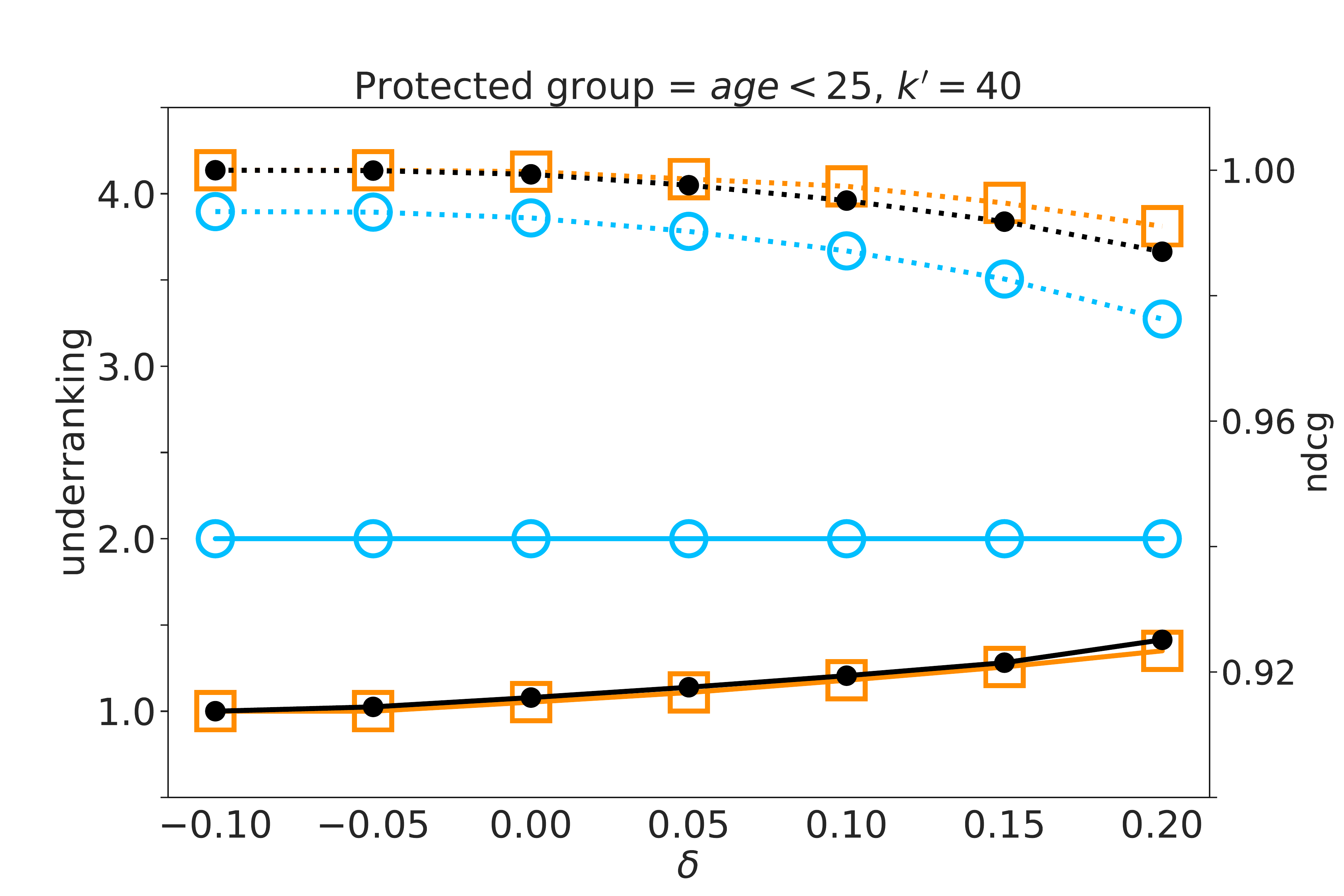} 
		\caption{Underranking, nDCG at top $40$ ranks}
		\label{fig:german_25_e}
	\end{subfigure}
	\begin{subfigure}[b]{0.33\linewidth}
		\centering
		\includegraphics[scale=0.149]{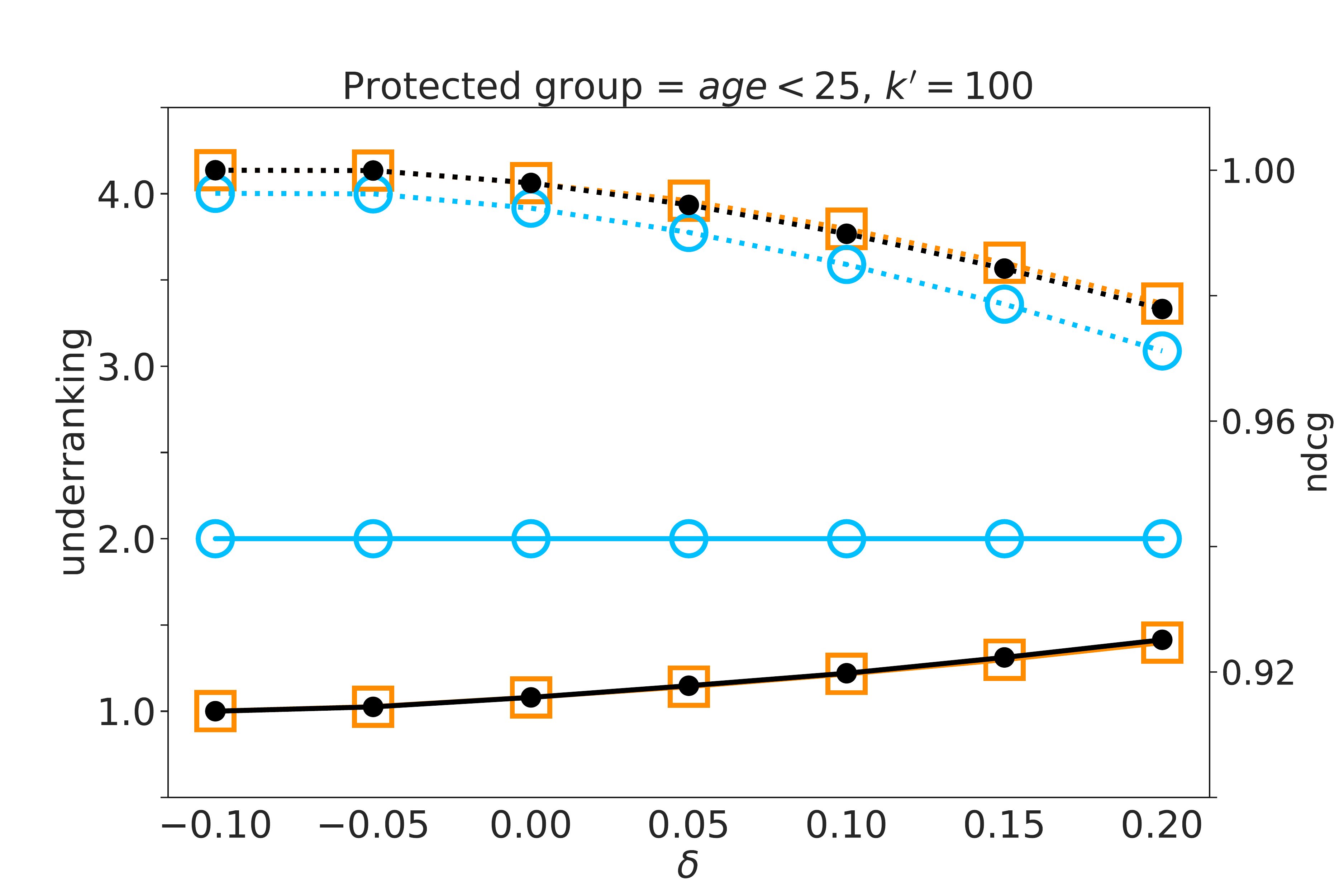} 
		\caption{Underranking, nDCG at top $100$ ranks}
		\label{fig:german_25_f}
	\end{subfigure}
	\caption{Results on the German Credit Risk dataset with \textit{age}$<25$ as the protected group.}
	\label{fig:german_25}
\end{figure*}

\begin{figure}[t]
	
	\begin{subfigure}[b]{\linewidth}
		\centering
		\includegraphics[scale=0.2]{results/legend.pdf} 
	\end{subfigure}
	
	\begin{subfigure}[b]{0.33\linewidth}
		\centering
		\includegraphics[scale=0.149]{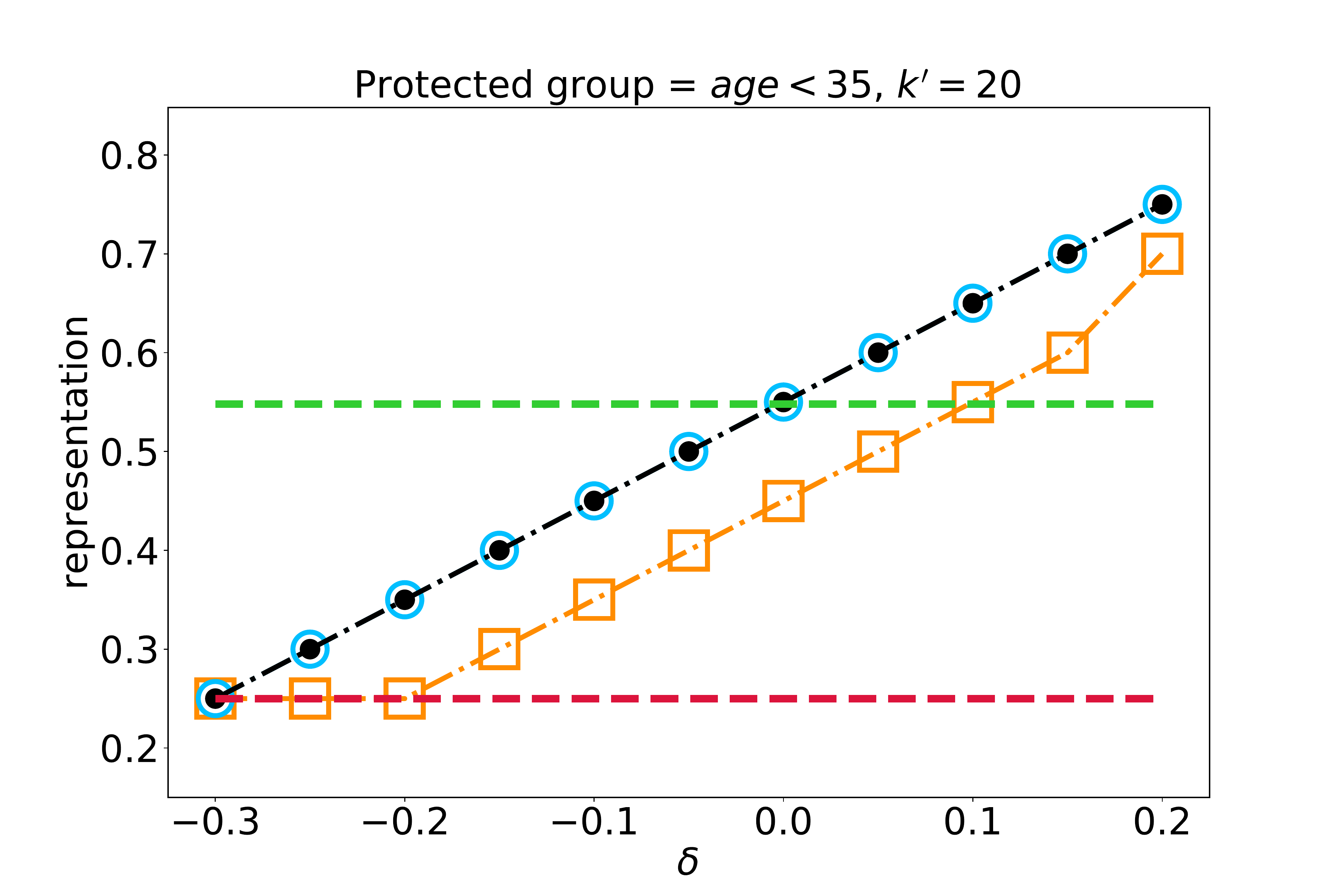} 
		\caption{Representation at top $20$ ranks.}
		\label{fig:german_35_a}
	\end{subfigure}
	\begin{subfigure}[b]{0.33\linewidth}
		\centering
		\includegraphics[scale=0.149]{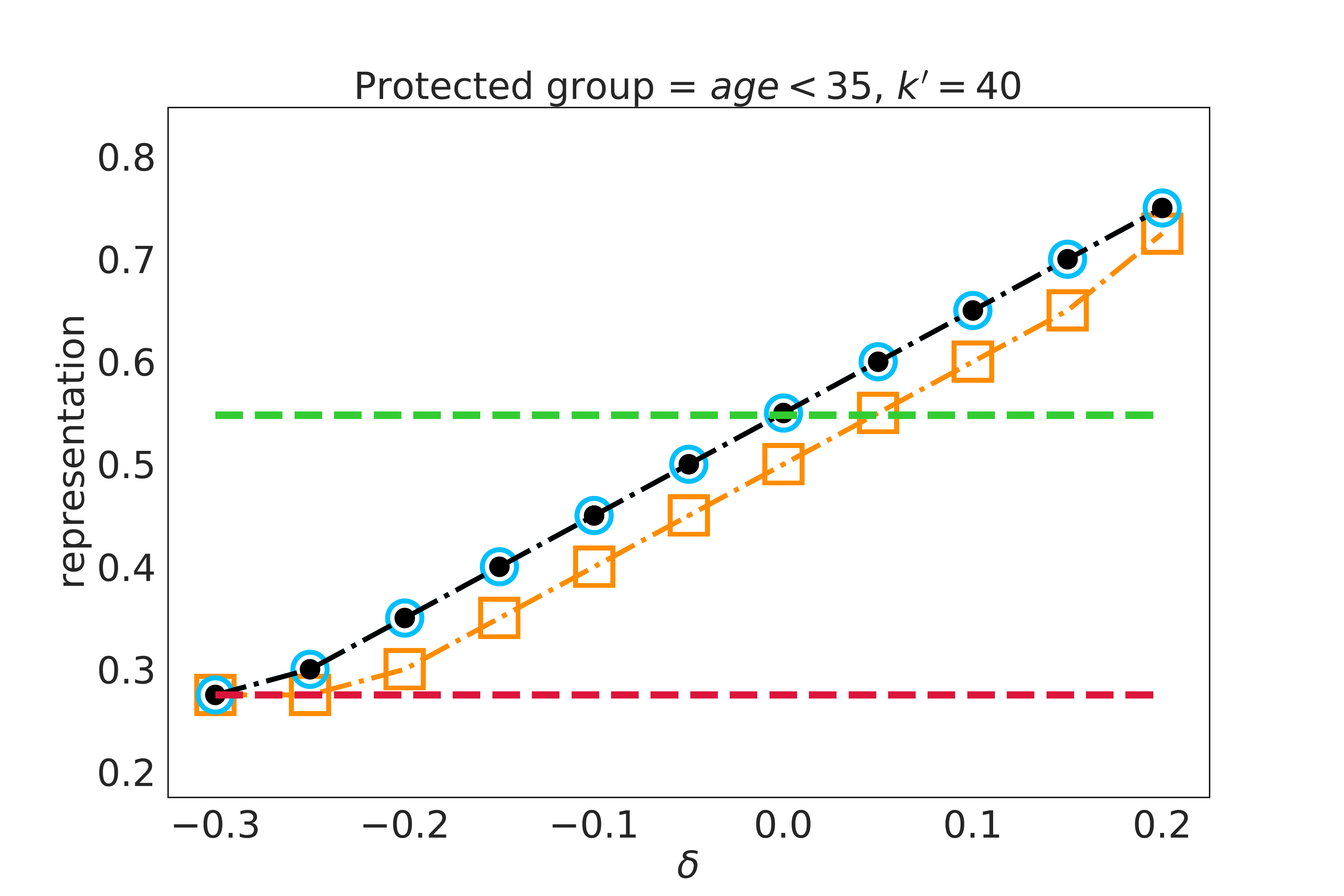} 
		\caption{Representation at top $40$ ranks.}
		\label{fig:german_35_b}
	\end{subfigure}
	\begin{subfigure}[b]{0.33\linewidth}
		\centering
		\includegraphics[scale=0.149]{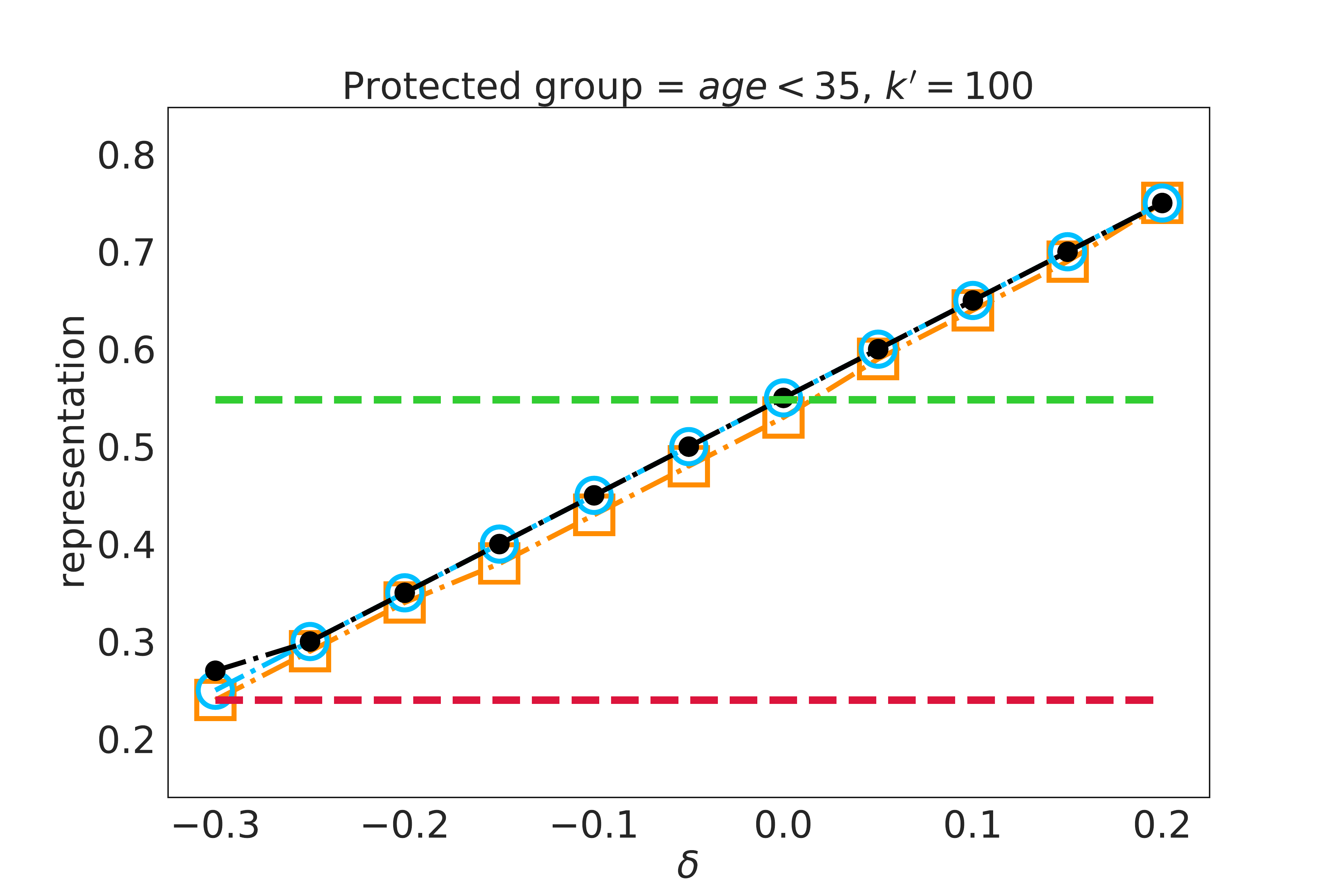} 
		\caption{Representation at top $100$ ranks.}
		\label{fig:german_35_c}
	\end{subfigure}
	\vspace*{-1mm}
	\begin{subfigure}[b]{0.33\linewidth}
		\centering
		\includegraphics[scale=0.149]{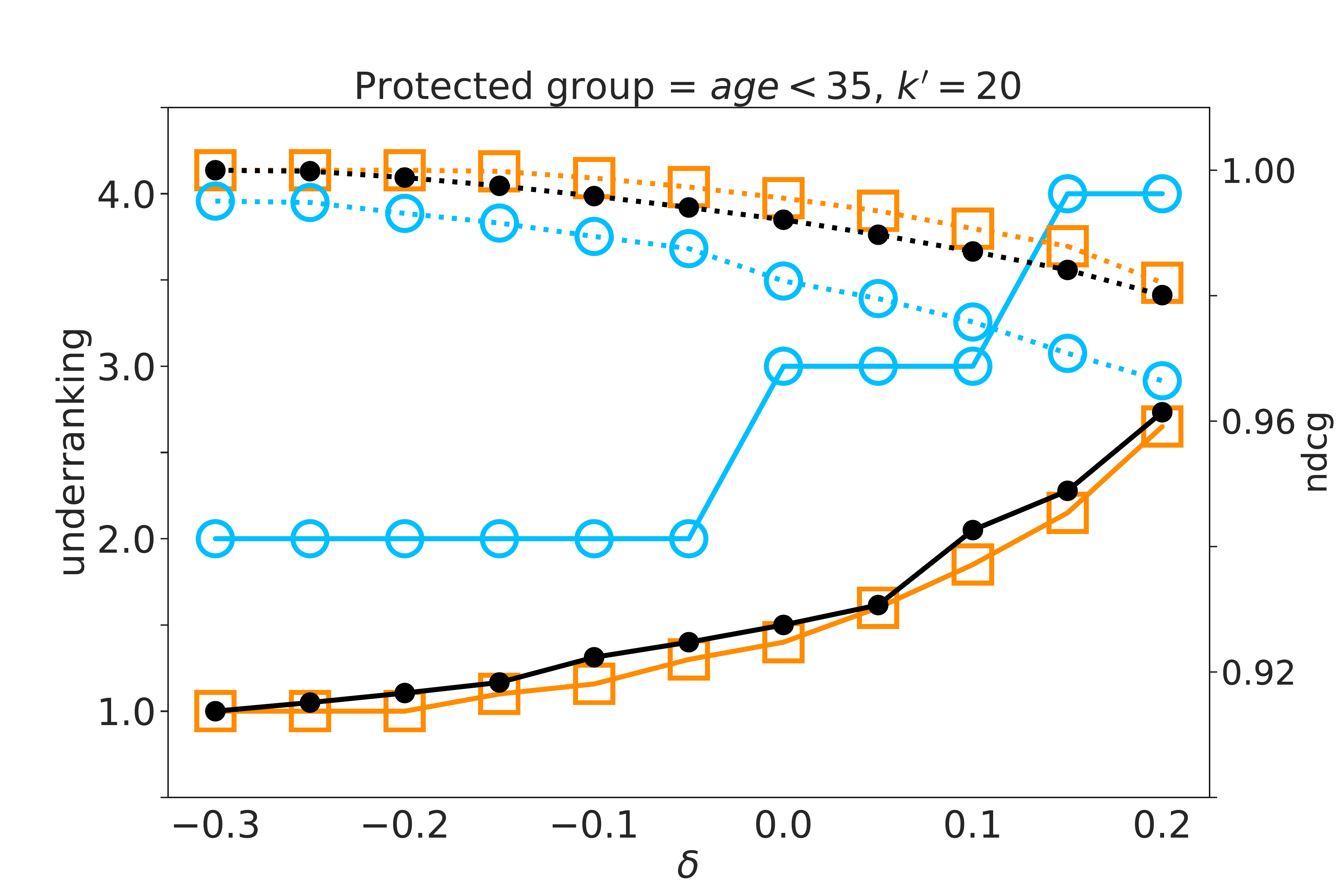} 
		\caption{Underranking, nDCG at top $20$ ranks.}
		\label{fig:german_35_d}
	\end{subfigure}
	\begin{subfigure}[b]{0.33\linewidth}
		\centering
		\includegraphics[scale=0.149]{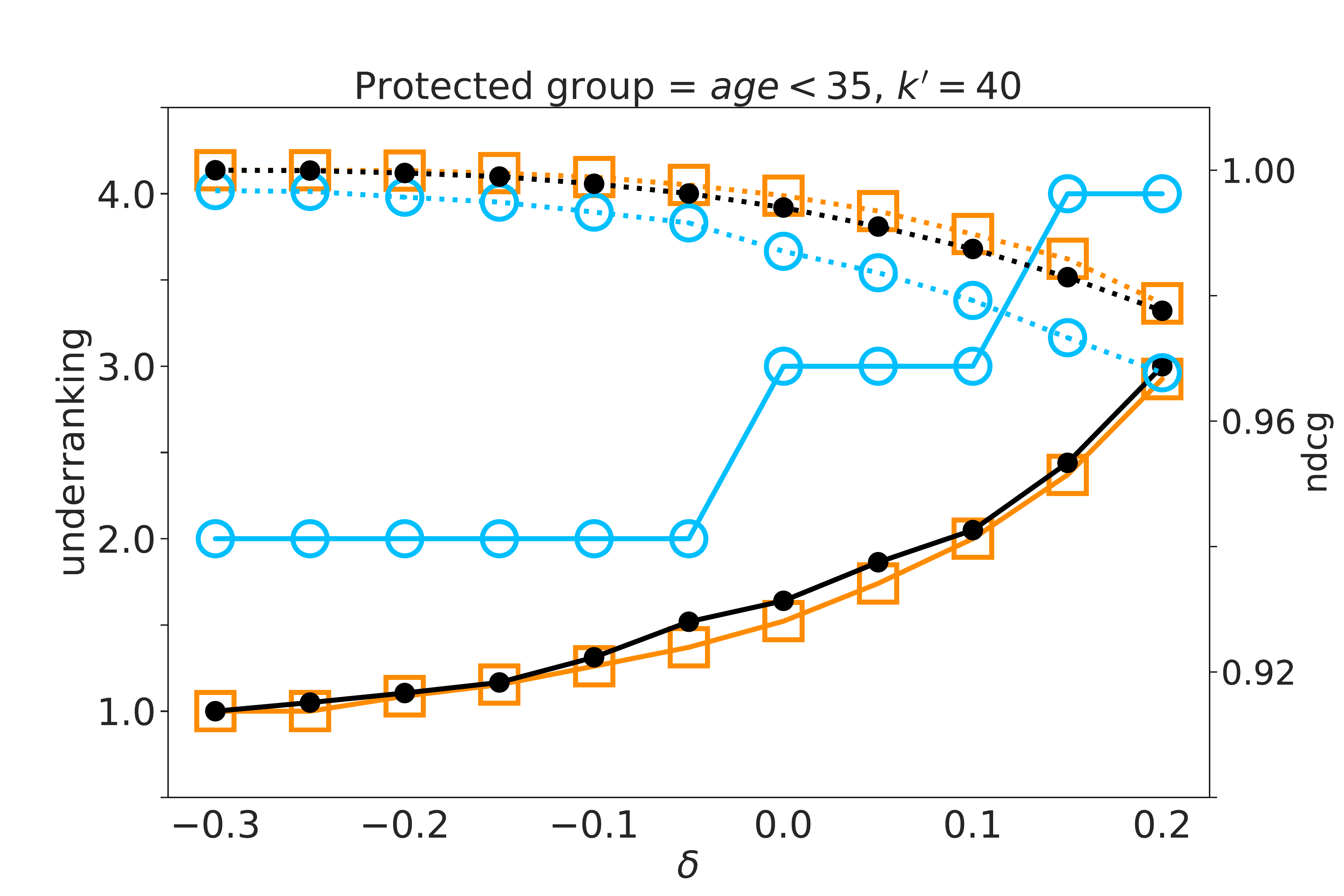} 
		\caption{Underranking, nDCG at top $40$ ranks.}
		\label{fig:german_35_e}
	\end{subfigure}
	\begin{subfigure}[b]{0.33\linewidth}
		\centering
		\includegraphics[scale=0.149]{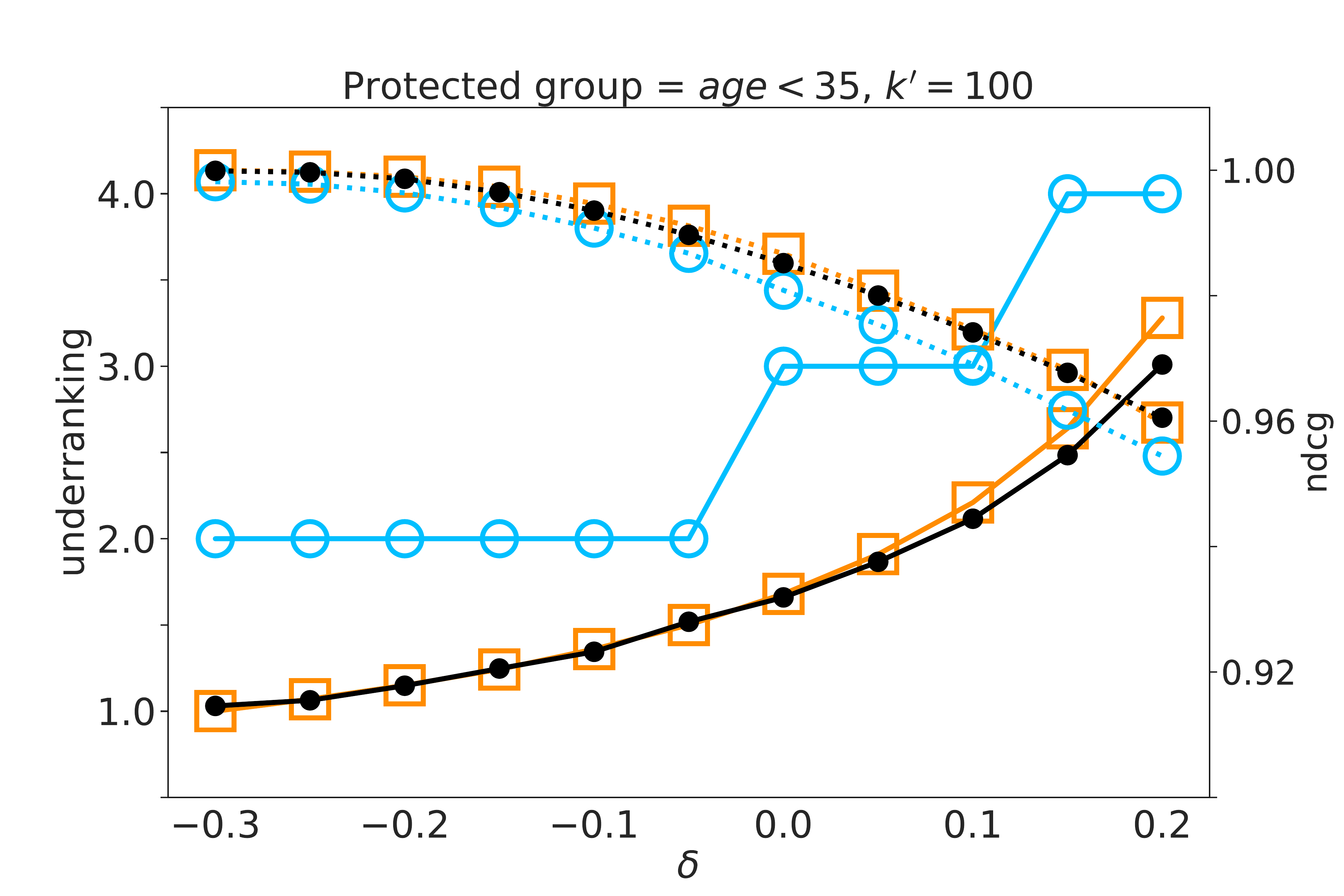} 
		\caption{Underranking, nDCG at top $100$ ranks.}
		\label{fig:german_35_f}
	\end{subfigure}
	\caption{Results on the German Credit Risk dataset with \textit{age}$<35$ as the protected group.}
	\label{fig:german_35}
\end{figure}

\begin{figure*}[t]
	
	\begin{subfigure}[b]{\linewidth}
		\centering
		\includegraphics[scale=0.2]{results/legend.pdf} 
	\end{subfigure}
	
	\begin{subfigure}[b]{0.33\linewidth}
		\centering
		\includegraphics[scale=0.149]{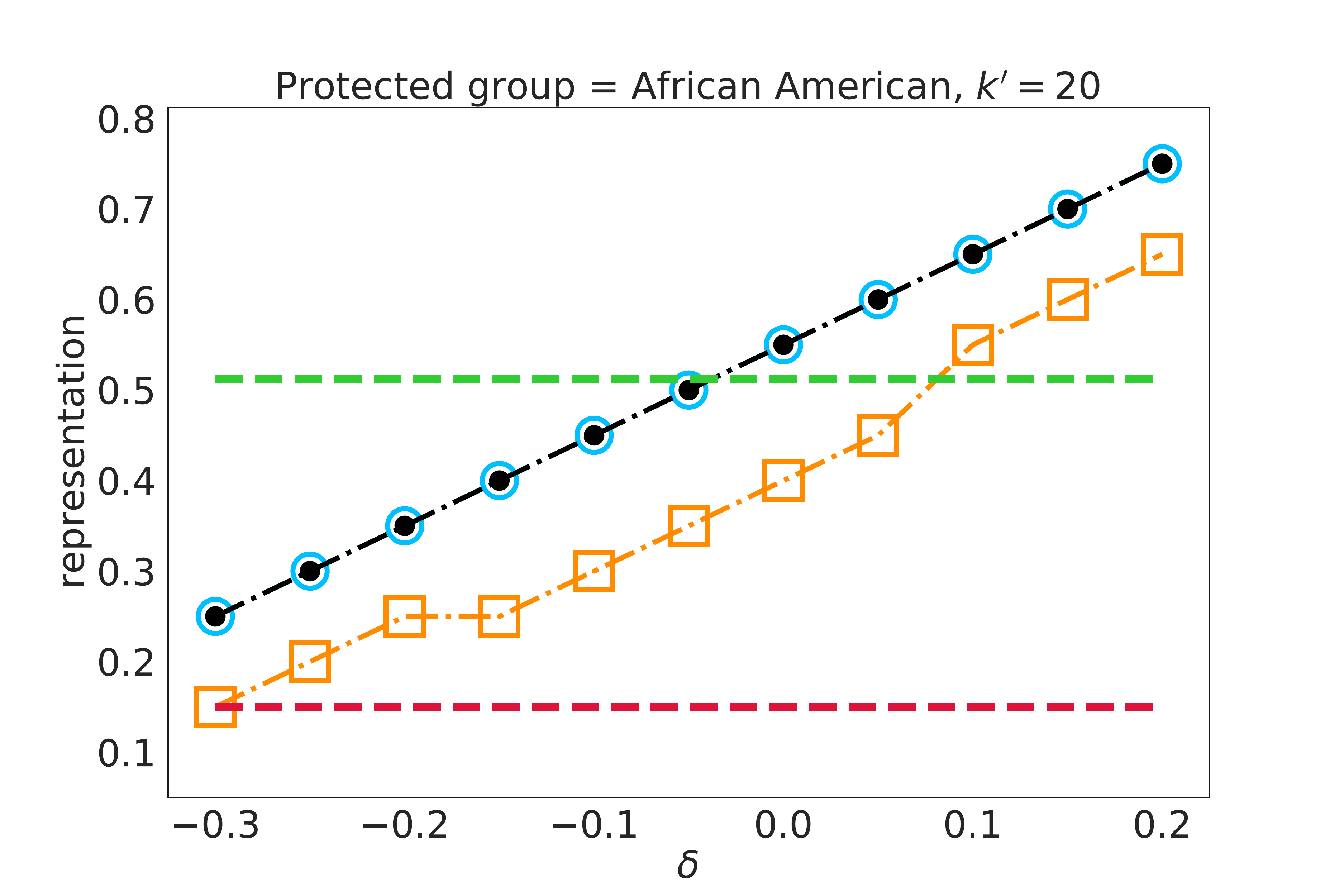} 
		\caption{Representation at top $20$ ranks.}
		\label{fig:compas_race_a}
	\end{subfigure}
	\begin{subfigure}[b]{0.33\linewidth}
		\centering
		\includegraphics[scale=0.149]{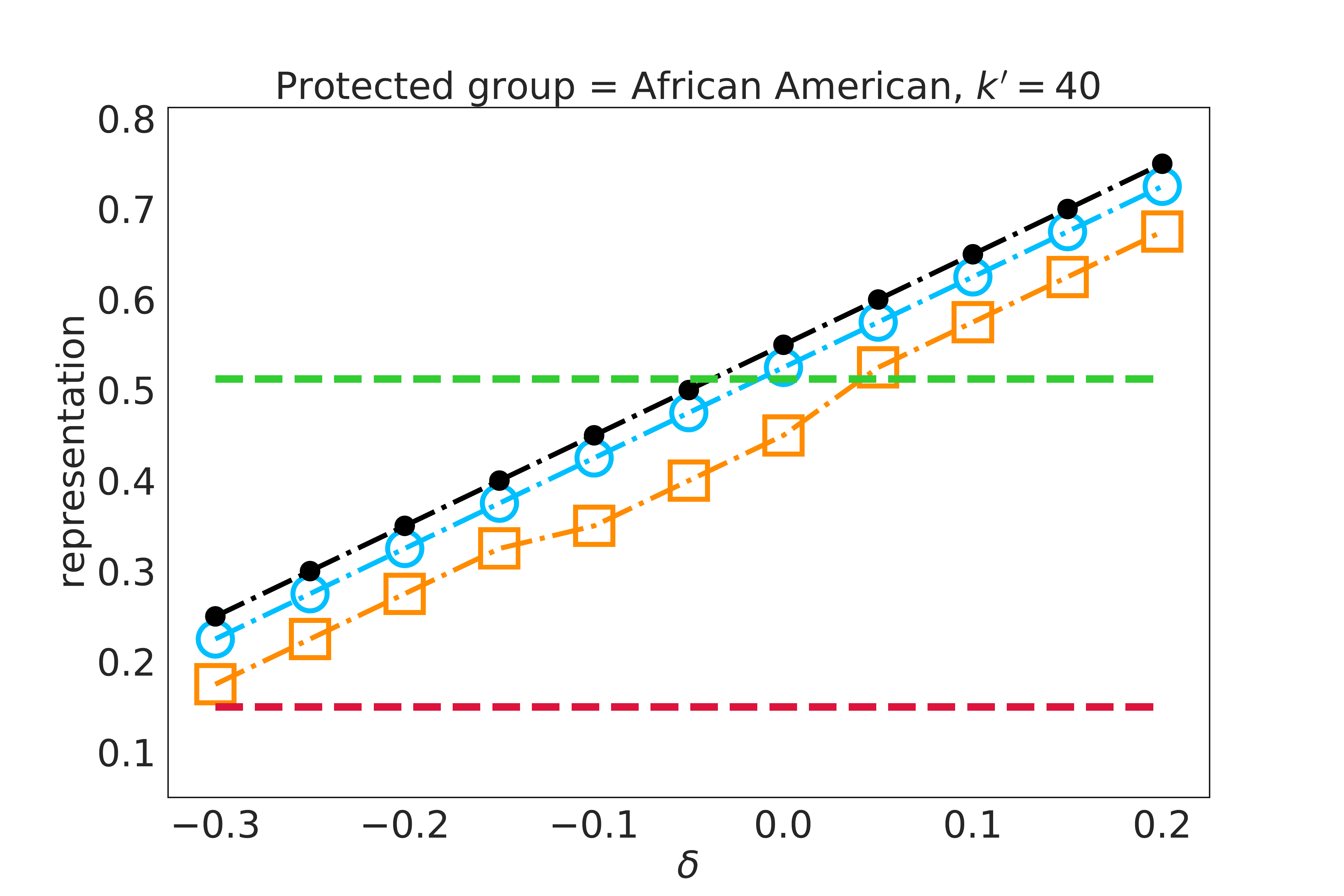} 
		\caption{Representation at top $40$ ranks.}
		\label{fig:compas_race_b}
	\end{subfigure}
	\begin{subfigure}[b]{0.33\linewidth}
		\centering
		\includegraphics[scale=0.149]{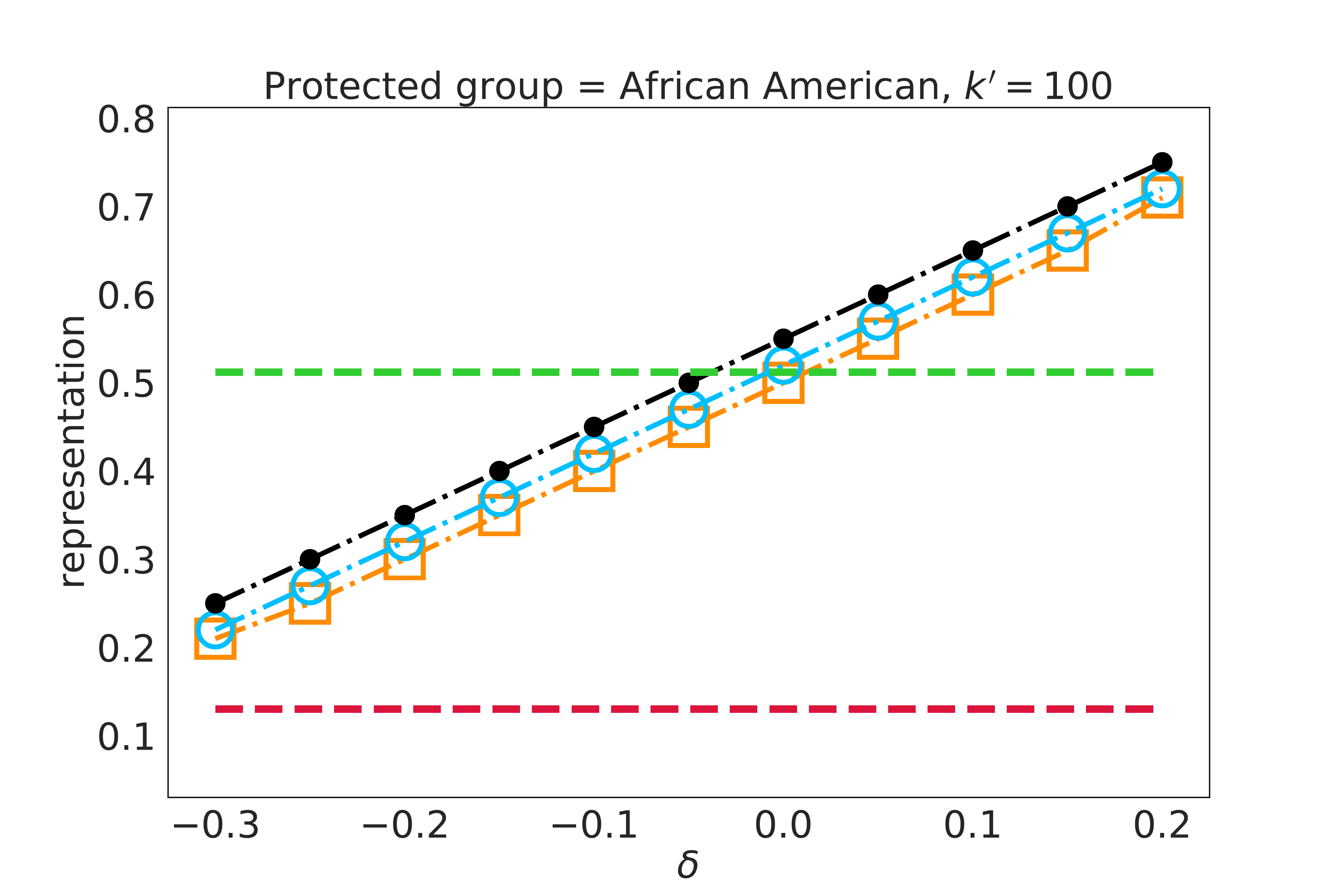} 
		\caption{Representation at top $100$ ranks.}
		\label{fig:compas_race_c}
	\end{subfigure}
	
	\begin{subfigure}[b]{0.33\linewidth}
		\centering
		\includegraphics[scale=0.149]{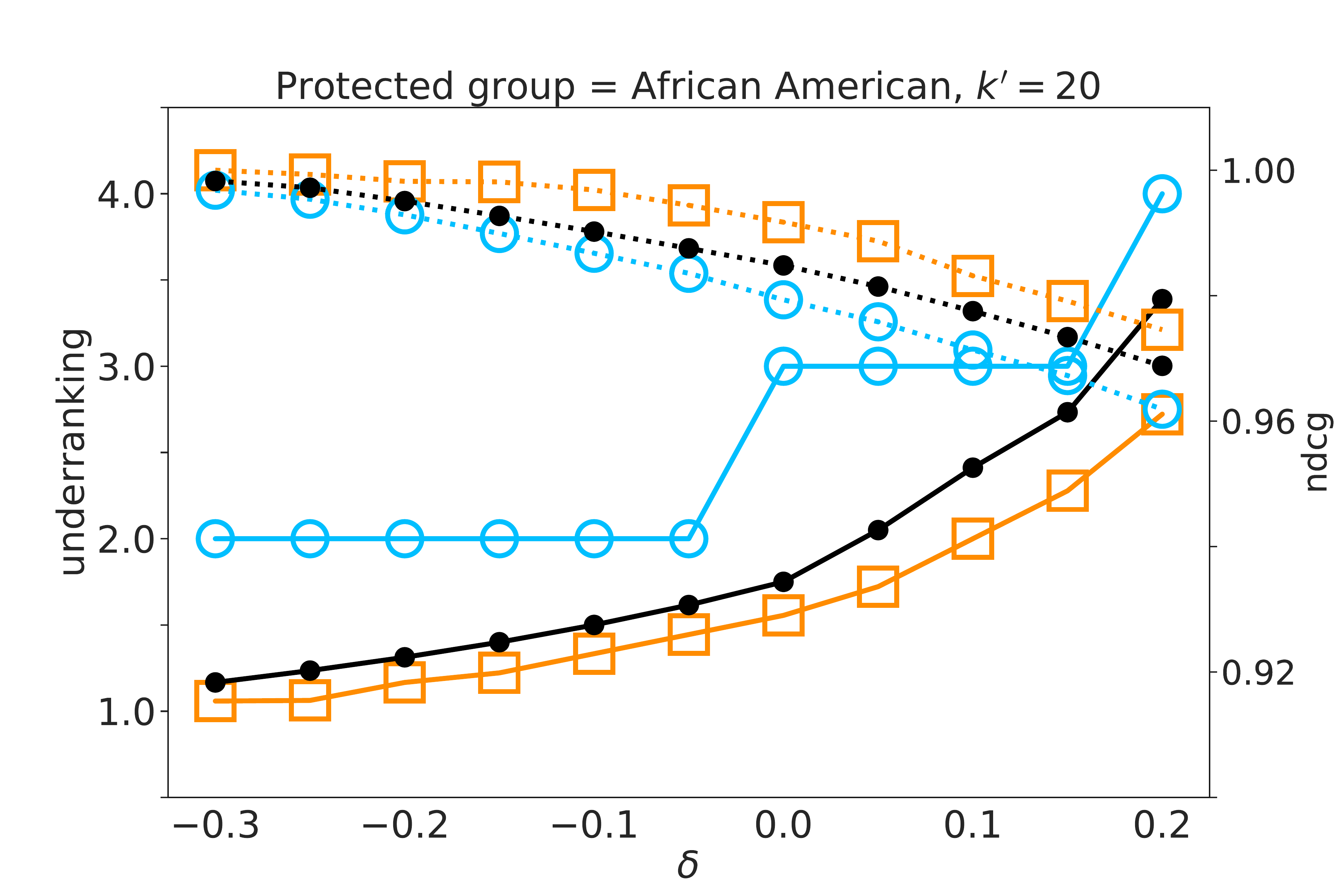} 
		\caption{Underranking, nDCG at top $20$ ranks.}
		\label{fig:compas_race_d}
	\end{subfigure}
	\begin{subfigure}[b]{0.33\linewidth}
		\centering
		\includegraphics[scale=0.149]{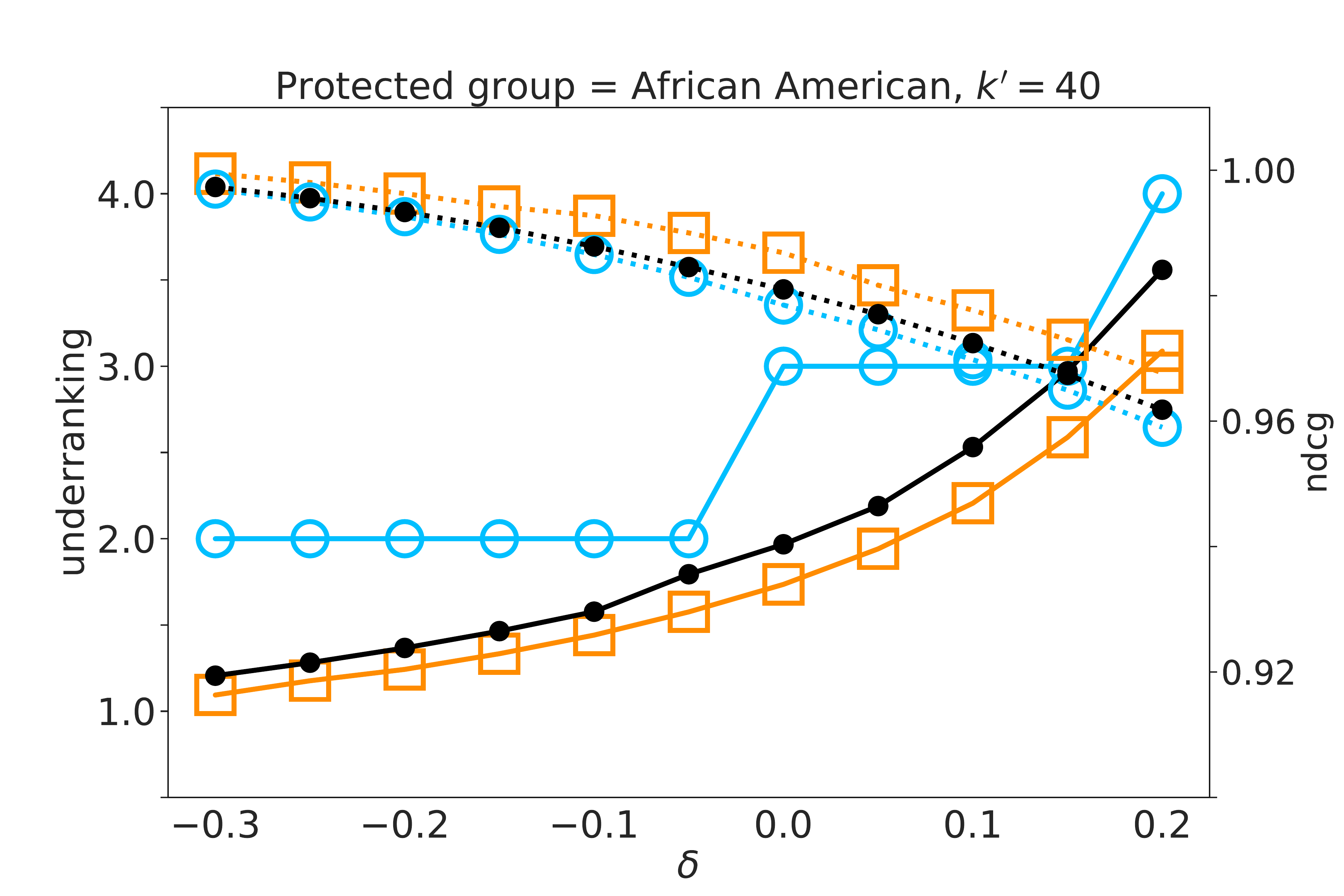} 
		\caption{Underranking, nDCG at top $40$ ranks.}
		\label{fig:compas_race_e}
	\end{subfigure}
	\begin{subfigure}[b]{0.33\linewidth}
		\centering
		\includegraphics[scale=0.149]{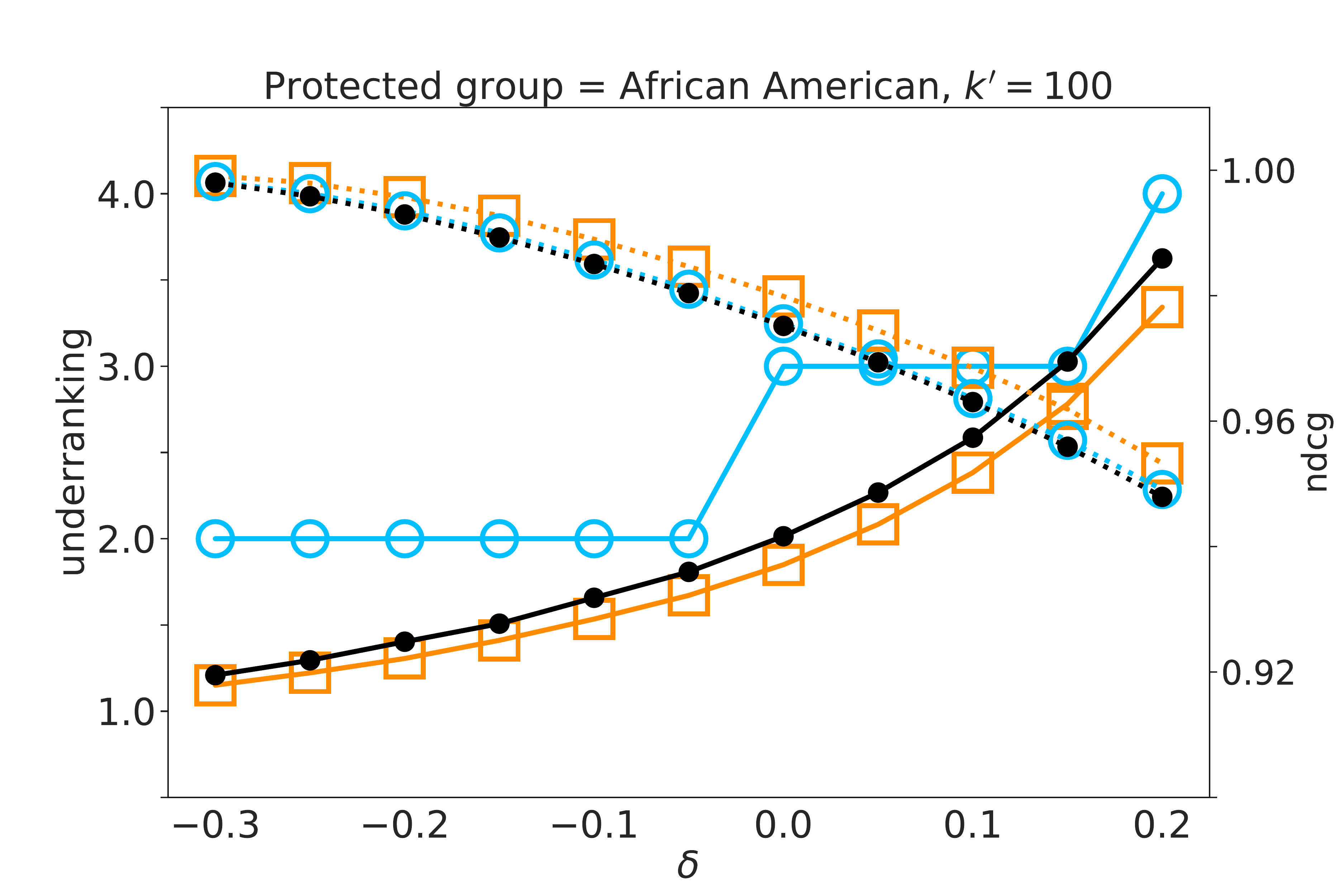} 
		\caption{Underranking, nDCG at top $100$ ranks.}
		\label{fig:compas_race_f}
	\end{subfigure}
	\caption{Results on the COMPAS Recidivism dataset with \textit{African American} as the protected group.}
	\label{fig:compas_race}
\end{figure*}

\subsection{Experimental Setup}

\textbf{Fairness constraints. }
Representation of a group in a ranking is measured by its proportion in the ranking.
FA*IR can only work with one protected and one non-protected group, and can only handle minimum representation requirements in each prefix of the top $k$ ranks. 
Hence, in all the results shown in Figures \ref{fig:german_25} to \ref{fig:compas_gender} there is only one protected group and the algorithms are run only with lower bound constraints on representation of the protected group as follows. 
Let $l=1$ and $l=2$ correspond to protected and non-protected group respectively.
Let $p_l^*$ be the representation of the group $l$ in the entire dataset.
Sufficient representation of a group need not necessarily mean there has to be exactly $p_l^*$ fraction of items in the top $k$ ranks from group $l$.
Hence, we run experiments by varying this sufficient representation requirement using a control parameter $\delta$.
FA*IR is run with $p = p_1^* + \delta$,  the DP algorithm from \cite{ranking_with_fairness_constraints} is run with the fairness constraints, $\forall k' \in [k], L_{1, k'} = \ceil{(p_1^* + \delta) k'}, L_{2,k'} = 0, U_{1, k'} = k', U_{2, k'} = k'$.
ALG is also run with group fairness constraints $\paren{\boldsymbol{\alpha} = (1, 1), \boldsymbol{\beta} = (p_1^* + \delta, 0), k=100}$, and the parameter $\epsilon = 0.4$.
In \Cref{fig:german_age}, we run ALG with fairness constraints $\paren{\boldsymbol{\alpha} = (p_1^* + \delta, p_2^* + \delta, p_3^* + \delta), \boldsymbol{\beta} = (p_1^* - \delta, p_2^* - \delta, p_3^* - \delta), k=100}$, and $\epsilon = 0.4$.

\textbf{Evaluation metrics.}
In all the datasets, the true ranking is generated based on the decreasing order of the score (or relevance) of the item.
For an algorithm run with $k = 100$, we evaluate its group fairness, underranking and ranking utility -- normalised discounted cumulative gain (nDCG) -- at top $k' = 20, 40, 100$ ranks since we are comparing with the baselines that have group fairness constraints in every prefix of the top $k$ ranking. 

\begin{enumerate}
	\item Let $G_1$ represent the ranks assigned to items from the protected group. Then,
	\begin{align*}
	\text{repersentation in top }k' = \abs{\set{i\in G_1, i \le k'}}/k'.
	\end{align*}
	\item For an item ranked at $j \leq k'$ in true ranking, let $r_j$ be its rank in the group-fair ranking. Then,
	\begin{align*}
	\text{underranking for top $k'$ ranks} = \max_{j \in [k'], r_j }r_j/j. 
	\end{align*} 
	\item Let $y_i$ be the score of the item at rank $i$ in true ranking and $\hat{y}_i$ be the score of the item at rank $i$ in the group-fair ranking.
	Then,
	\begin{align*}
	\text{nDCG}_{k'} = \paren{\sum_{i=1}^{k'}\frac{2^{\hat{y}_i}}{\log_2(i+1)}} / \paren{\sum_{i=1}^{k'}\frac{2^{y_i}}{\log_2(i+1)}}.
	\end{align*}
\end{enumerate}

\begin{figure*}[t]
	\begin{subfigure}[b]{\linewidth}
		\centering
		\includegraphics[scale=0.2]{results/legend.pdf} 
	\end{subfigure}
	
	\begin{subfigure}[b]{0.33\linewidth}
		\centering
		\includegraphics[scale=0.149]{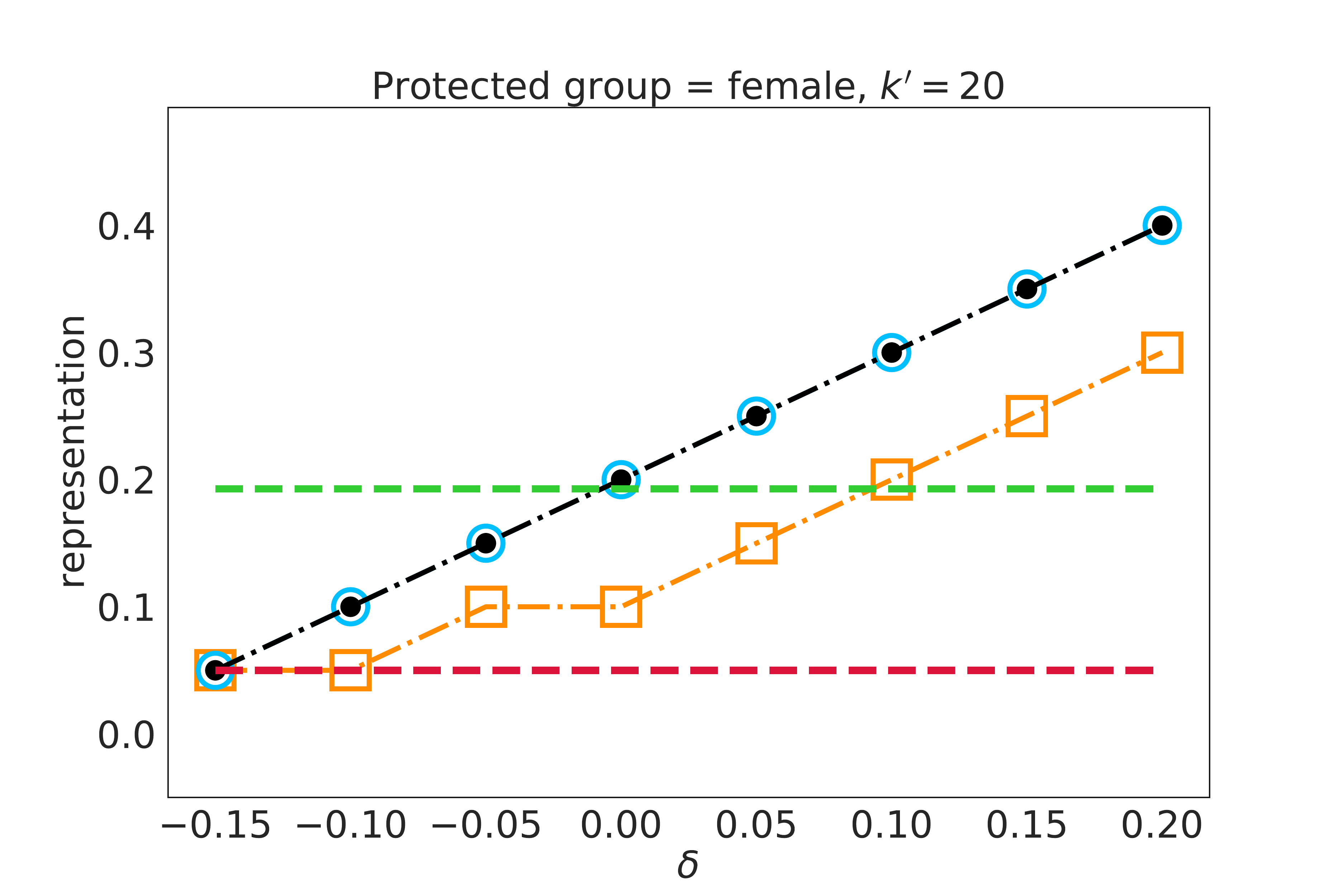} 
		\caption{Representation at top $20$ ranks.}
		\label{fig:compas_gender_a}
	\end{subfigure}
	\begin{subfigure}[b]{0.33\linewidth}
		\centering
		\includegraphics[scale=0.149]{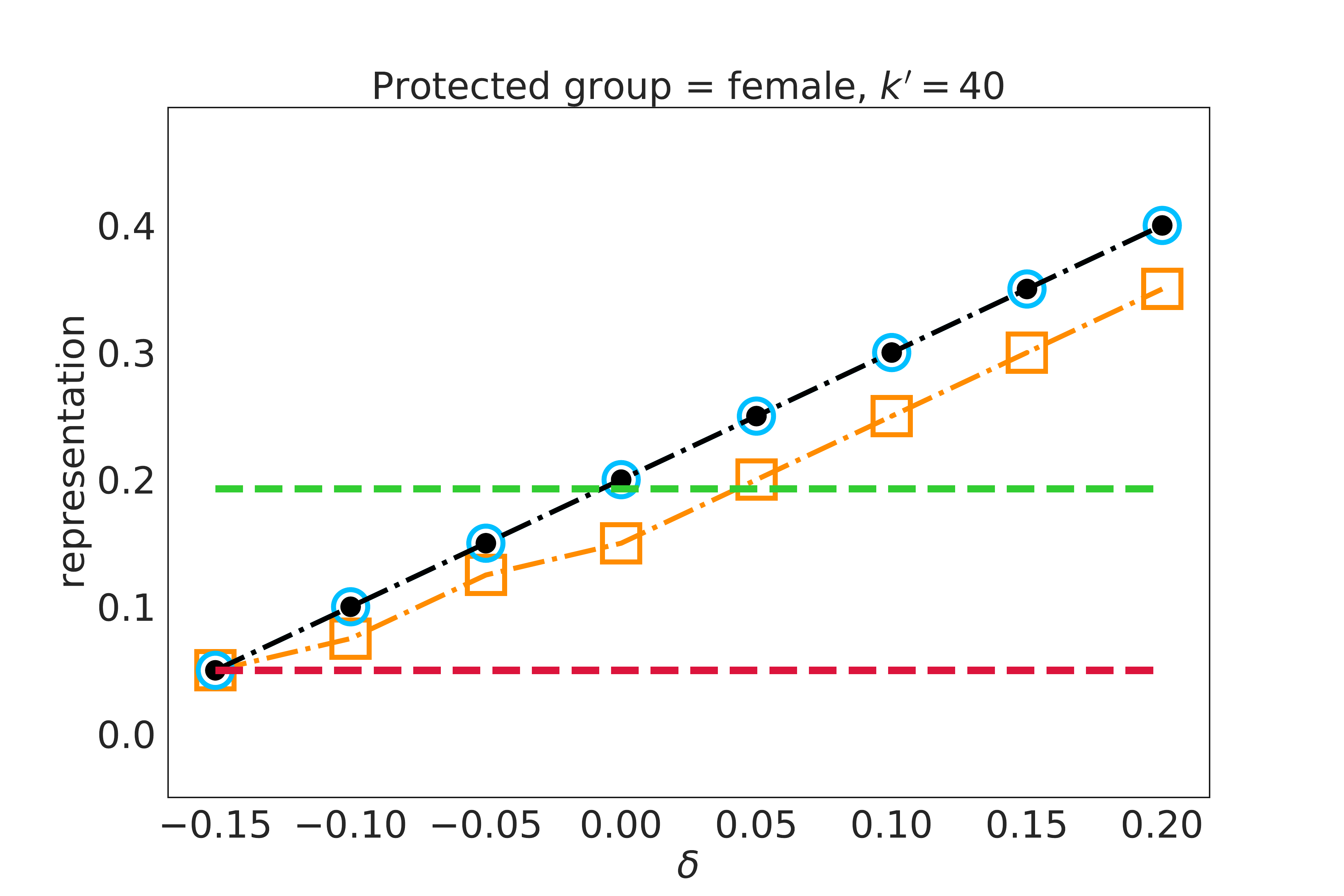} 
		\caption{Representation at top $40$ ranks.}
		\label{fig:compas_gender_b}
	\end{subfigure}
	\begin{subfigure}[b]{0.33\linewidth}
		\centering
		\includegraphics[scale=0.149]{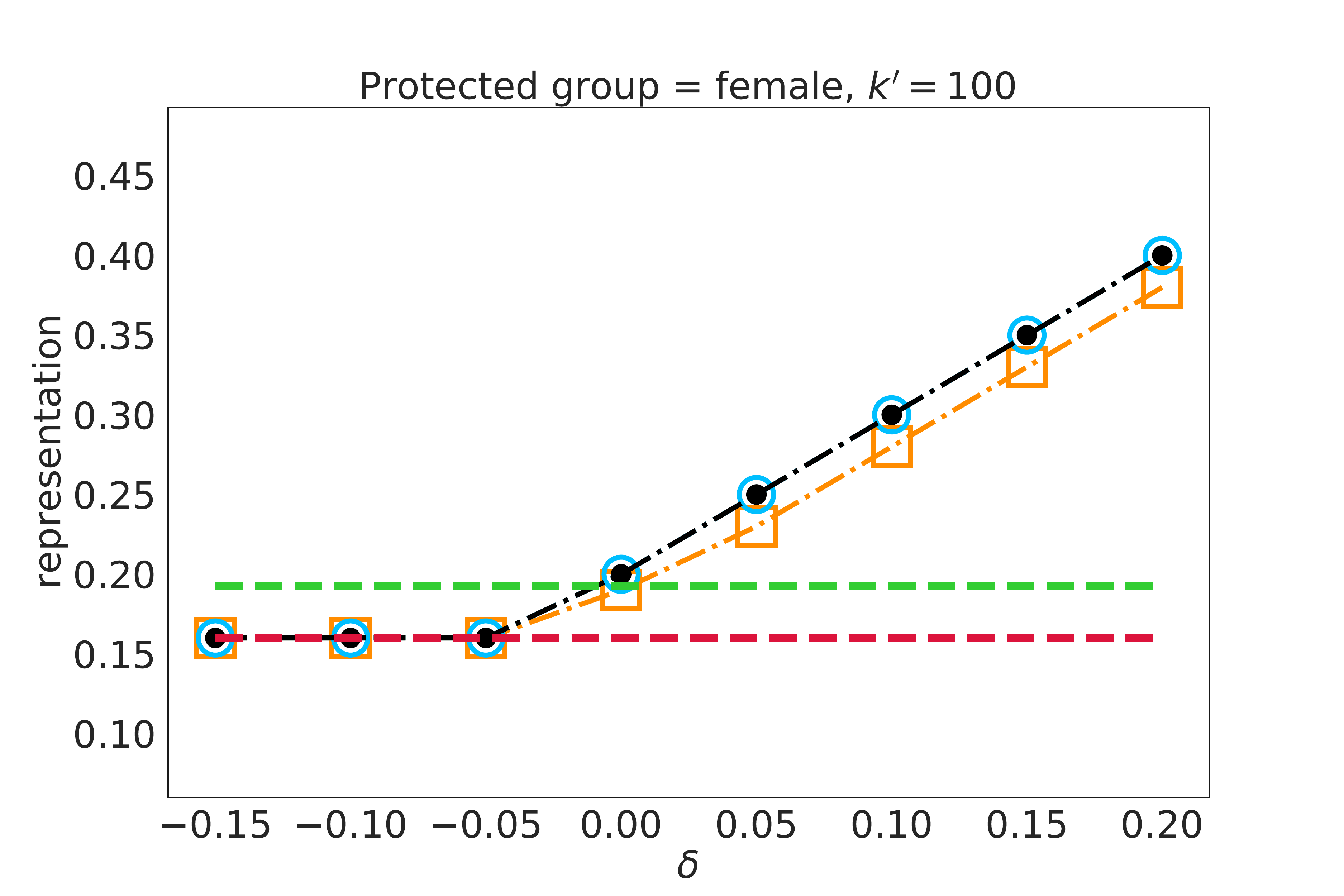} 
		\caption{Representation at top $100$ ranks.}
		\label{fig:compas_gender_c}
	\end{subfigure}
	
	\begin{subfigure}[b]{0.33\linewidth}
		\centering
		\includegraphics[scale=0.149]{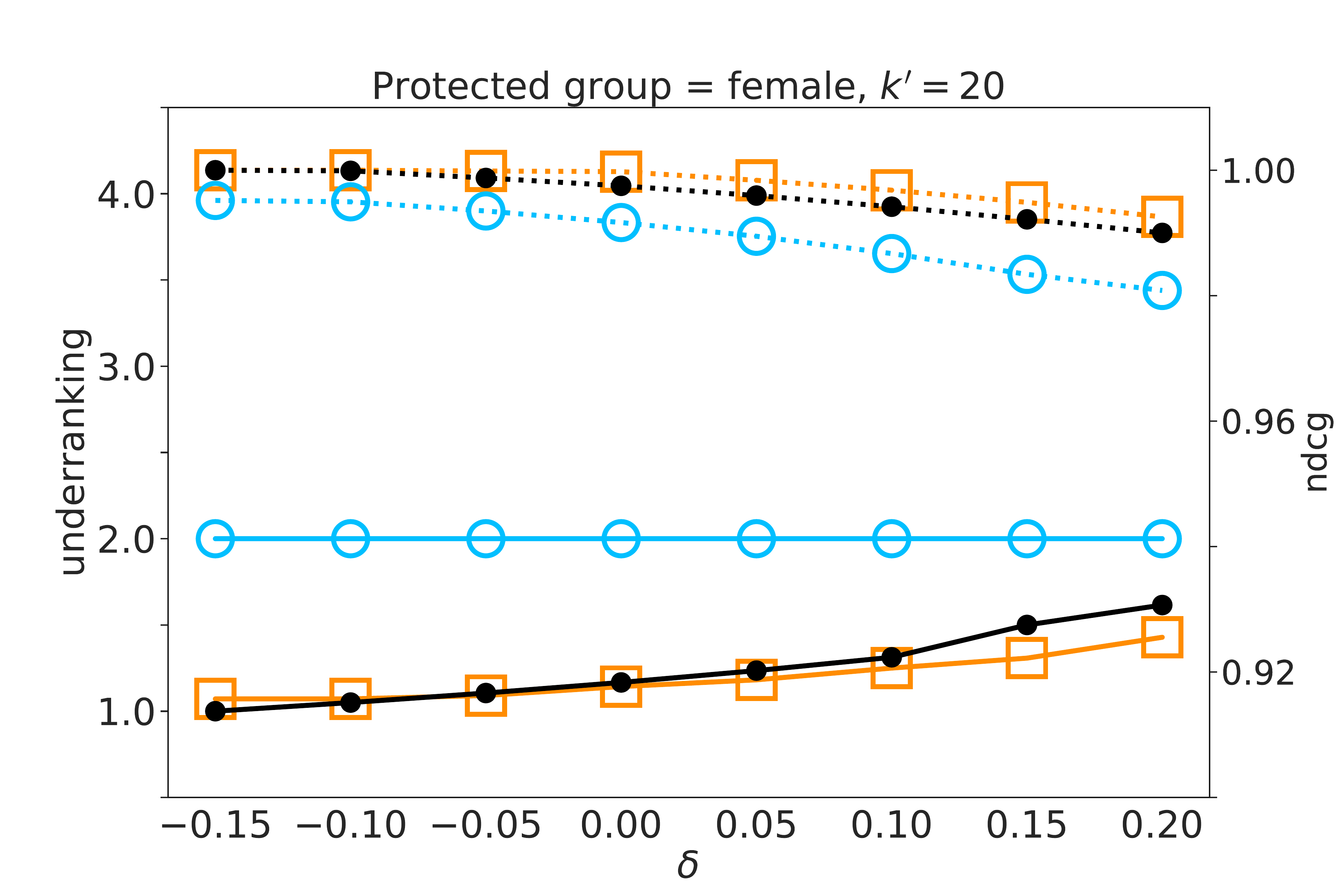} 
		\caption{Underranking, nDCG at top $20$ ranks.}
		\label{fig:compas_gender_d}
	\end{subfigure}
	\begin{subfigure}[b]{0.33\linewidth}
		\centering
		\includegraphics[scale=0.149]{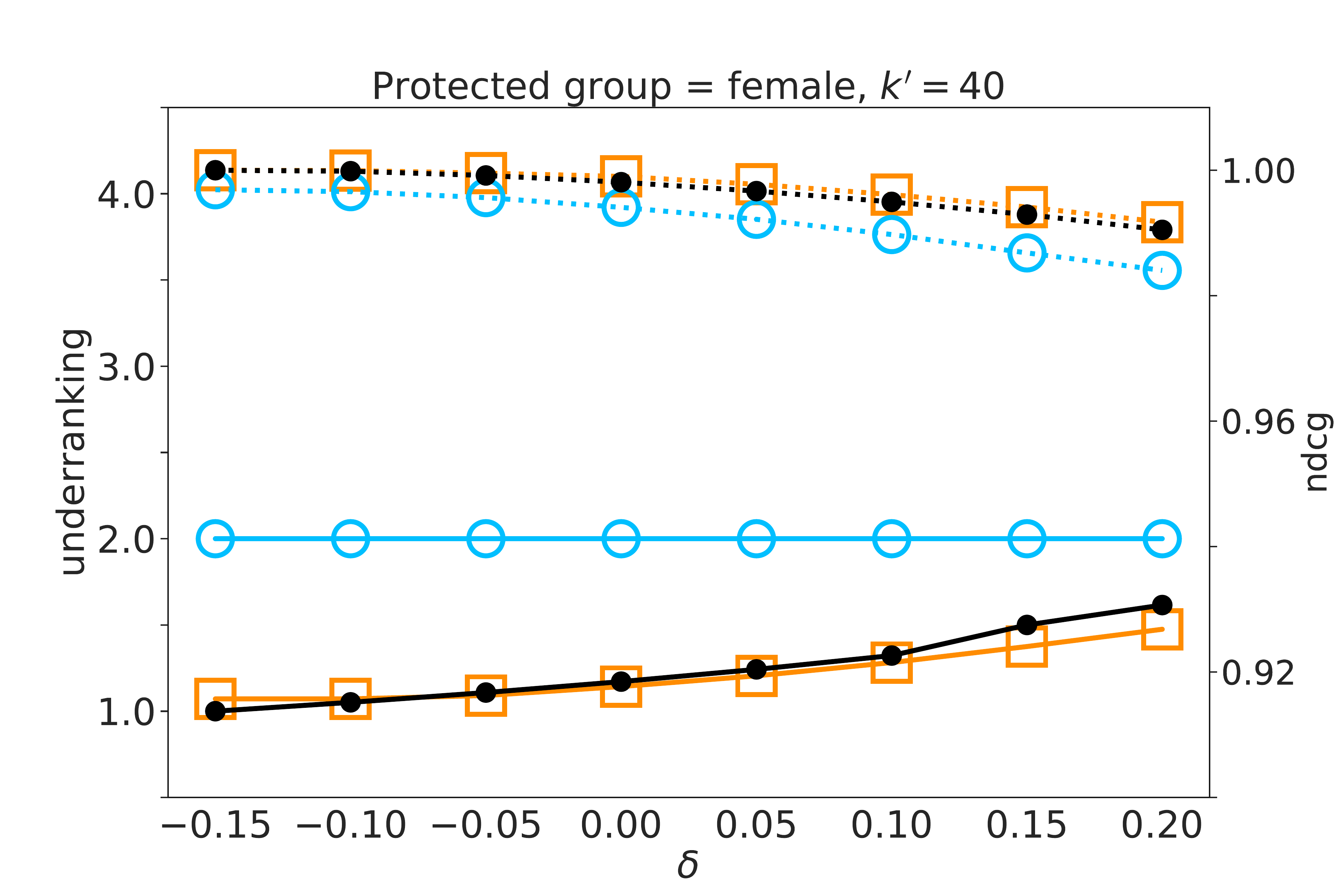} 
		\caption{Underranking, nDCG at top $40$ ranks.}
		\label{fig:compas_gender_e}
	\end{subfigure}
	\begin{subfigure}[b]{0.33\linewidth}
		\centering
		\includegraphics[scale=0.149]{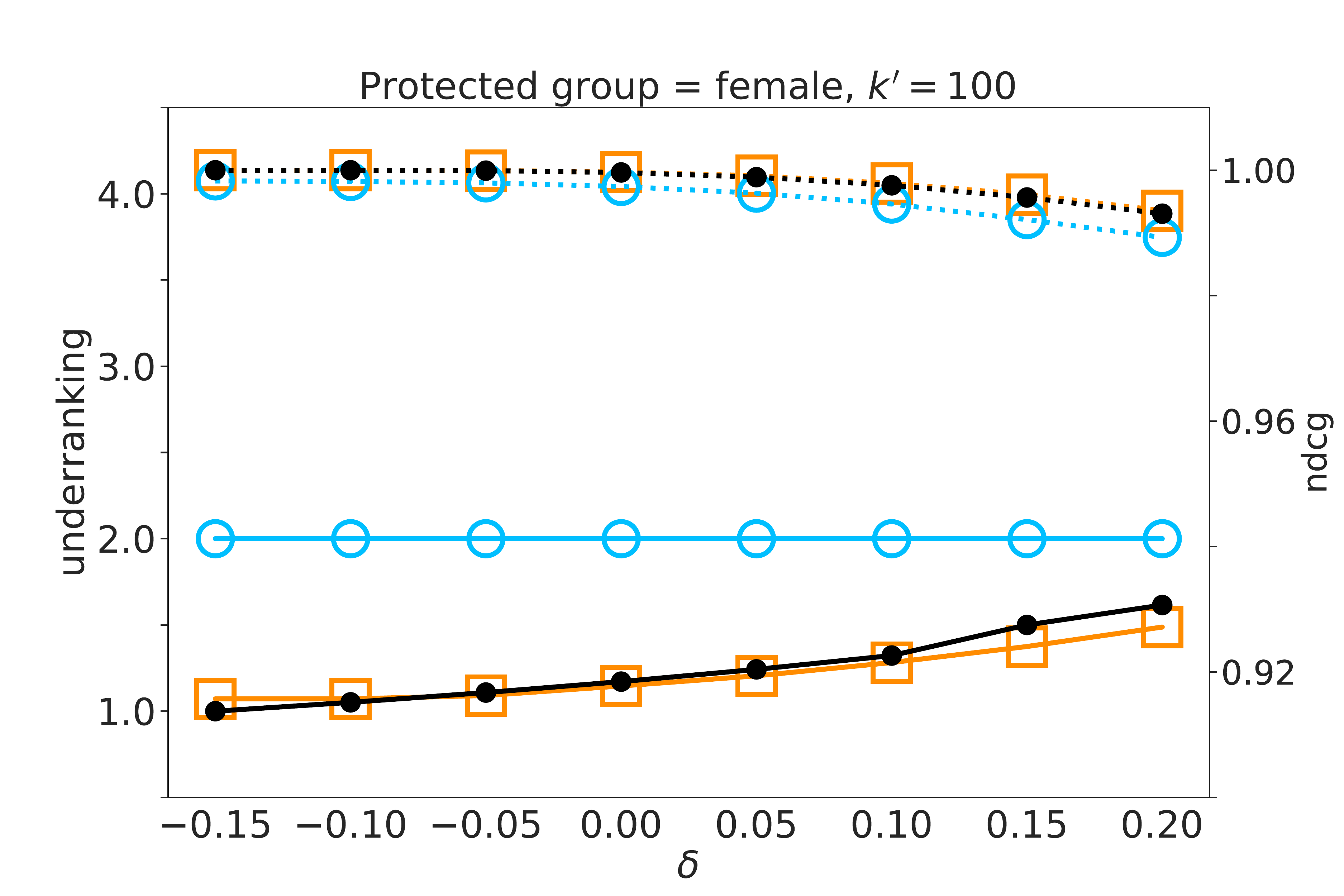} 
		\caption{Underranking, nDCG at top $100$ ranks.}
		\label{fig:compas_gender_f}
	\end{subfigure}
	\caption{Results on the COMPAS Recidivism dataset with \textit{Female} as the protected group.}
	\label{fig:compas_gender}
\end{figure*}

\textbf{Reading the plots.} For every combination of a dataset and a protected group, we show a \emph{pair} of plots.
Consider \Cref{fig:german_25_a}, \ref{fig:german_25_d}. Y-axis in \Cref{fig:german_25_a} shows the representation of the protected group \textit{age }$<25$ in the top $20$ ranks for each run of the algorithm with fairness constraints controlled by $\delta$ on X-axis. 
Here, the dashed green line shows the proportion of \textit{age }$<25$ in the dataset, whereas the dashed red line shows their proportion in the top $20$ ranks of the true ranking. 
These two lines serve as guidelines to understand the behavior of various algorithms.
\Cref{fig:german_25_d} shows corresponding underranking and nDCG in the top $20$ ranks.

\subsection{Experimental Observations}
\paragraph{Trade-off between underranking and group fairness.}
In the COMPAS dataset the female candidates are underrepresented in any of the top $k' \in [k]$ ranks (see dashed red lines  in Figures \ref{fig:compas_gender_a} to \ref{fig:compas_gender_c}) compared to their true representation in the dataset, $p^* = 0.19$ (dashed green lines).
By varying $\delta$ (on X-axis), we run the experiments with the minimum female representation constraint, $p^* + \delta$. 
Now, comparing \Cref{fig:compas_gender_a} with \Cref{fig:compas_gender_d} for varying $\delta$, we show the trade-off between group fairness and underranking.
As the value of $\delta$ increases from $-0.15$ to $0.2$, the underranking gets worse since more number of male candidates with lower true ranks have to be moved to higher ranks in order to accomodate for the required female repersentation in the top $k$ ranks.
Similarly, even though African Americans have very high representation, $p^* = 0.55$ (see green line in \Cref{fig:compas_race_a}), in the true ranking, their representation in the top $k'$ ranks is again significantly less (see red lines in Figures \ref{fig:compas_race_a} to \ref{fig:compas_race_c}).
Even in this case, we observe a trade-off between group fairness and underranking in the top $k'$ ranks.
These trends are also observed at any of the top $k' = 20, 40, 100$ ranks in the German Credit dataset (see \Cref{fig:german_25}).
\Cref{fig:german_35} shows the evaluation of ALG and the baselines on the German credit risk dataset with \textit{age }$< 35$ as the protected group.
The candidates with \textit{age }$< 35$ are underrepresented in the top $k'$ ranks (dashed red lines), even though they have very high represention (dashed green lines) in the whole dataset.
Hence, we observe a trade-off between representation of the protected group and underranking with all the algorithms.

We also partition the candidates in the German Credit dataset into three groups based on \textit{age}, and enforce the constraints that each group $l \in \set{\textit{age} < 25, \textit{age }  25-35, \textit{age} > 35}$, with corresponding $p_l^*$, should have minimum and maximum represention requirements $p_l^* - \delta$ and $p_l^* + \delta$ respectively (see \Cref{fig:german_age}).
Even in this case, the underranking gets worse with increase in the lower bound representation requirements because of the underrepresentation of the protected groups, hence confirming the trade-off even for more than two groups.

In \Cref{fig:german_age_gender}, we show results of ALG with the group fairness constraints for the groups formed by the intersection of two groups \textit{age} and \textit{gender}. 
We partition the candidates into six groups and their true representations are as show in \Cref{tab:true_german_age_gender}.
ALG is run with the constraints $\alpha_l = p_l^* + \delta$ and $\beta_l = p_l^* - \delta$, where $l$ represents the group.
ALG achieves proportional representation while the underranking decreases as $\delta$ increases, hence confirming the trade-off.
Despite being designed to satisfy group fairness constraints for every $k$ consecutive ranks, ALG achieves representation in prefixes of the ranking similar to that of the baselines.
These experimental results show evidence of the trade-off between group fairness and underranking in the real-world datasets.

\paragraph{Underranking for comparing different group-fair rankings.}
In previous work, only the trade-off between fairness and utility of the ranking such as nDCG has been studied. 
However, as established in \Cref{fig:example} and is evident from our experimental results, high nDCG does not imply anything for the underranking of a group-fair ranking.
For example, in \Cref{fig:compas_race}, Celis et al.'s DP algorithm achieves almost same nDCG and group fairness as ALG but suffers badly in terms of underranking.
Hence, underranking allows us to break ties when aggregate ranking utility and group fairness are same for any two group-fair rankings.

\paragraph{ALG achieves best trade-off between group fairness and underranking.}
In all the results in Figures \ref{fig:german_25} to \ref{fig:compas_gender}, Celis et al.'s DP algorithm achieves worse underranking and same representation as ALG, and FA*IR achieves worse representation and same underranking as ALG.
We note that even though in \Cref{fig:compas_race_d} to \ref{fig:compas_race_f}, FA*IR seems to achieve better underranking than ALG, it has significantly below the minimum protected group representation in the top $k'$ ranks.
Hence, we posit that ALG achieves the best trade-off between underranking and representation of the protected group in the real-world datasets.

\begin{figure}[H]
	\centering
	\includegraphics[width=0.85\textwidth]{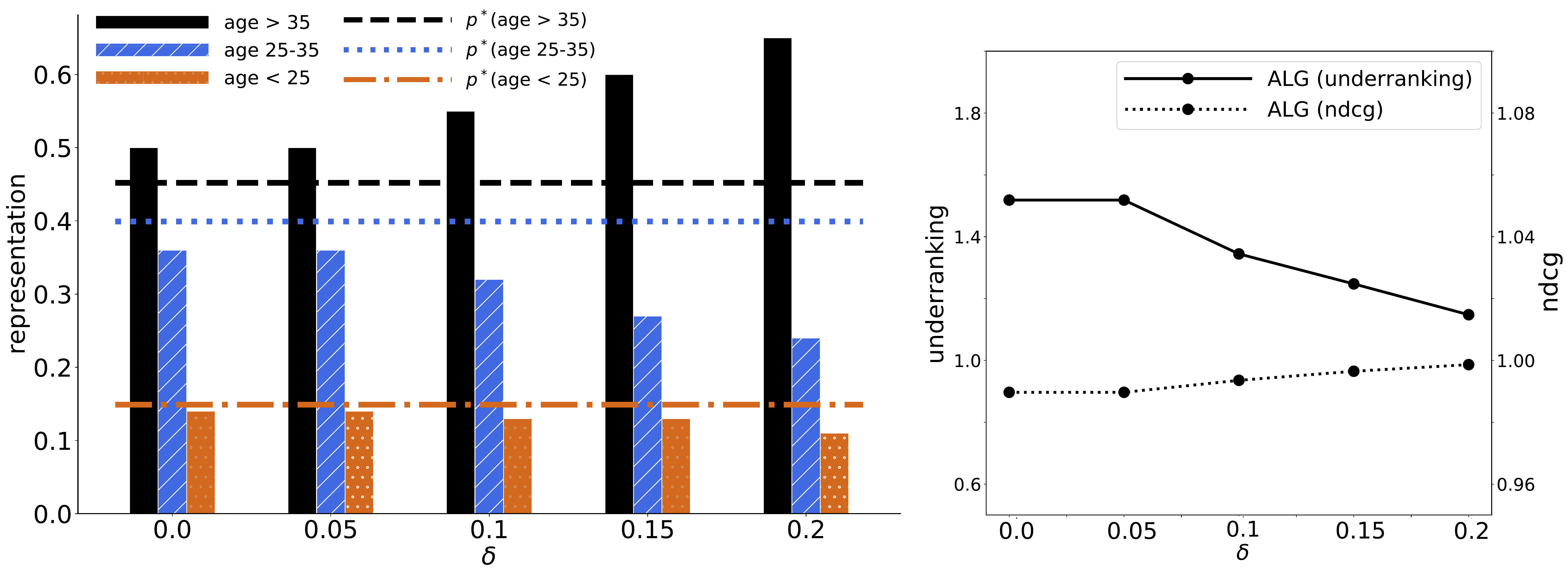}  	
	\caption{Results of ALG on the German Credit Risk dataset with three groups based on \textit{age}. We run ALG with fairness constraints $\paren{\boldsymbol{\alpha} = (p_1^* + \delta, p_2^* + \delta, p_3^* + \delta), \boldsymbol{\beta} = (p_1^* - \delta, p_2^* - \delta, p_3^* - \delta), k=100}$, where $l = 1,2,3$ correspond to the groups \textit{age }$>35$, \textit{age }$25 - 35$, and \textit{age }$<25$ respectively. The bar plot (left) shows the representation achieved by ALG at top $100$ ranks. The line $p^*$ for each group shows its representation in the dataset. Corresponding underranking, nDCG shown on the right. }
	\label{fig:german_age}
\end{figure}

\begin{figure}[H]
	\begin{subfigure}[b]{0.55\linewidth}
		\centering
		\includegraphics[scale=0.25]{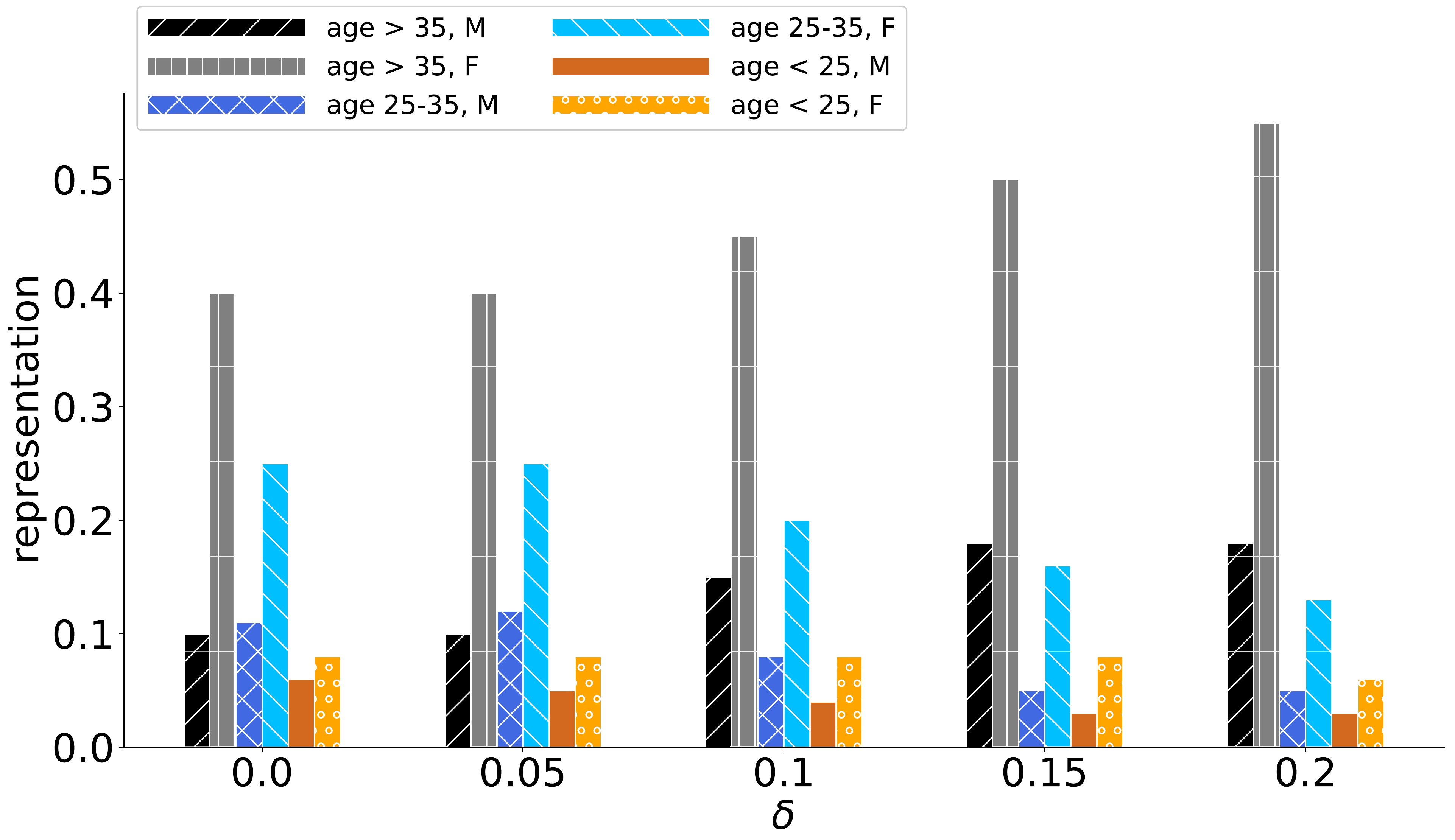} 
		\caption{}
	\end{subfigure}
	\begin{subfigure}[b]{0.37\linewidth}
		\centering
		\includegraphics[width=\textwidth]{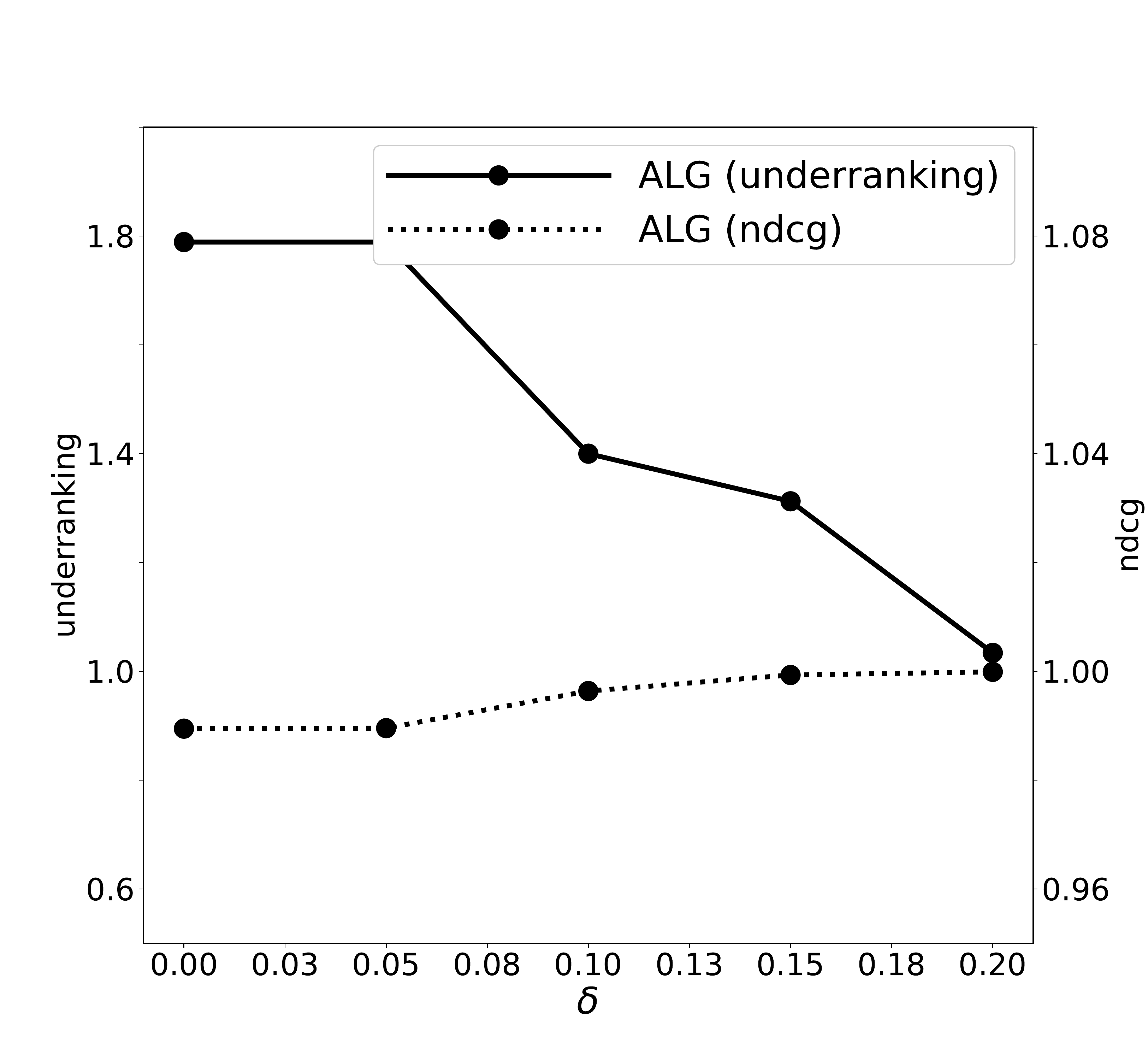} 
		\caption{}
	\end{subfigure}
	
	\caption{Results on the German Credit Risk dataset with six groups based on \textit{age, gender}.}
	\label{fig:german_age_gender}
	
\end{figure}

\begin{table}[H]
	\begin{center}
		\begin{tabular}{|c|c|c|c|c|c|c|}
			\hline
			\textbf{Group}&age $>35$, M   & age $>35$, F&age $25$ to $35$, M & age $25$ to $35$, F &  age $<25$, M &	age $<25$, F \\
			\hline
			$p^*$  &0.10 & 0.36 & 0.13 & 0.27 &  0.08 & 0.06 \\
			\hline
		\end{tabular}
		\caption{True representation of the groups in the dataset.}
		\label{tab:true_german_age_gender}
	\end{center}
\end{table}

\paragraph{ALG runs significantly faster than the baselines.}
\begin{table}[H]
	
	\centering
	
	\resizebox{0.6\textwidth}{!}{  
		\begin{tabular}{|lrrrr|}
			
			\hline
			& $k = 100$ &   $k = 300$ &  $k = 500$ &  $k = 1000$ \\
			\hline
			Celis et al.'s DP algorithm & 11.0 & 100.0 & 186.0 & 301.0\\
			FA*IR & 3.0 & 3.3 & 3.3 & 3.5 \\
			ALG & \textbf{1.1} & \textbf{1.6} & \textbf{1.8} & \textbf{1.8} \\
			\hline
			
		\end{tabular}
	}
	\caption{Average wall clock running time (in seconds) of five runs of the algorithms on the German Credit Risk dataset with \textit{age} $< 25$ as the protected group ($n = 1000$, $p^* = 0.15$), $\ell = 2$. For these experiments we choose, $\delta = 0$. ALG is run for $((1,1), (0.15, 0), k)$ group fairness, with $\epsilon = 0.4$. FA*IR is run with $p=0.15$ and Celis et al. with $L_{ \textit{age} < 25,k'} = \ceil{0.15\cdot k'}, \forall k' \in [k]$.}
	\label{tab:time}
	
\end{table}
Table \ref{tab:time} shows the average running time of each algorithm.
The experiments were run on a Dual Intel Xeon 4110 processor consisting of 16 cores (32 threads), with a clock speed of 2.1 GHz and DRAM of 128GB.
Although both Celis et al.'s DP algorithm and ALG can handle both upper and lower bound constraints on the representation on any number of groups, ALG has much faster average running times. ALG also runs faster than the greedy algorithm FA*IR.

\subsection{Evaluation at Consecutive Ranks}
Figures \ref{fig:german_25_block} to \ref{fig:compas_gender_block} show evaluation of the algorithms in consecutive ranks.
Plot (a) shows the evaluation at top $1$ to $20$ ranks, plot (b) shows the evaluation at top $21$ to $60$ ranks, and plot (c) shows the evaluation at top $61$ to $100$ ranks.
The plots (d), (e), (f) however still show underranking and nDCG in the top $20, 60$, and $100$ ranks respectively.
ALG achieves good trade-offs between underranking and representation in all different consecutive ranks compared to the baselines.
These experimental results show evidence of the trade-off between group fairness and underranking in the real-world datasets.

\begin{figure}[H]
	\begin{subfigure}[b]{\linewidth}
		\centering
		\includegraphics[scale=0.2]{results/legend.pdf} 
	\end{subfigure}
	
	\begin{subfigure}[b]{0.33\linewidth}
		\centering
		\includegraphics[scale=0.149]{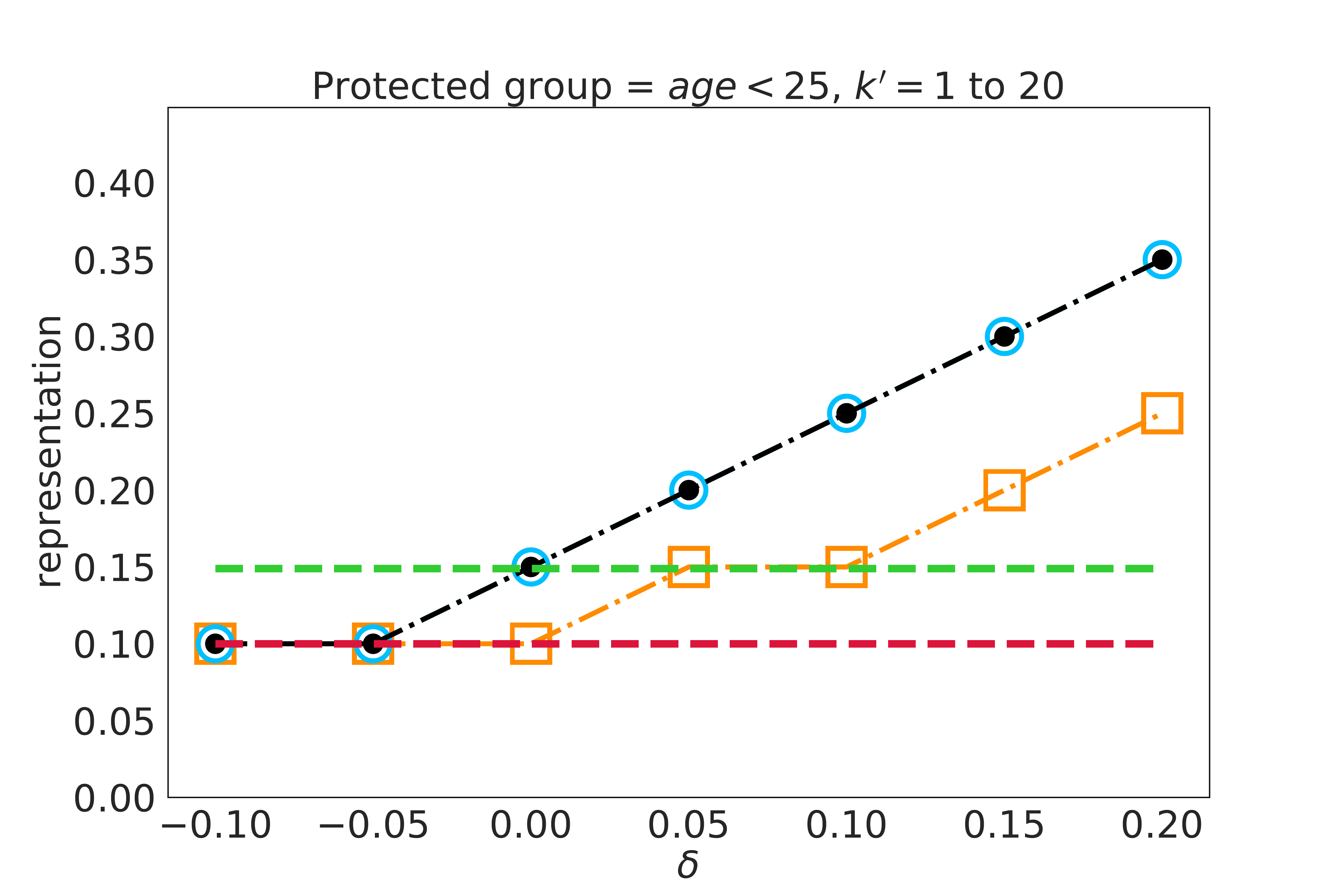} 
		\caption{Representation at top $20$ ranks.}
	\end{subfigure}
	\begin{subfigure}[b]{0.33\linewidth}
		\centering
		\includegraphics[scale=0.149]{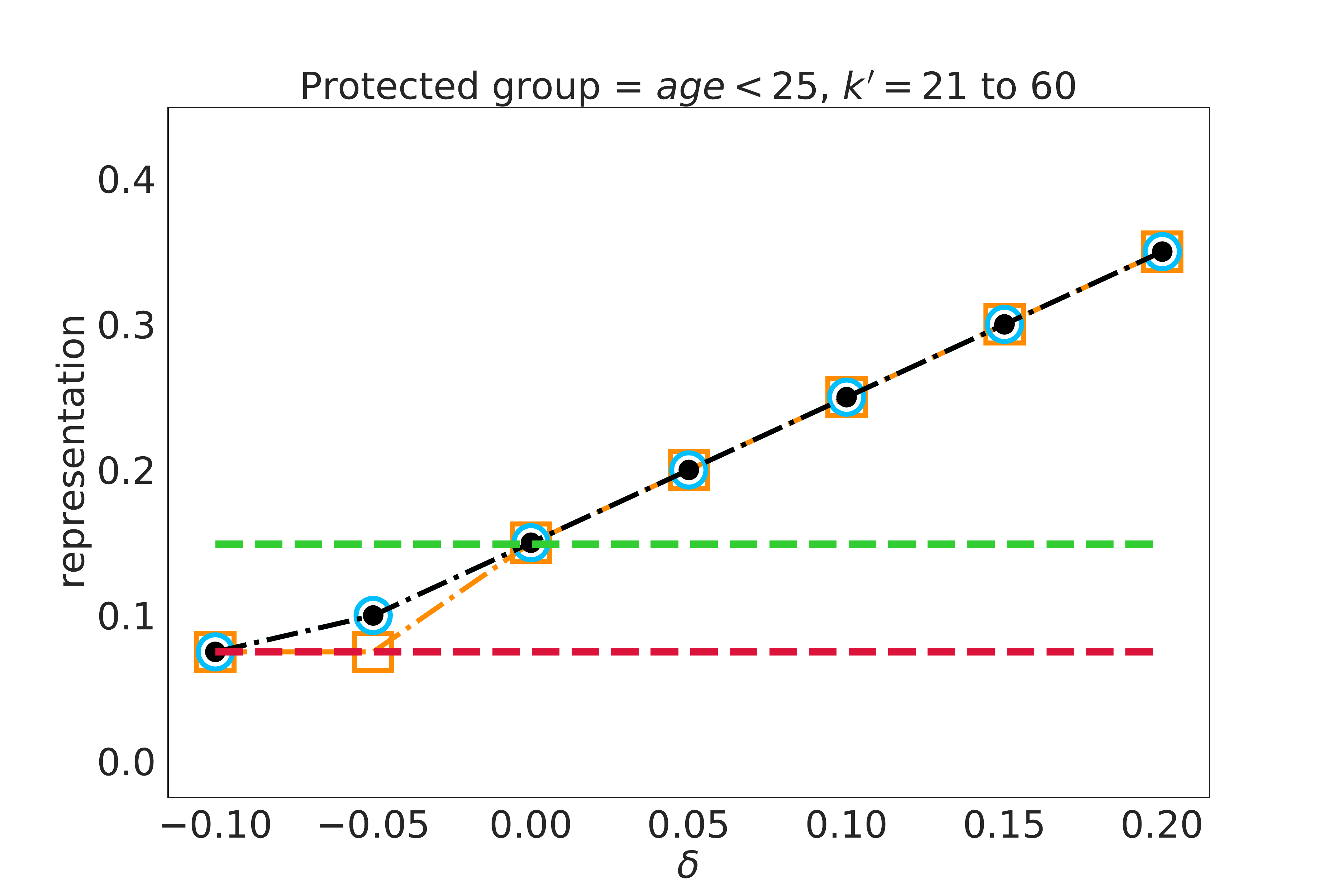} 
		\caption{Representation at top $40$ ranks.}
	\end{subfigure}
	\begin{subfigure}[b]{0.33\linewidth}
		\centering
		\includegraphics[scale=0.149]{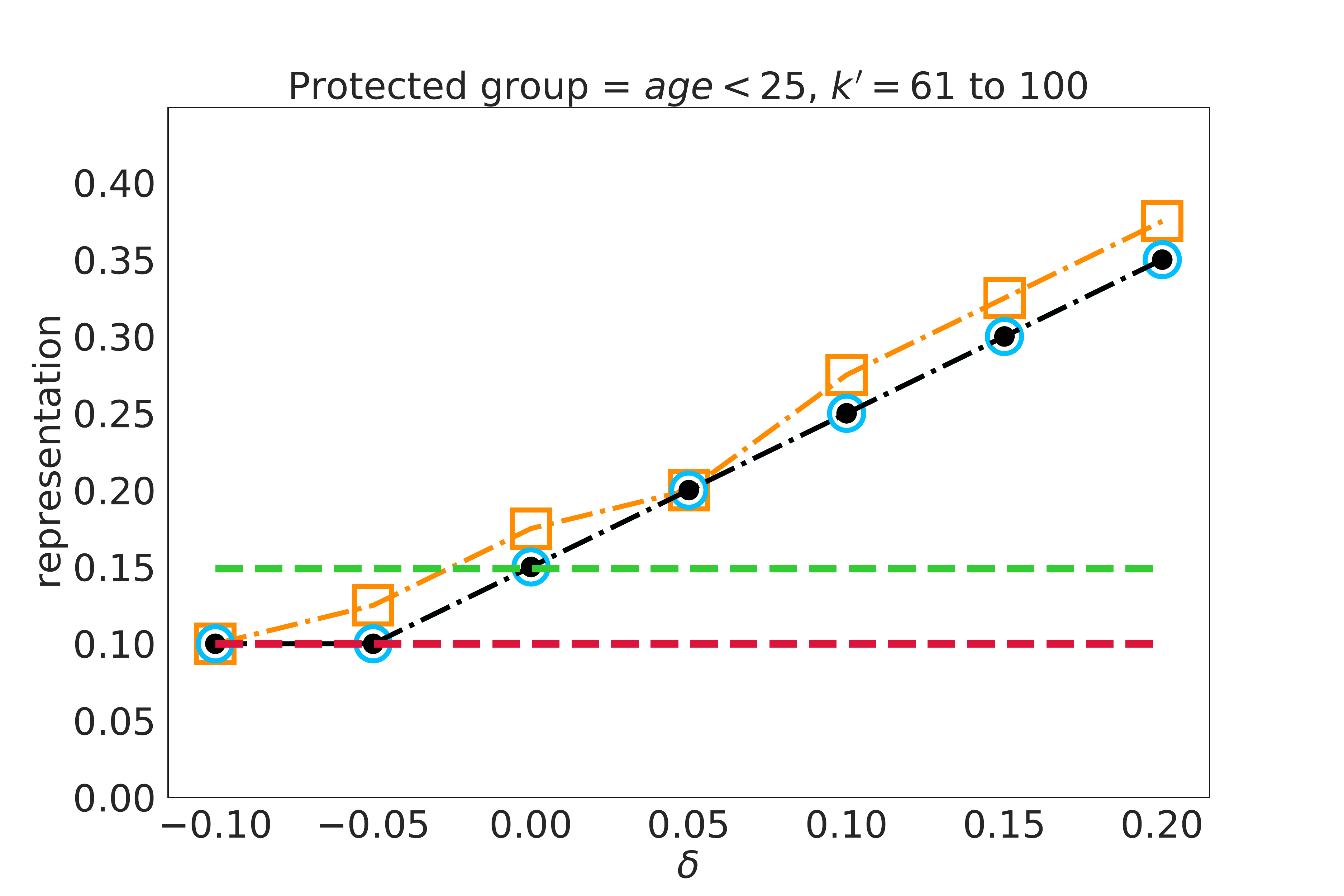} 
		\caption{Representation at top $100$ ranks.}
	\end{subfigure}
	
	\begin{subfigure}[b]{0.33\linewidth}
		\centering
		\includegraphics[scale=0.149]{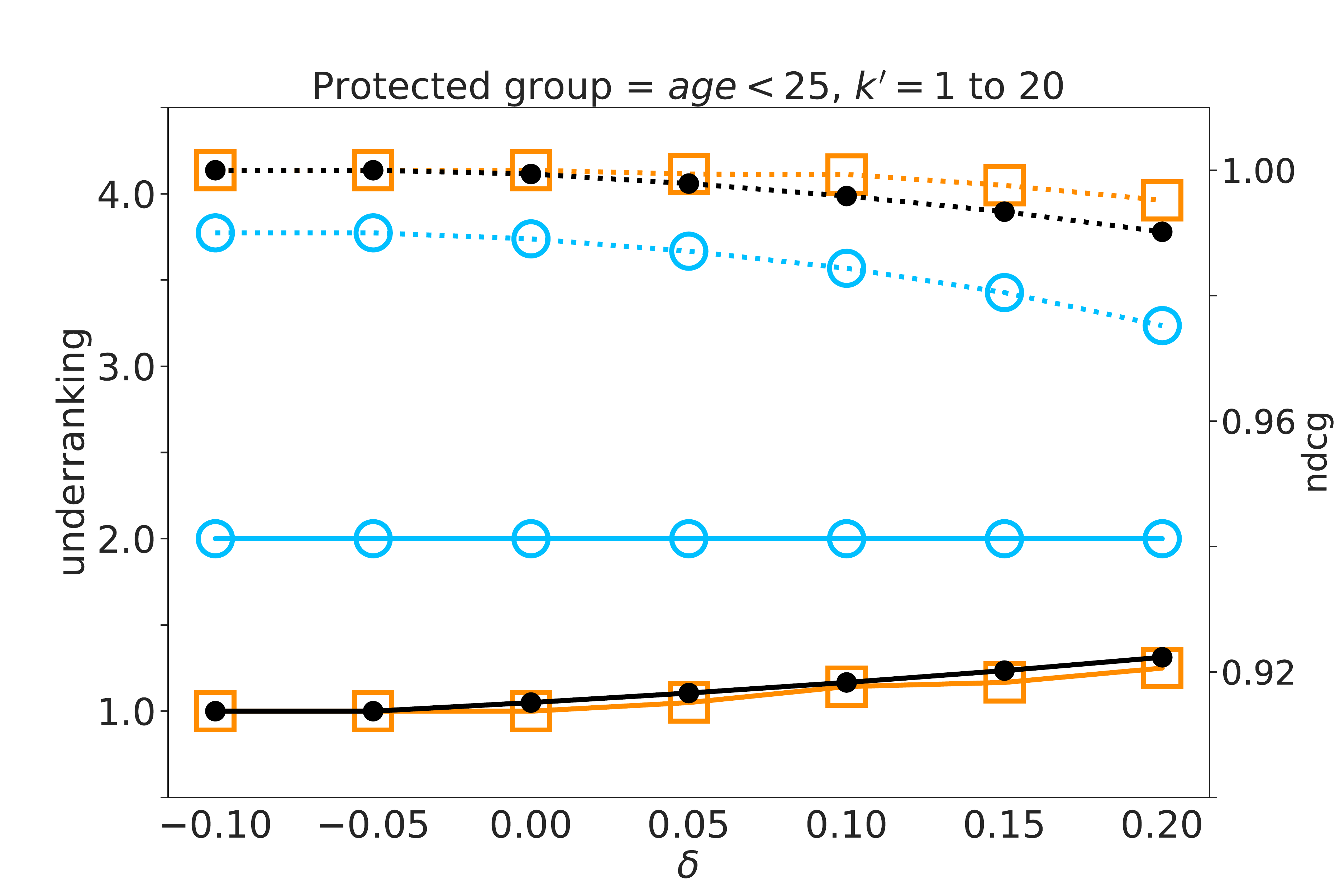} 
		\caption{Underranking, nDCG at top $20$ ranks.}
	\end{subfigure}
	\begin{subfigure}[b]{0.33\linewidth}
		\centering
		\includegraphics[scale=0.149]{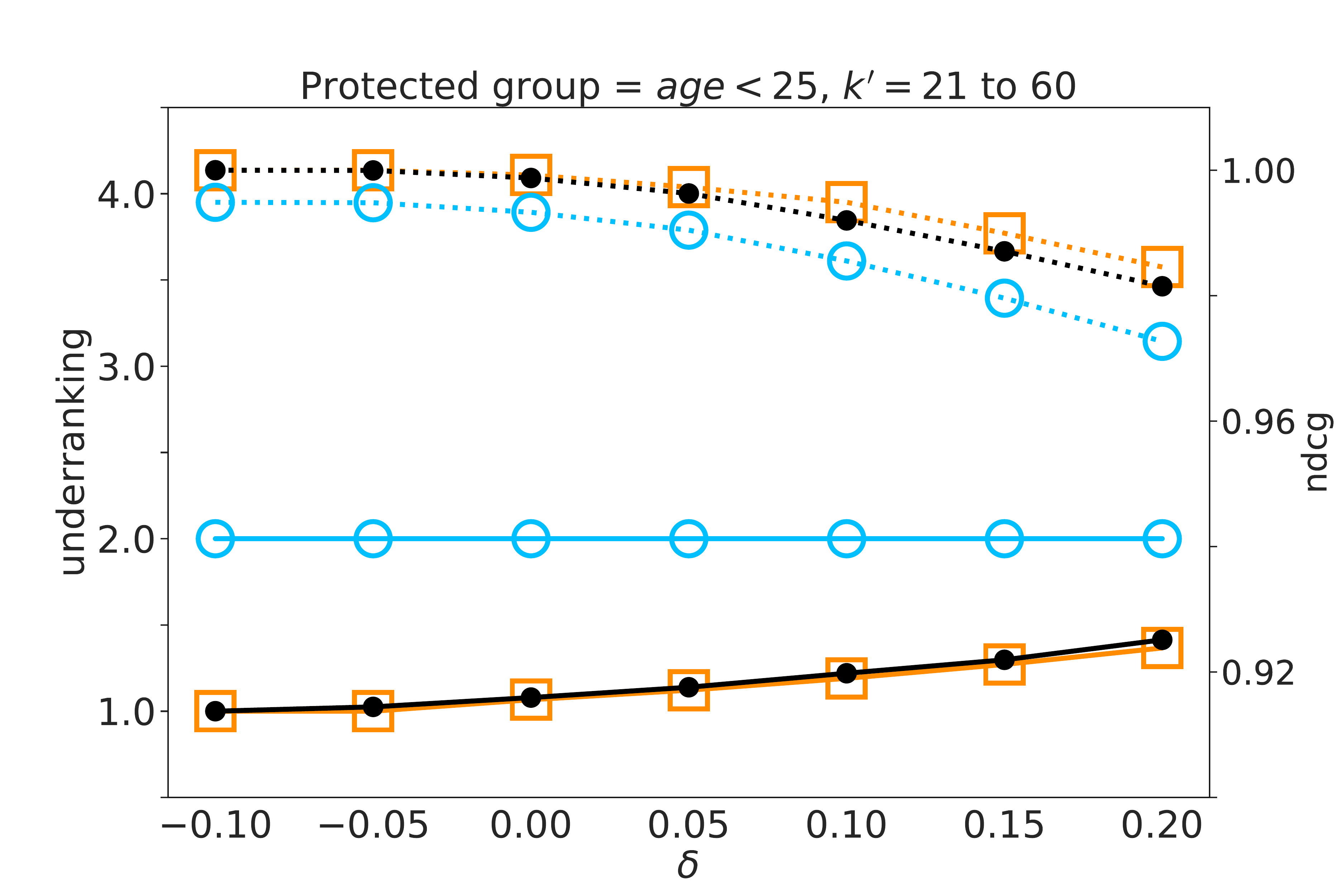} 
		\caption{Underranking, nDCG at top $40$ ranks.}
	\end{subfigure}
	\begin{subfigure}[b]{0.33\linewidth}
		\centering
		\includegraphics[scale=0.149]{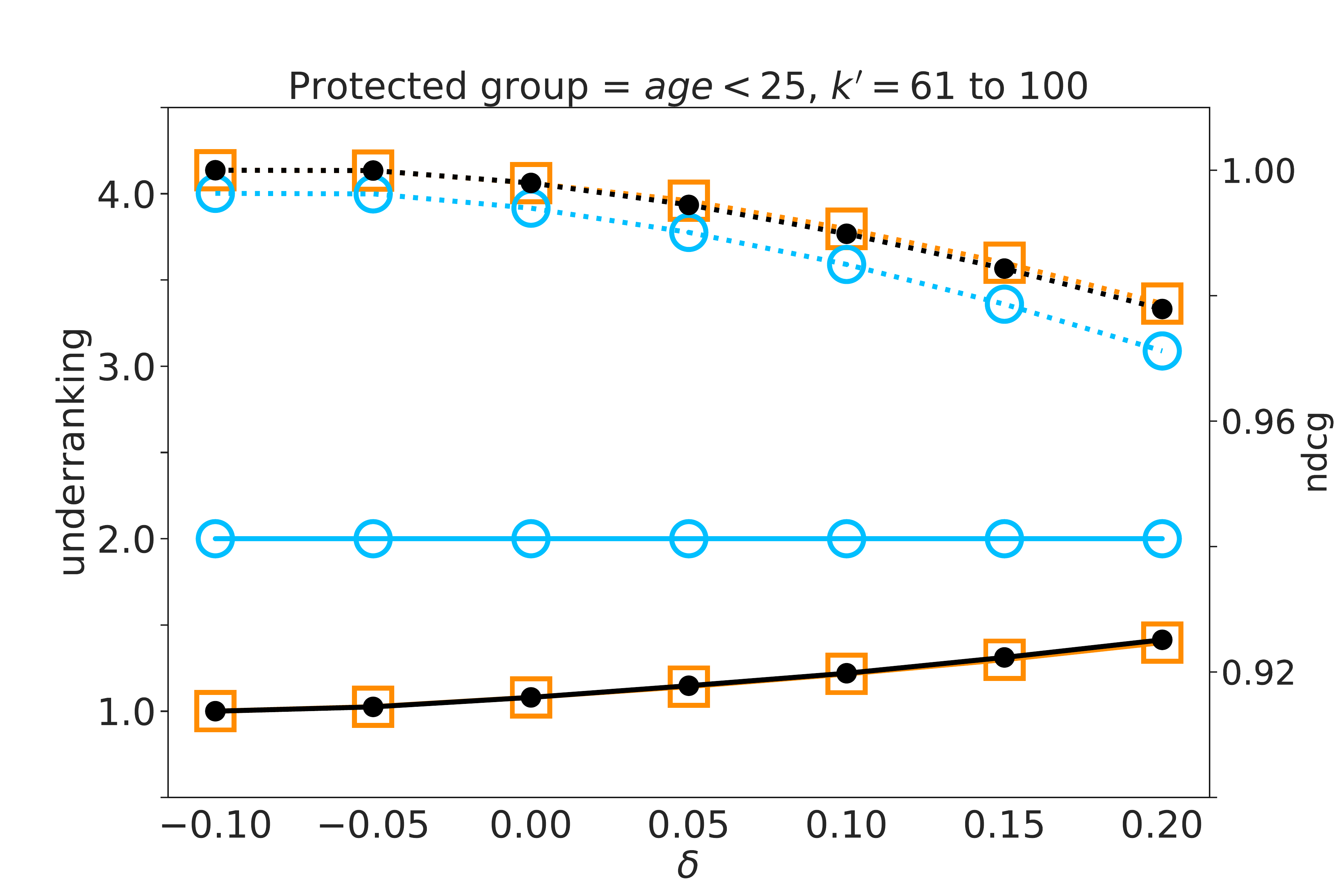} 
		\caption{Underranking, nDCG at top $100$ ranks.}
	\end{subfigure}
	\caption{Results on the German Credit Risk dataset with \textit{age}$<25$ as the protected group.}
	\label{fig:german_25_block}
\end{figure}

\begin{figure}[H]
	\begin{subfigure}[b]{\linewidth}
		\centering
		\includegraphics[scale=0.2]{results/legend.pdf} 
	\end{subfigure}
	
	\begin{subfigure}[b]{0.33\linewidth}
		\centering
		\includegraphics[scale=0.149]{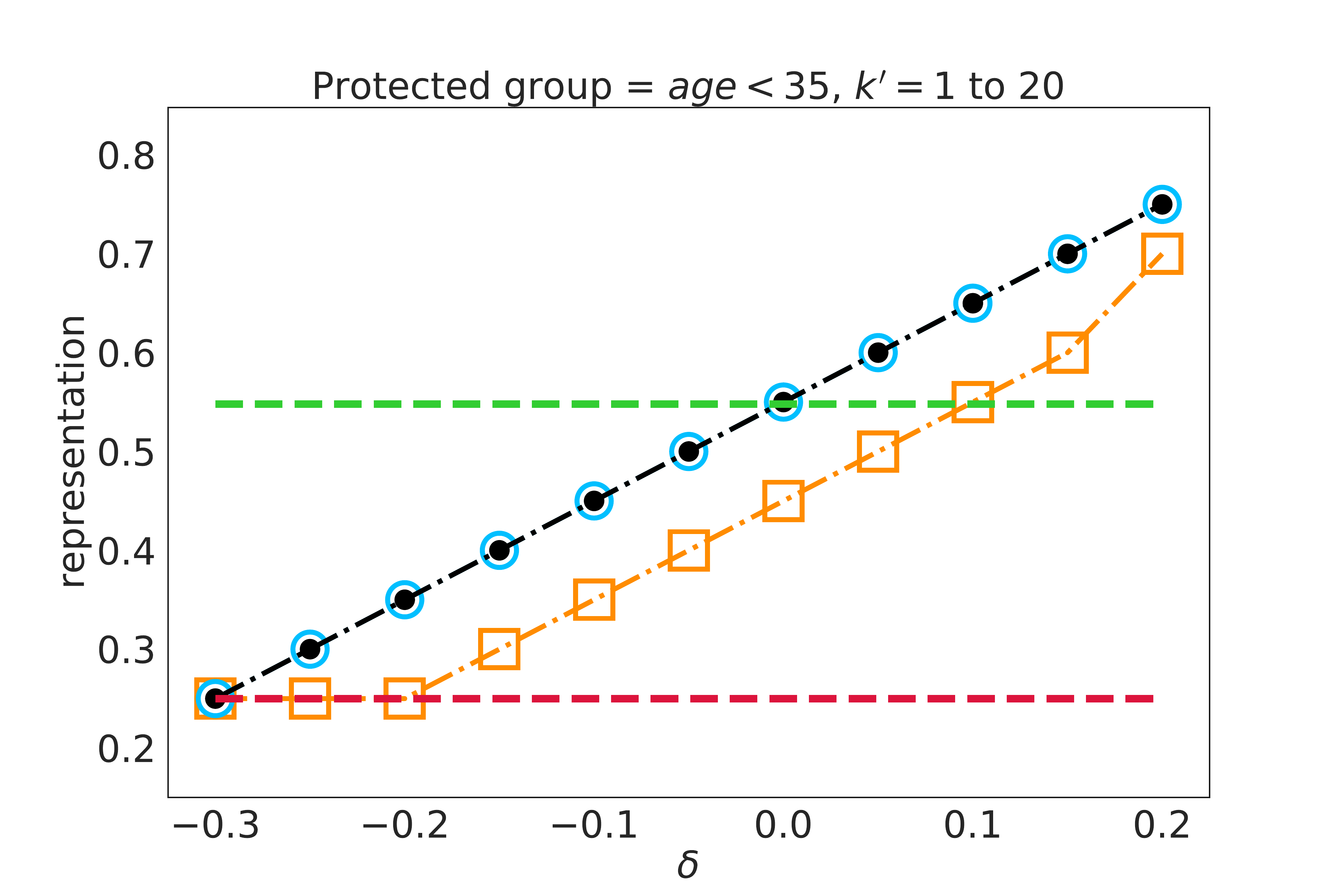} 
		\caption{Representation at ranks $1$ to $20$.}
	\end{subfigure}
	\begin{subfigure}[b]{0.33\linewidth}
		\centering
		\includegraphics[scale=0.149]{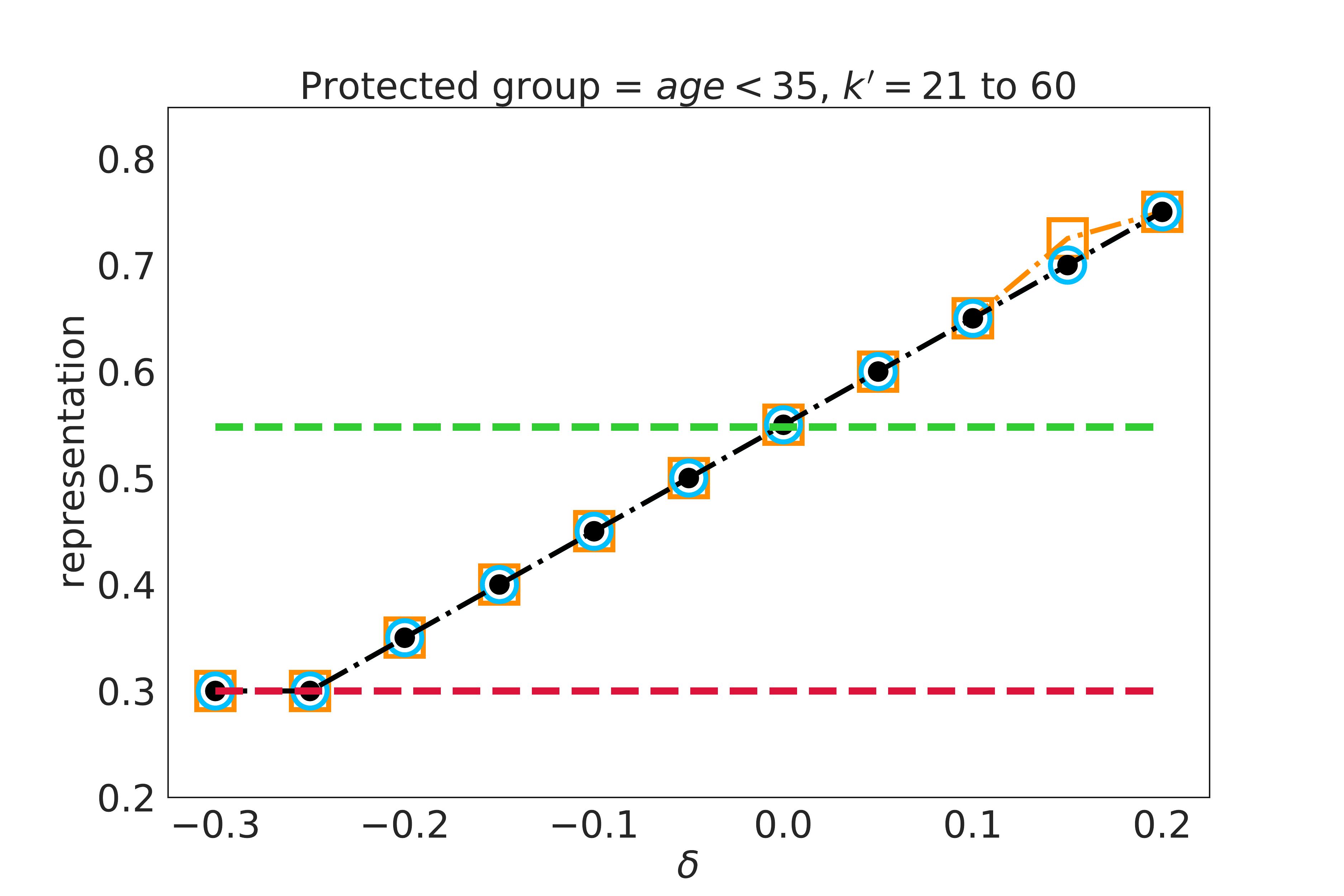} 
		\caption{Representation at ranks $21$ to $60$.}
	\end{subfigure}
	\begin{subfigure}[b]{0.33\linewidth}
		\centering
		\includegraphics[scale=0.149]{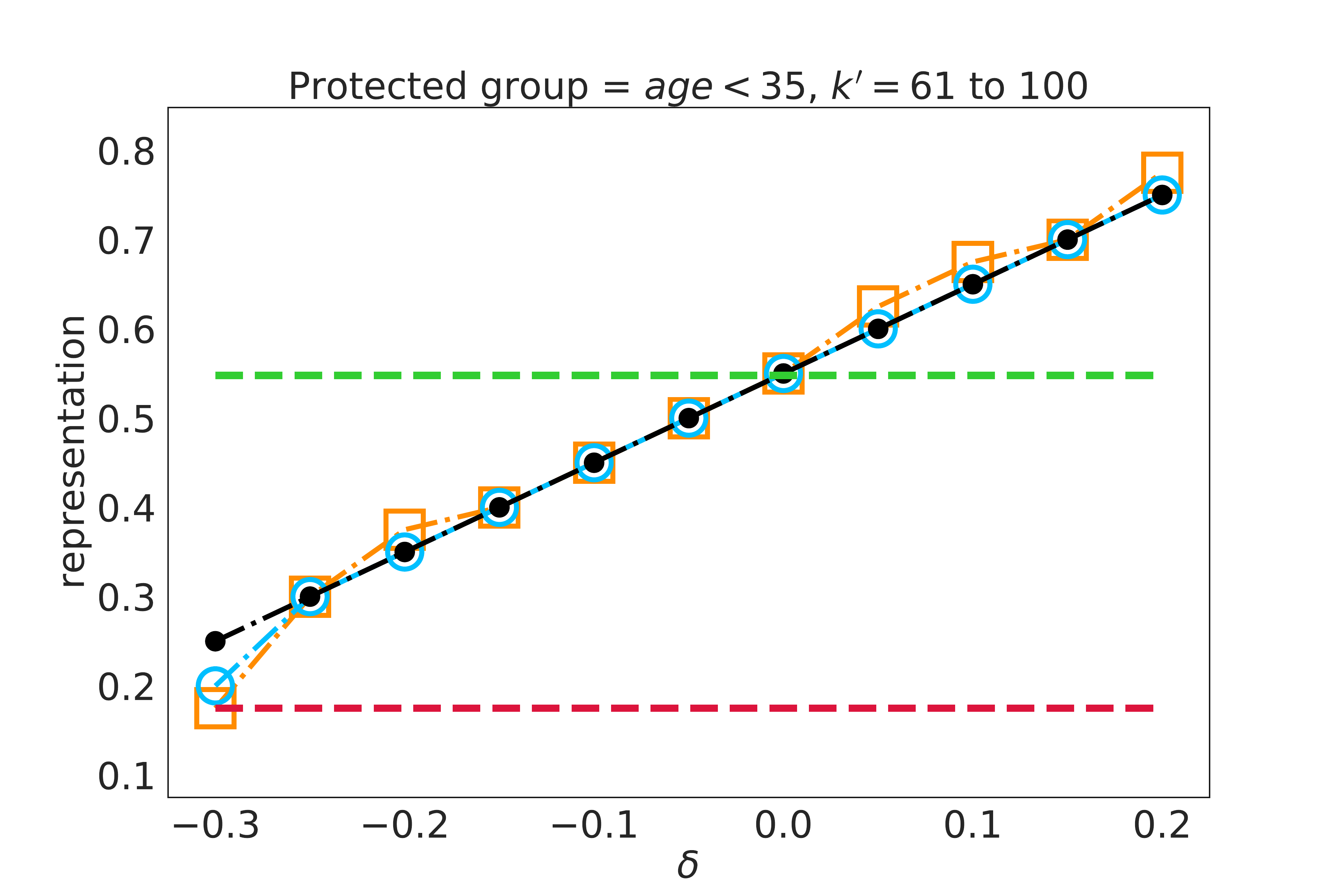} 
		\caption{Representation at ranks $61$ to $100$.}
	\end{subfigure}
	
	\begin{subfigure}[b]{0.33\linewidth}
		\centering
		\includegraphics[scale=0.149]{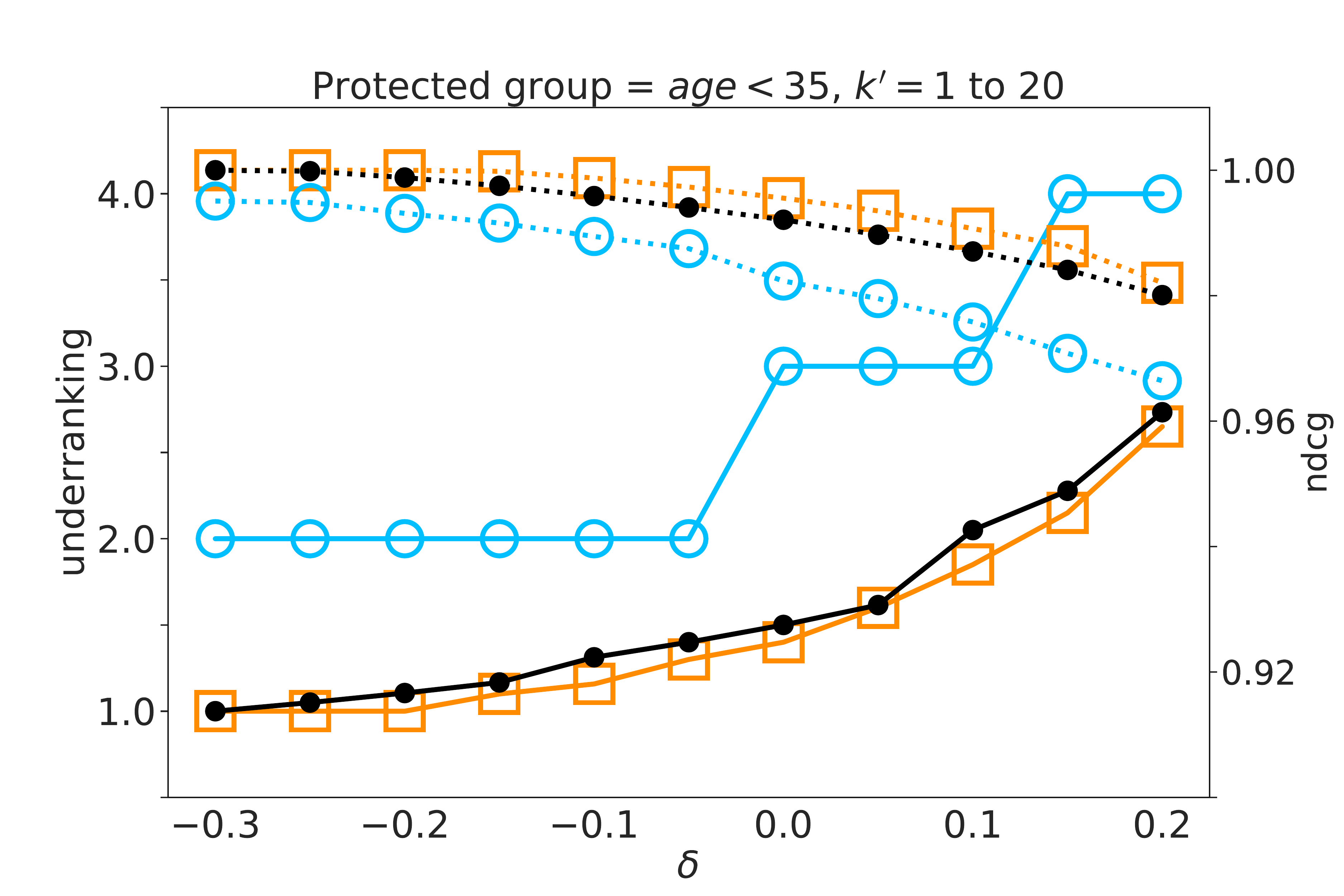} 
		\caption{Underranking, nDCG at top $20$ ranks.}
	\end{subfigure}
	\begin{subfigure}[b]{0.33\linewidth}
		\centering
		\includegraphics[scale=0.149]{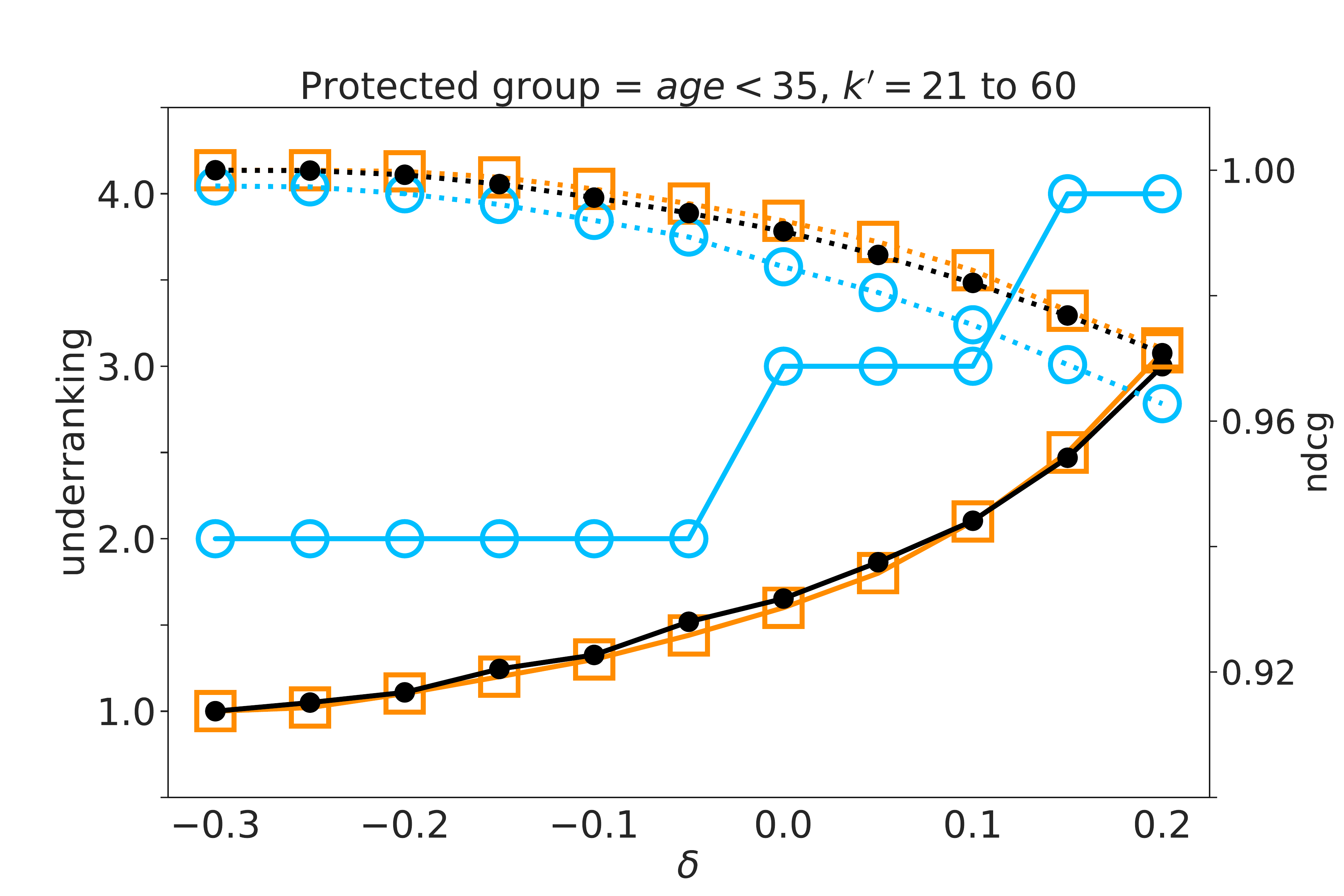} 
		\caption{Underranking, nDCG at top $60$ ranks.}
	\end{subfigure}
	\begin{subfigure}[b]{0.33\linewidth}
		\centering
		\includegraphics[scale=0.149]{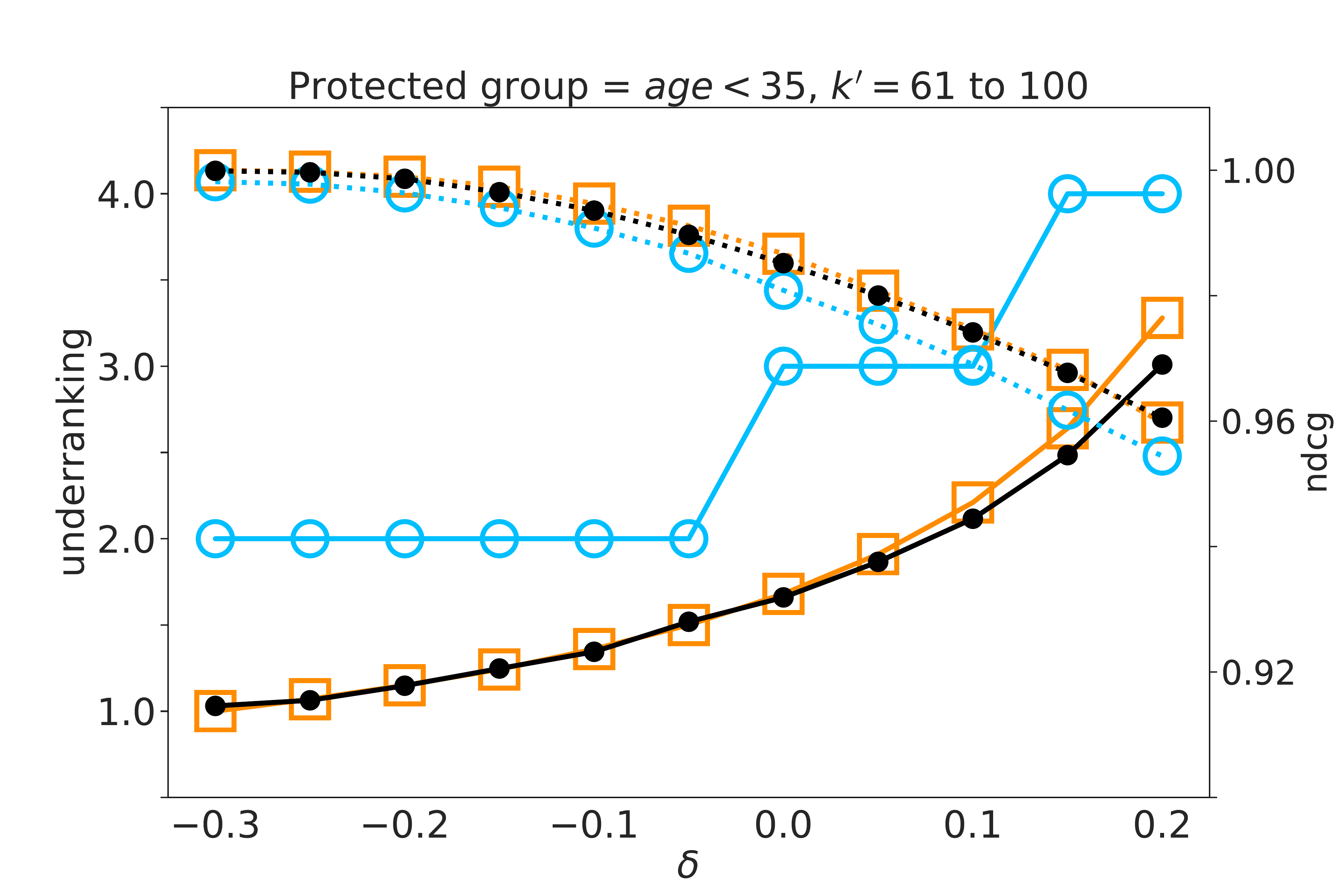} 
		\caption{Underranking, nDCG at top $100$ ranks.}
	\end{subfigure}
	\caption{Results on the German Credit Risk dataset with \textit{age}$<35$ as the protected group.}
	\label{fig:german_35_block}
\end{figure}

\begin{figure}[H]
	\begin{subfigure}[b]{\linewidth}
		\centering
		\includegraphics[scale=0.2]{results/legend.pdf} 
	\end{subfigure}
	
	\begin{subfigure}[b]{0.33\linewidth}
		\centering
		\includegraphics[scale=0.149]{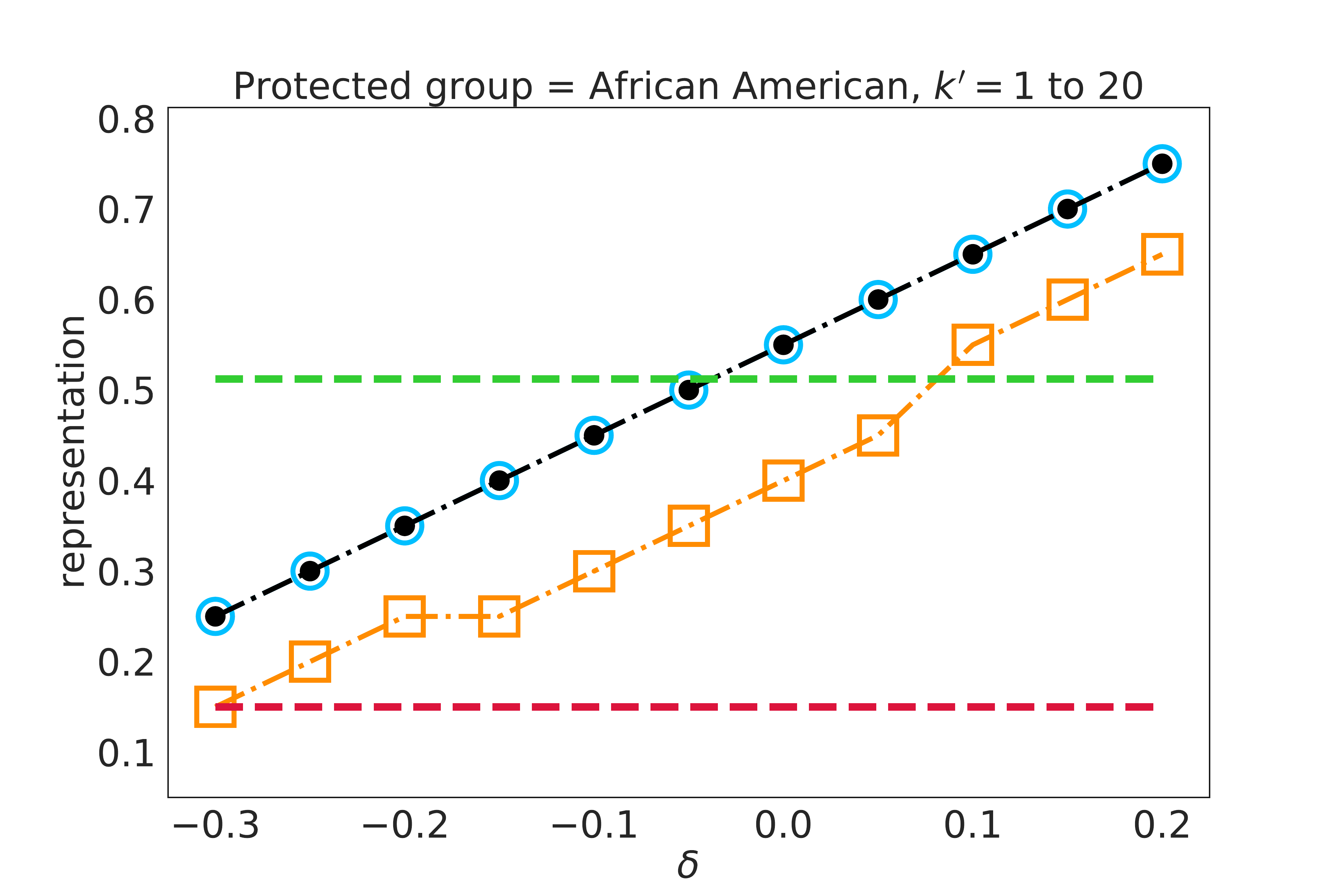} 
		\caption{Representation at ranks $1$ to $20$.}
	\end{subfigure}
	\begin{subfigure}[b]{0.33\linewidth}
		\centering
		\includegraphics[scale=0.149]{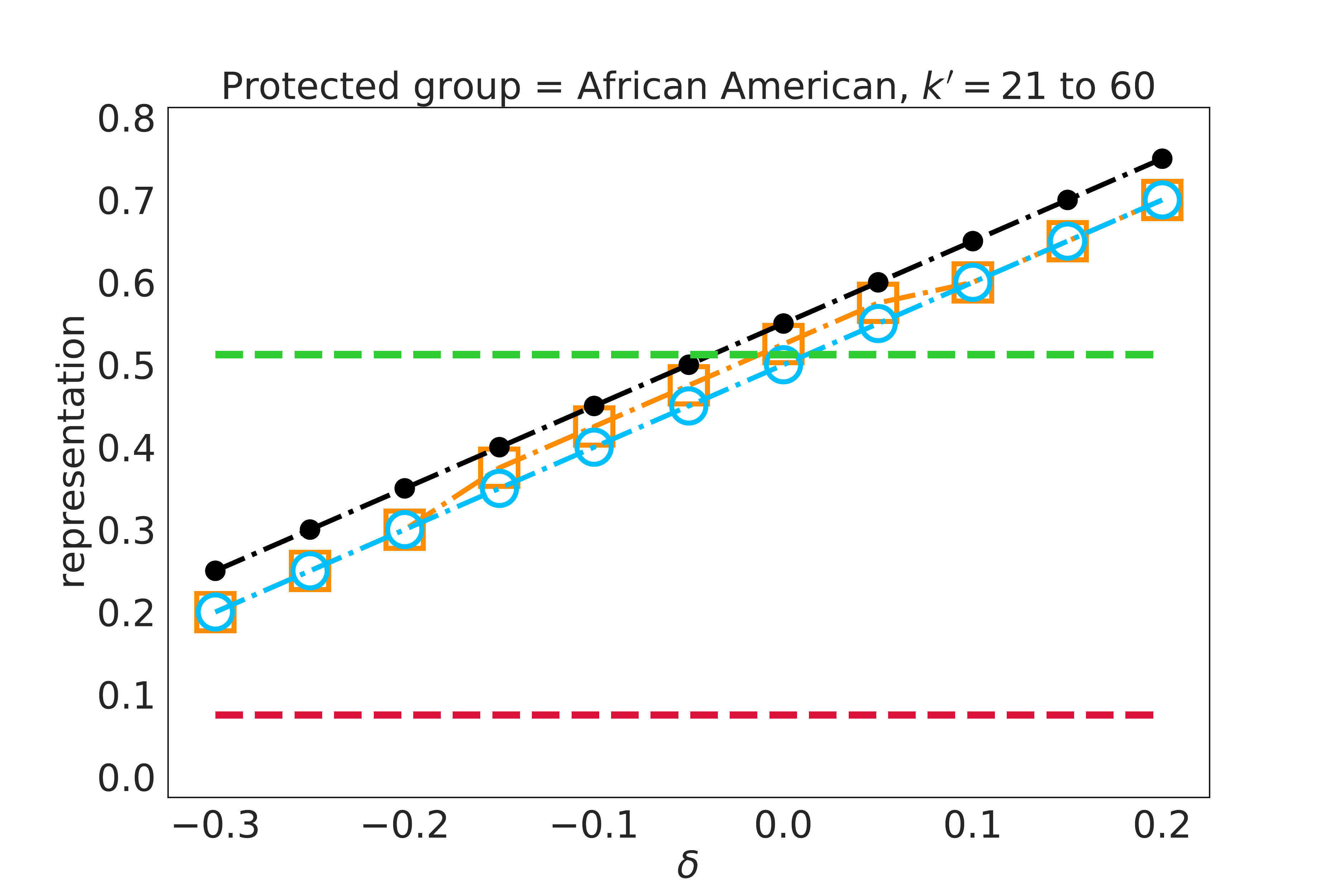} 
		\caption{Representation at ranks $21$ to $60$.}
	\end{subfigure}
	\begin{subfigure}[b]{0.33\linewidth}
		\centering
		\includegraphics[scale=0.149]{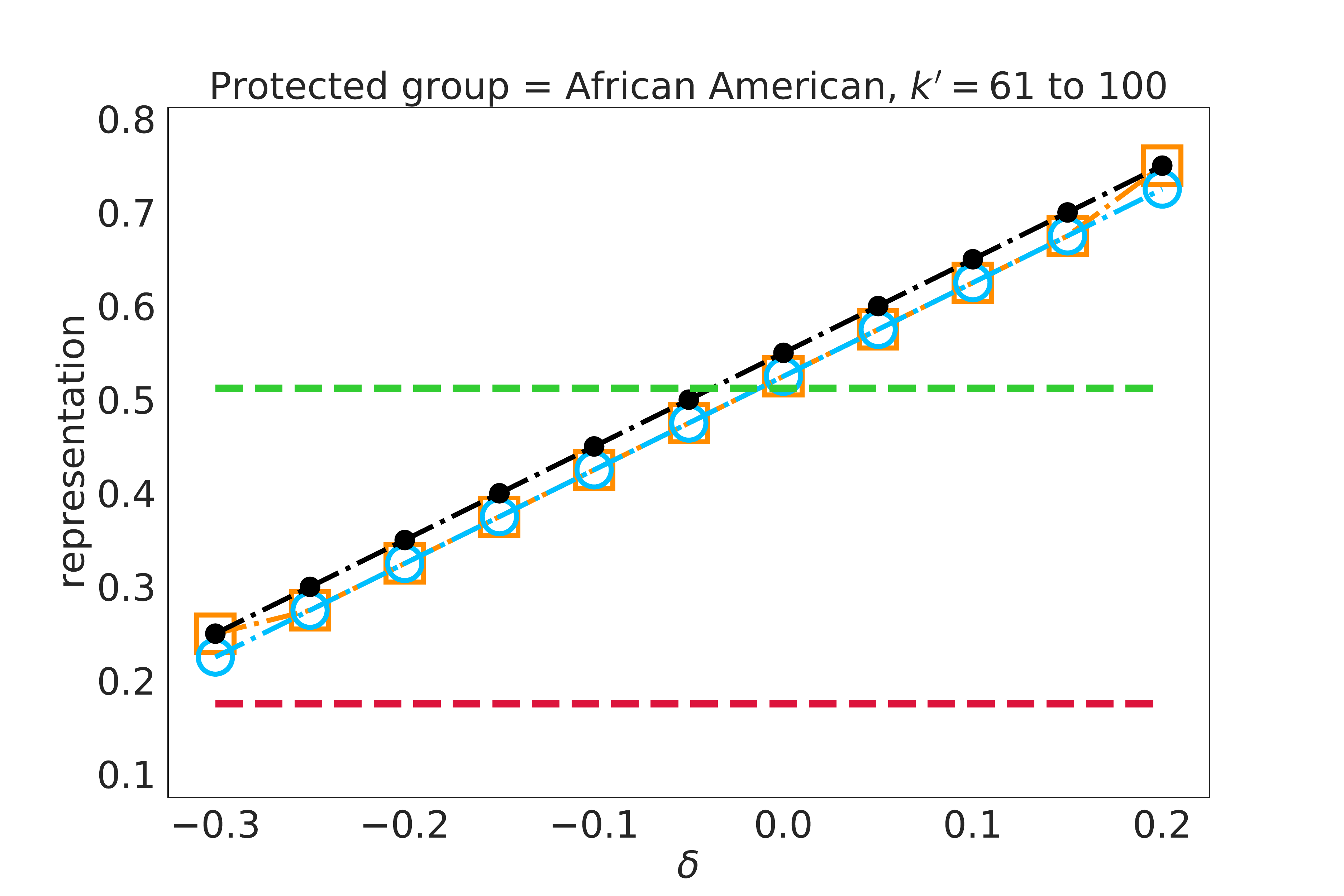} 
		\caption{Representation at ranks $61$ to $100$.}
	\end{subfigure}
	
	\begin{subfigure}[b]{0.33\linewidth}
		\centering
		\includegraphics[scale=0.149]{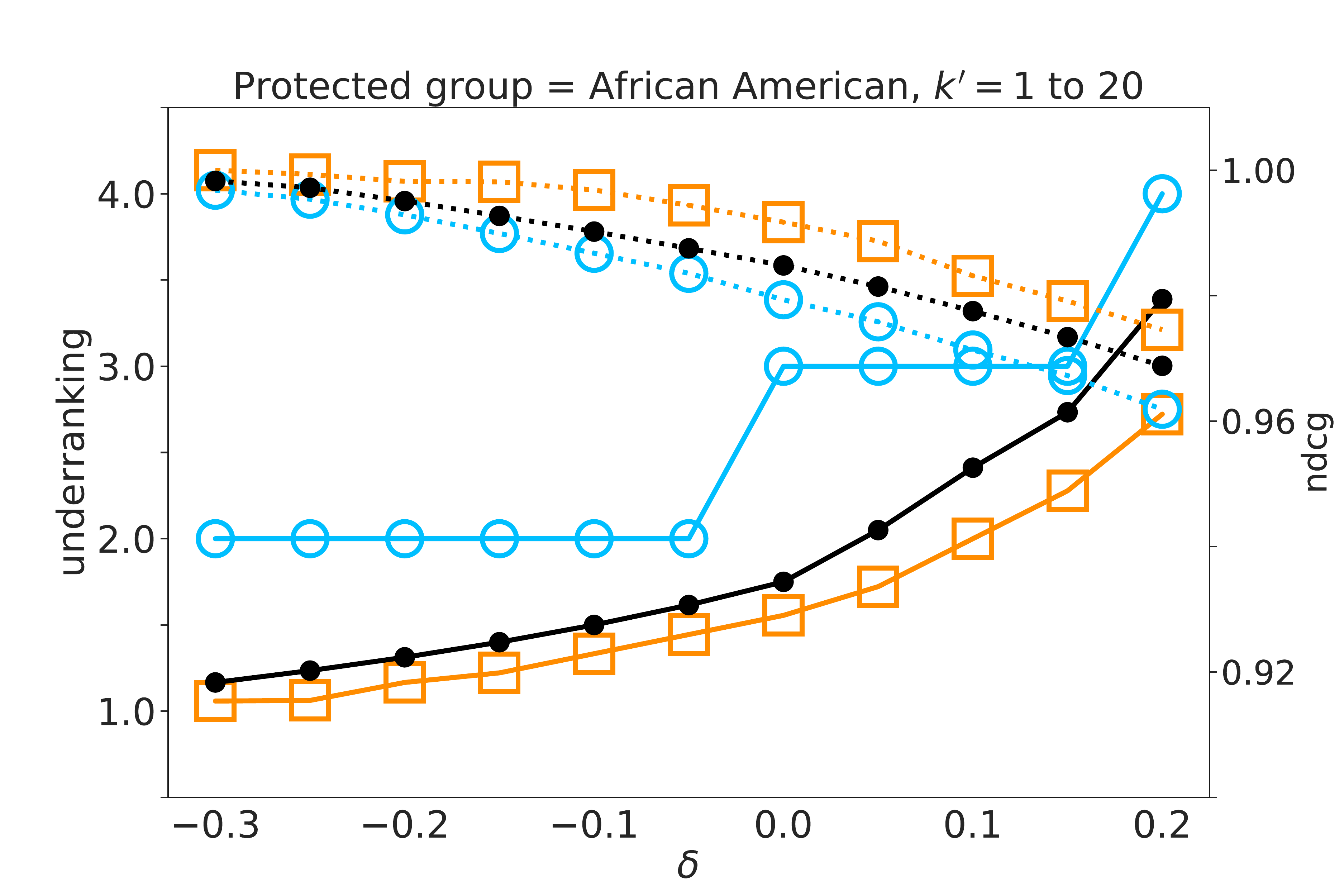} 
		\caption{Underranking, nDCG at top $20$ ranks.}
	\end{subfigure}
	\begin{subfigure}[b]{0.33\linewidth}
		\centering
		\includegraphics[scale=0.149]{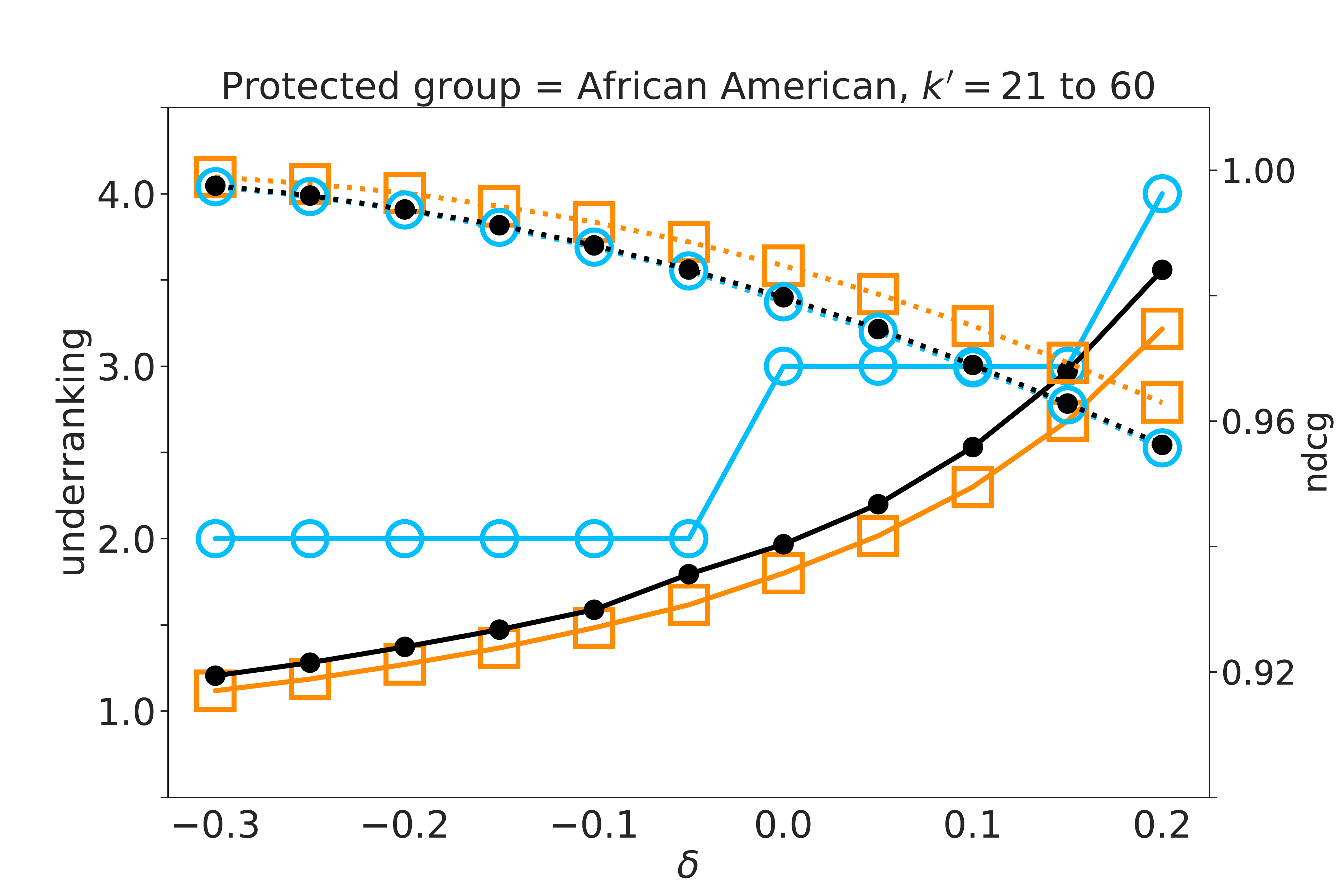} 
		\caption{Underranking, nDCG at top $60$ ranks.}
	\end{subfigure}
	\begin{subfigure}[b]{0.33\linewidth}
		\centering
		\includegraphics[scale=0.149]{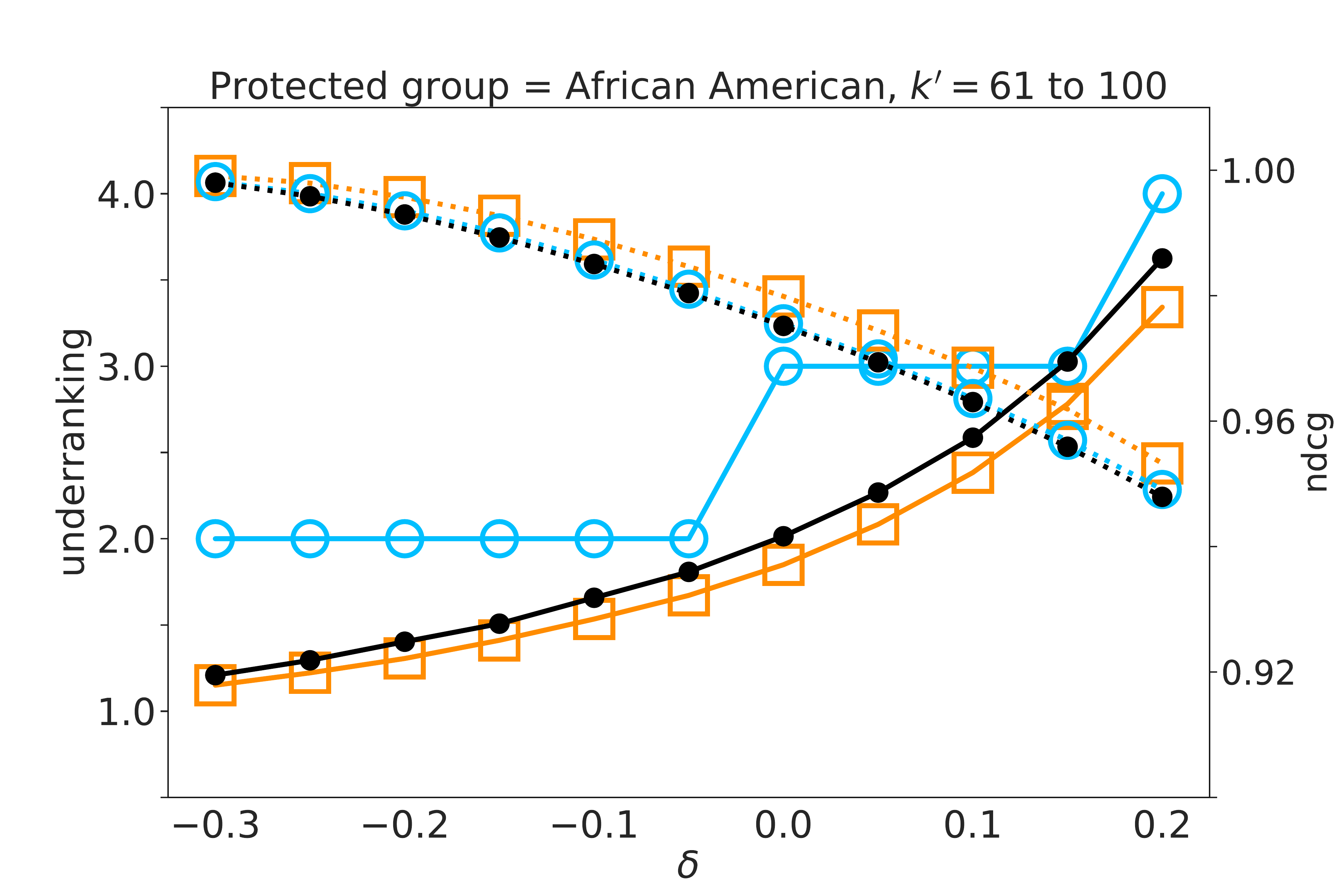} 
		\caption{Underranking, nDCG at top $100$ ranks.}
	\end{subfigure}
	\caption{Results on the COMPAS Recidivism dataset with \textit{African American} as the protected group.}
	\label{fig:compas_race_block}
\end{figure}

\begin{figure}[H]
	\begin{subfigure}[b]{\linewidth}
		\centering
		\includegraphics[scale=0.2]{results/legend.pdf} 
	\end{subfigure}
	
	\begin{subfigure}[b]{0.33\linewidth}
		\centering
		\includegraphics[scale=0.149]{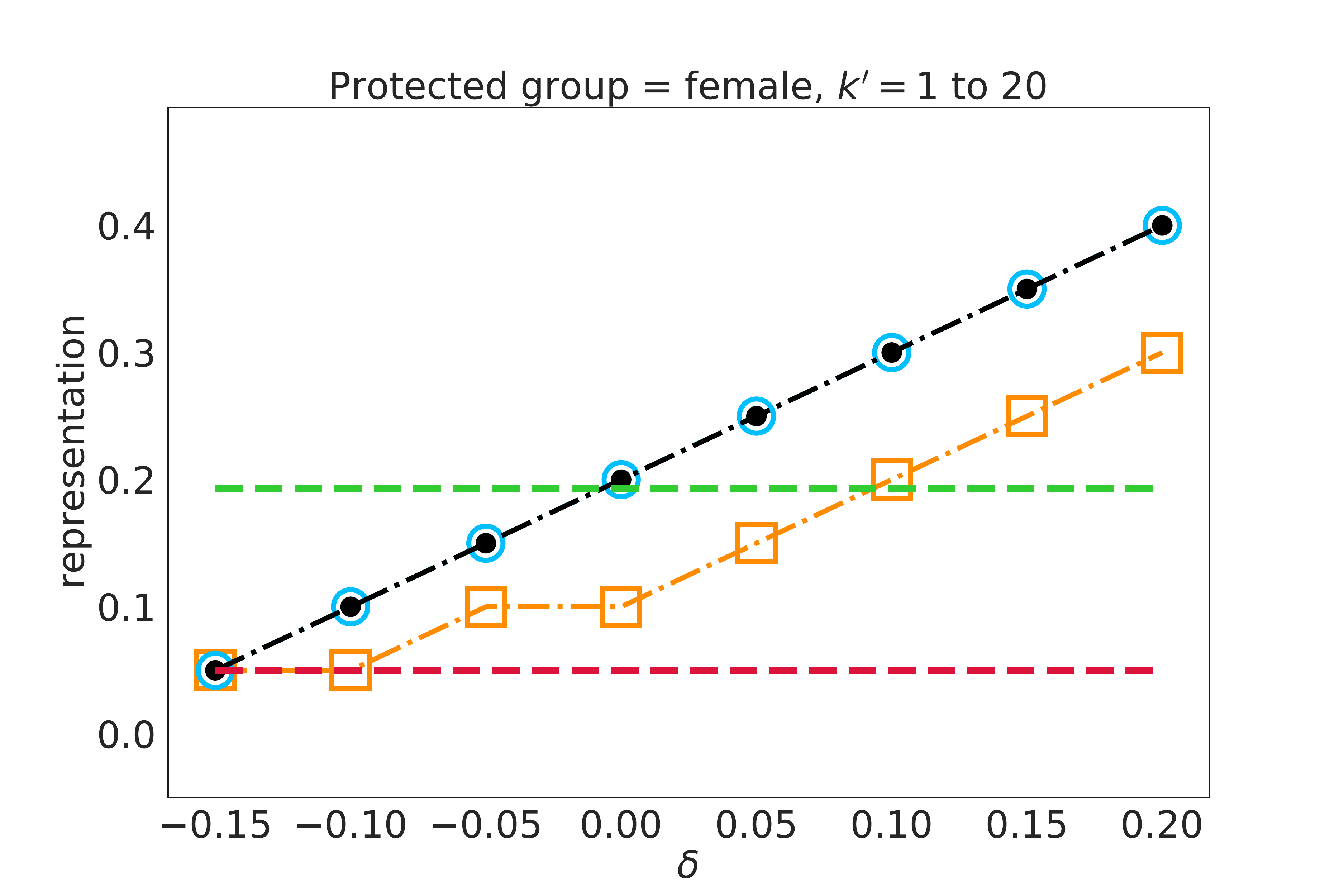} 
		\caption{Representation at ranks $1$ to $20$.}
	\end{subfigure}
	\begin{subfigure}[b]{0.33\linewidth}
		\centering
		\includegraphics[scale=0.149]{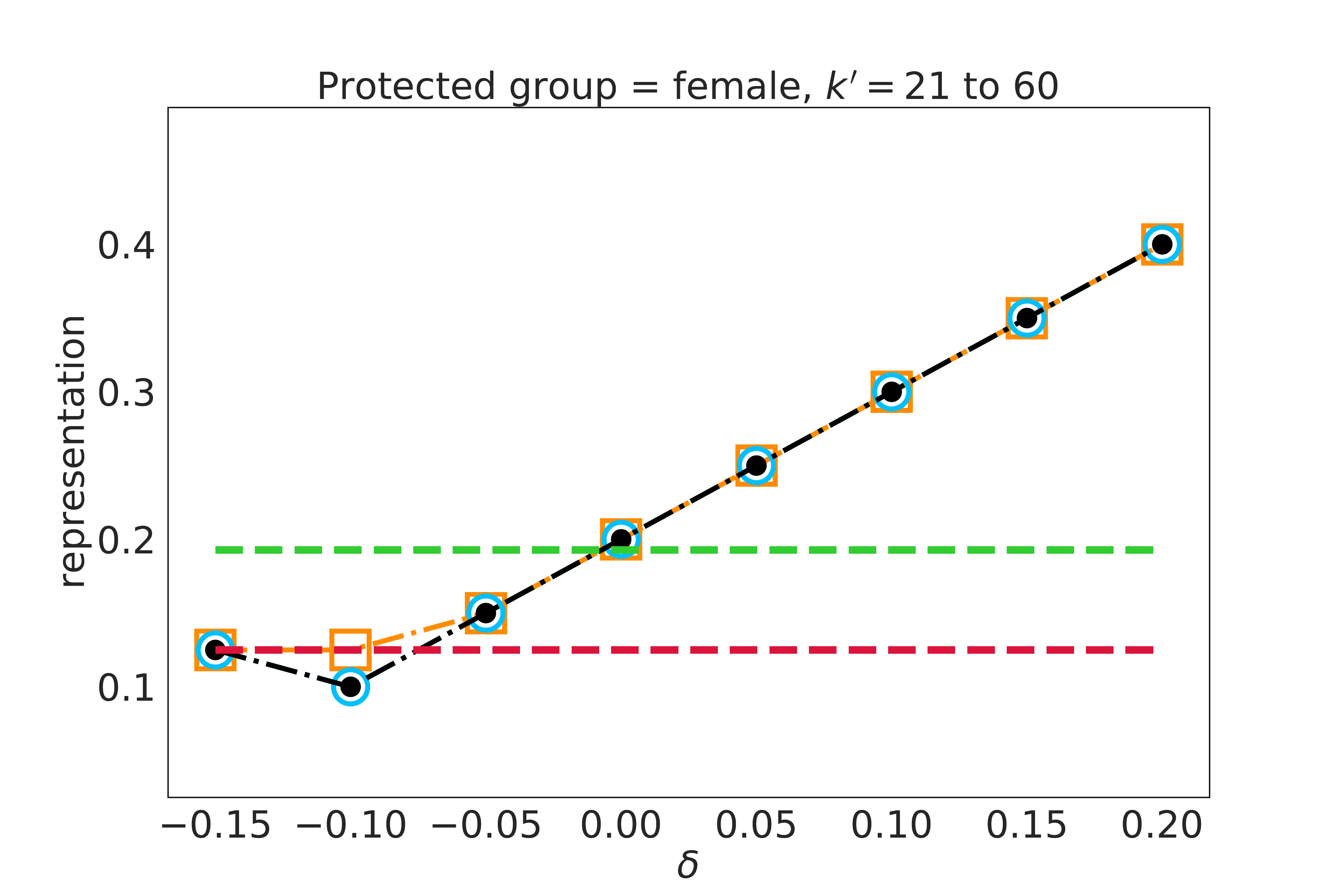} 
		\caption{Representation at ranks $21$ to $60$.}
	\end{subfigure}
	\begin{subfigure}[b]{0.33\linewidth}
		\centering
		\includegraphics[scale=0.149]{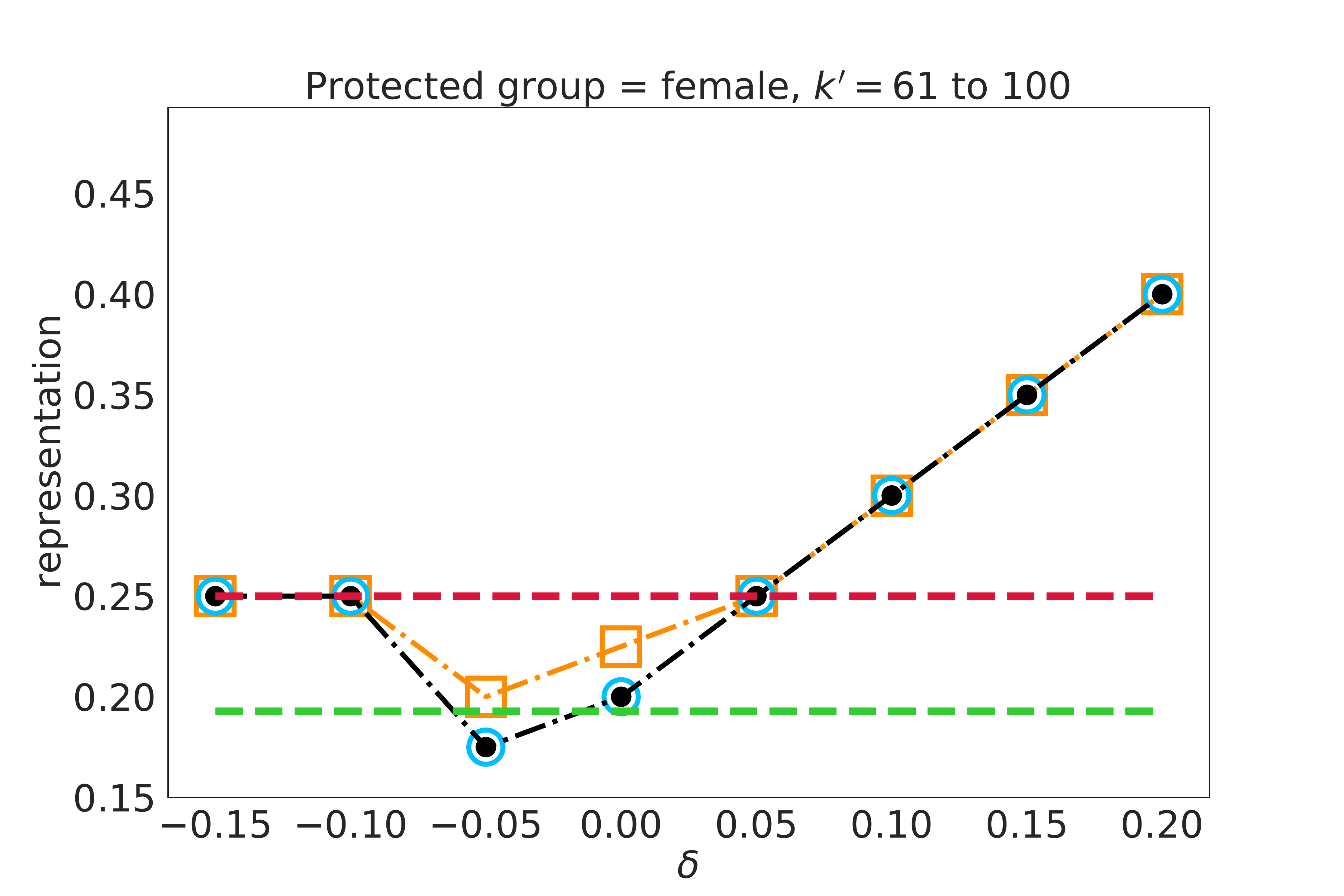} 
		\caption{Representation at ranks $61$ to $100$.}
	\end{subfigure}
	
	\begin{subfigure}[b]{0.33\linewidth}
		\centering
		\includegraphics[scale=0.149]{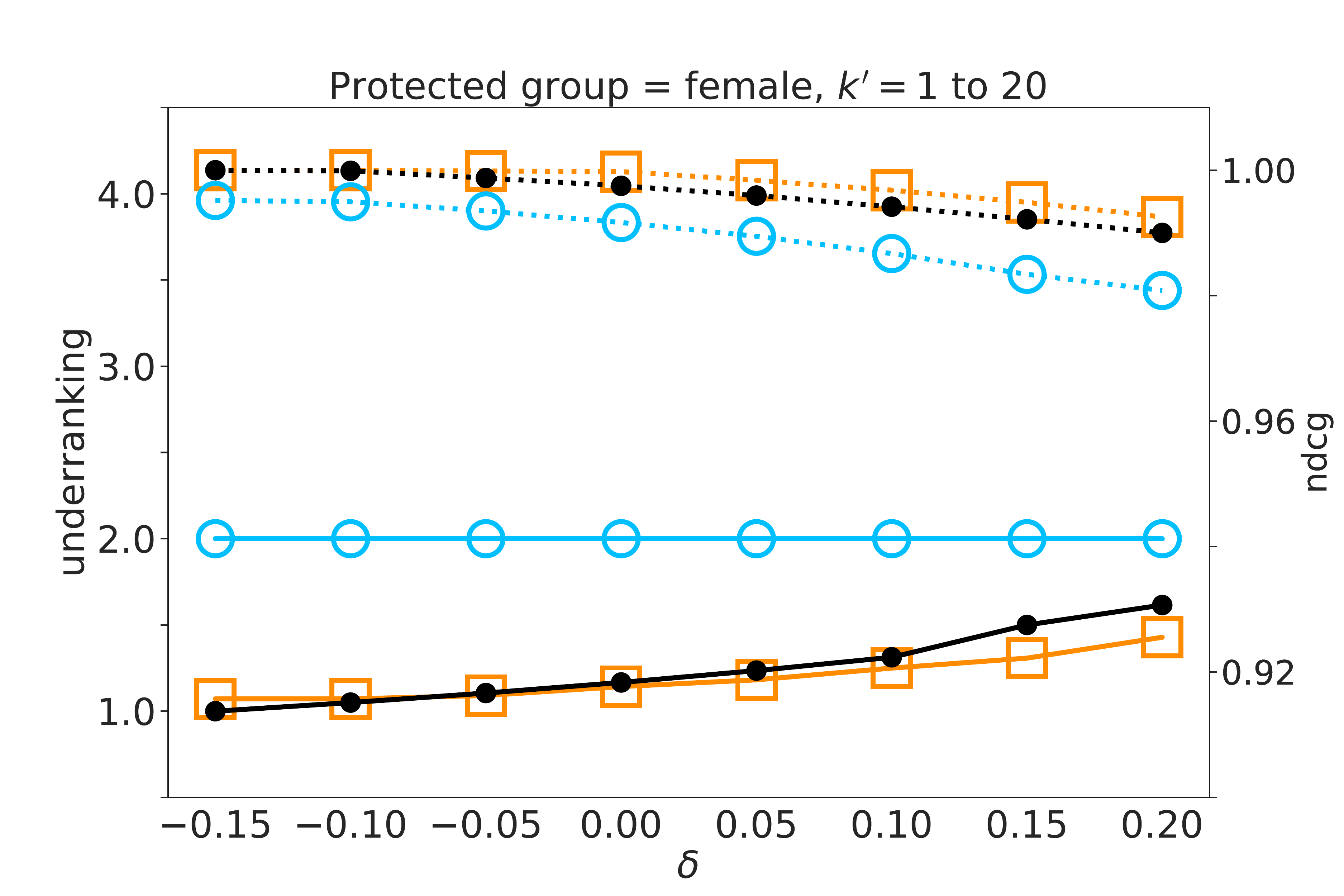} 
		\caption{Underranking, nDCG at top $20$ ranks.}
	\end{subfigure}
	\begin{subfigure}[b]{0.33\linewidth}
		\centering
		\includegraphics[scale=0.149]{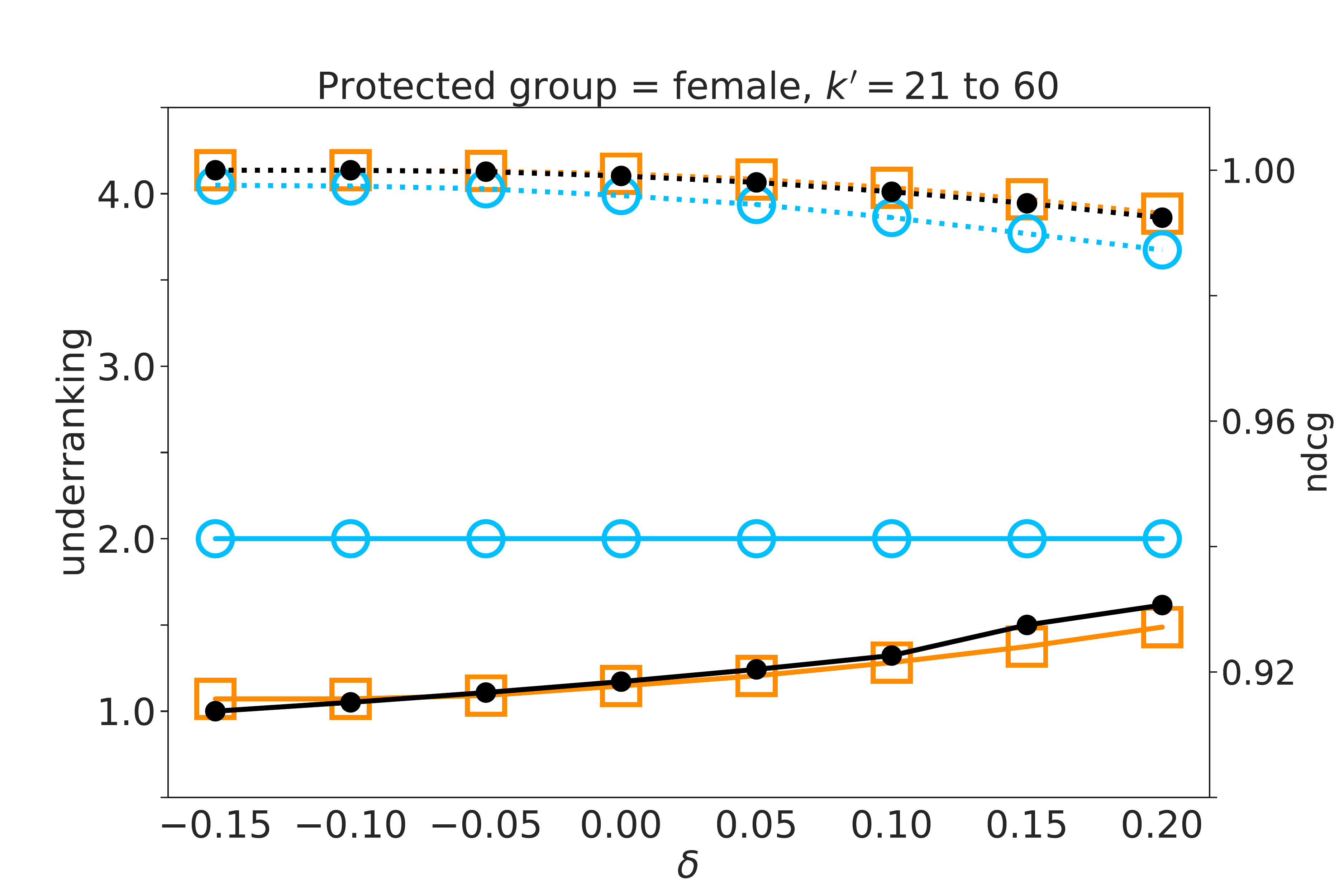} 
		\caption{Underranking, nDCG at top $60$ ranks.}
	\end{subfigure}
	\begin{subfigure}[b]{0.33\linewidth}
		\centering
		\includegraphics[scale=0.149]{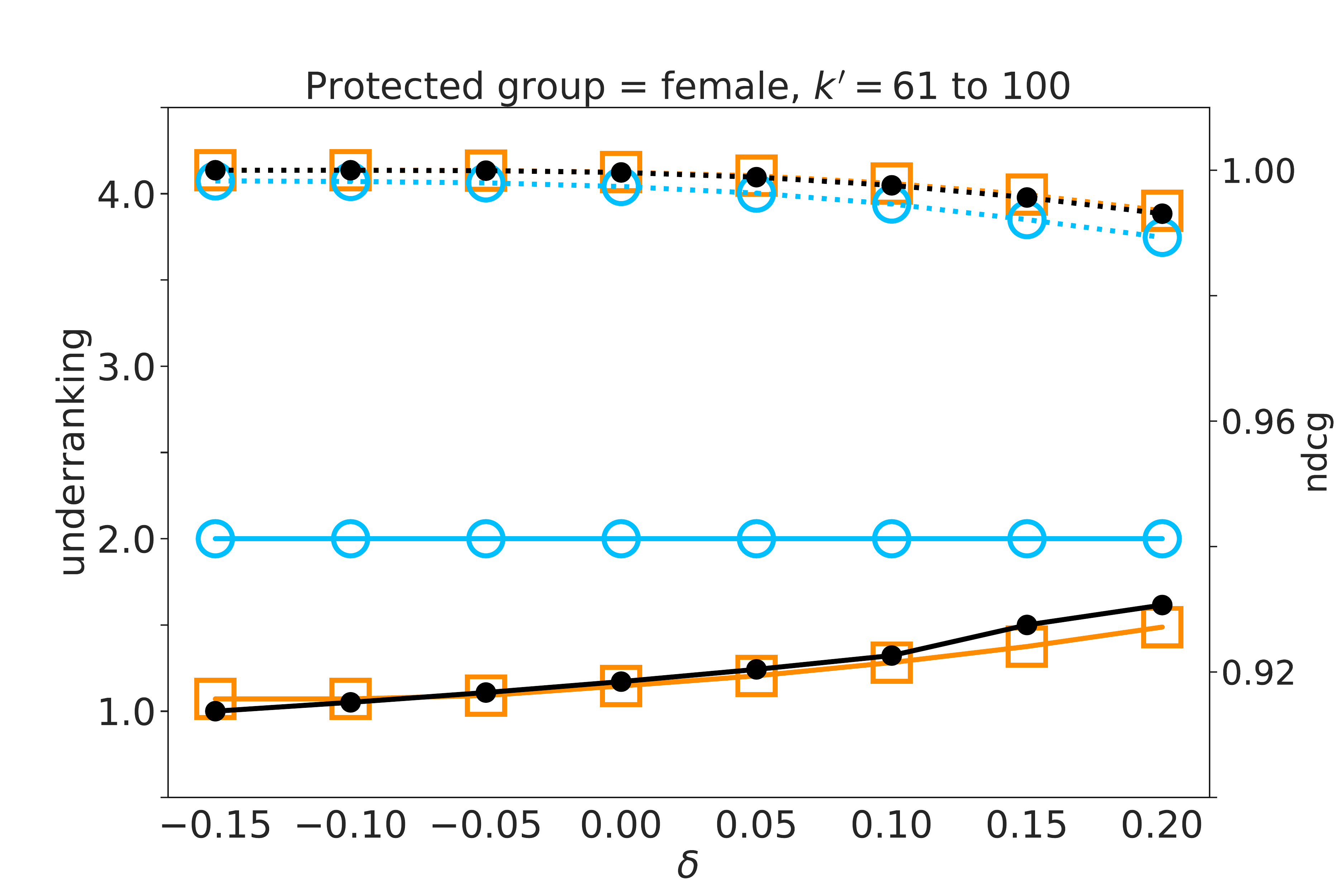} 
		\caption{Underranking, nDCG at top $100$ ranks.}
	\end{subfigure}
	\caption{Results on the COMPAS Recidivism dataset with \textit{Female} as the protected group.}
	\label{fig:compas_gender_block}
\end{figure}

\subsection{Reverse Score based Ranking as True Ranking}

In \Cref{fig:german_25_rev} to \Cref{fig:compas_race_rev}, the true ranking is based on the negative score (or relavance) of the items.
In the German Credit dataset, we use negative Schufa score for ranking, i.e., the individual with highest negative Schufa score will be ranked at the top and so on.
When using the negative scores to obtain the true ranking, we observe that the individuals from the protected group are overrepresented in the top few ranks. 
In the plots (a) - (c) in these figures, the dahsed red line (representation of the group in the top $k'$ ranks) is significantly above the dashed green line (representation of the group in the entire dataset).
In this case, we can achieve proportional representation by placing upper bound constraints on the representation of the candidates from the protected group.
We run Celis et al.'s DP algorithm with the constraints $\forall k' \in [k], L_{1, k'} = 0, L_{2, k'} = 0, U_{1, k'} = \floor{(p_1^* + \delta)\cdot k'}, U_{2, k'} = k'$ and $k = 100$, where subscript $1$ represents protected group and subscript $2$ represents non-protected group. 
We run ALG with group fairness requirements $\paren{\boldsymbol{\alpha} = (p_1^* + \delta, 1), \boldsymbol{\beta} = (0, 0), k=100}$
In this COMPAS dataset, we use recidivism risk score to obtain true ranking as opposed to the negative recidivism score in \Cref{fig:compas_race} and \Cref{fig:compas_gender}.
We observe here also that the protected groups \textit{African American} and \textit{female} is overrepresented in the top $k'$ ranks.
Hence we use upper bound constraints on these groups and all other constraints are removed.

We again observe a trade-off between group fairness and underranking and notice that in all the plots, ALG achieves better underranking than Celis et al.'s DP algorithm and also achieves very good group fairness.

Figures \ref{fig:german_25_rev_block} to \ref{fig:compas_race_rev_block} are the evaluation of the algorithms for consecutive ranks.
We again observe trade-off between representation and underranking.
ALG achieves good representation in the consecutive ranks while achieving better underranking than the baselines.

\begin{figure}[H]
	\begin{subfigure}[b]{\linewidth}
		\centering
		\includegraphics[scale=0.2]{results/legend.pdf} 
	\end{subfigure}
	
	\begin{subfigure}[b]{0.33\linewidth}
		\centering
		\includegraphics[scale=0.149]{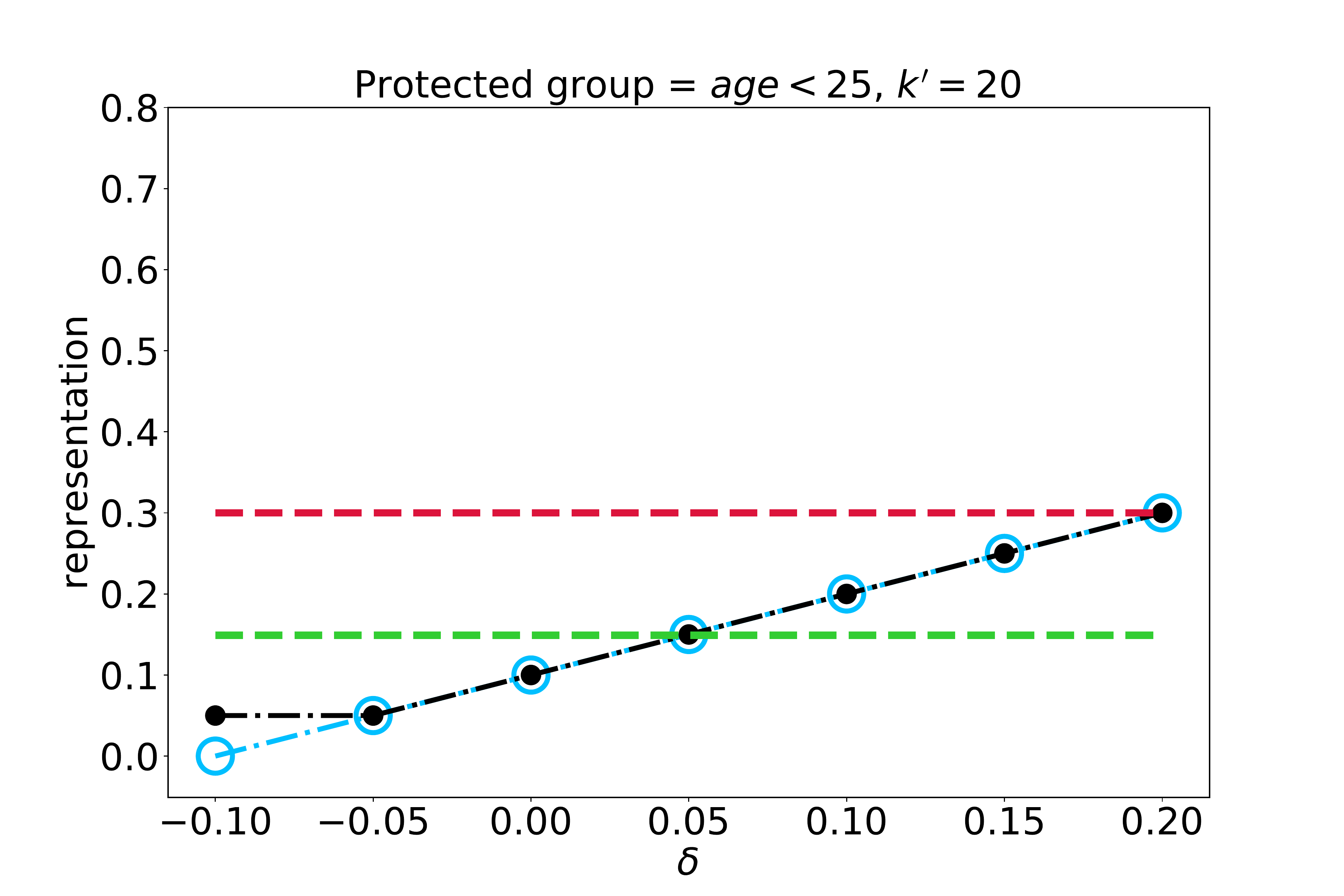} 
		\caption{Representation at top $20$ ranks.}
	\end{subfigure}
	\begin{subfigure}[b]{0.33\linewidth}
		\centering
		\includegraphics[scale=0.149]{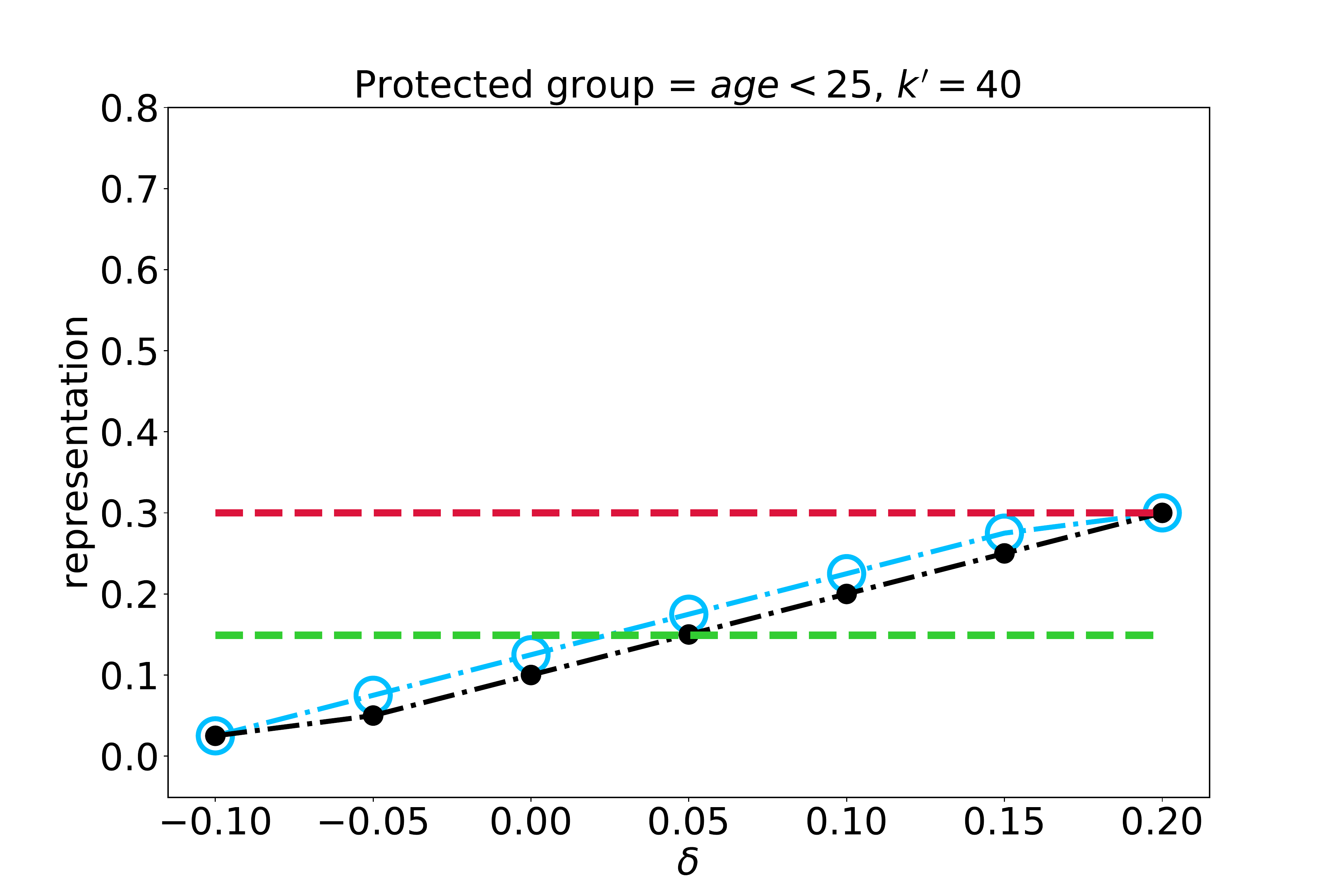} 
		\caption{Representation at top $40$ ranks.}
	\end{subfigure}
	\begin{subfigure}[b]{0.33\linewidth}
		\centering
		\includegraphics[scale=0.149]{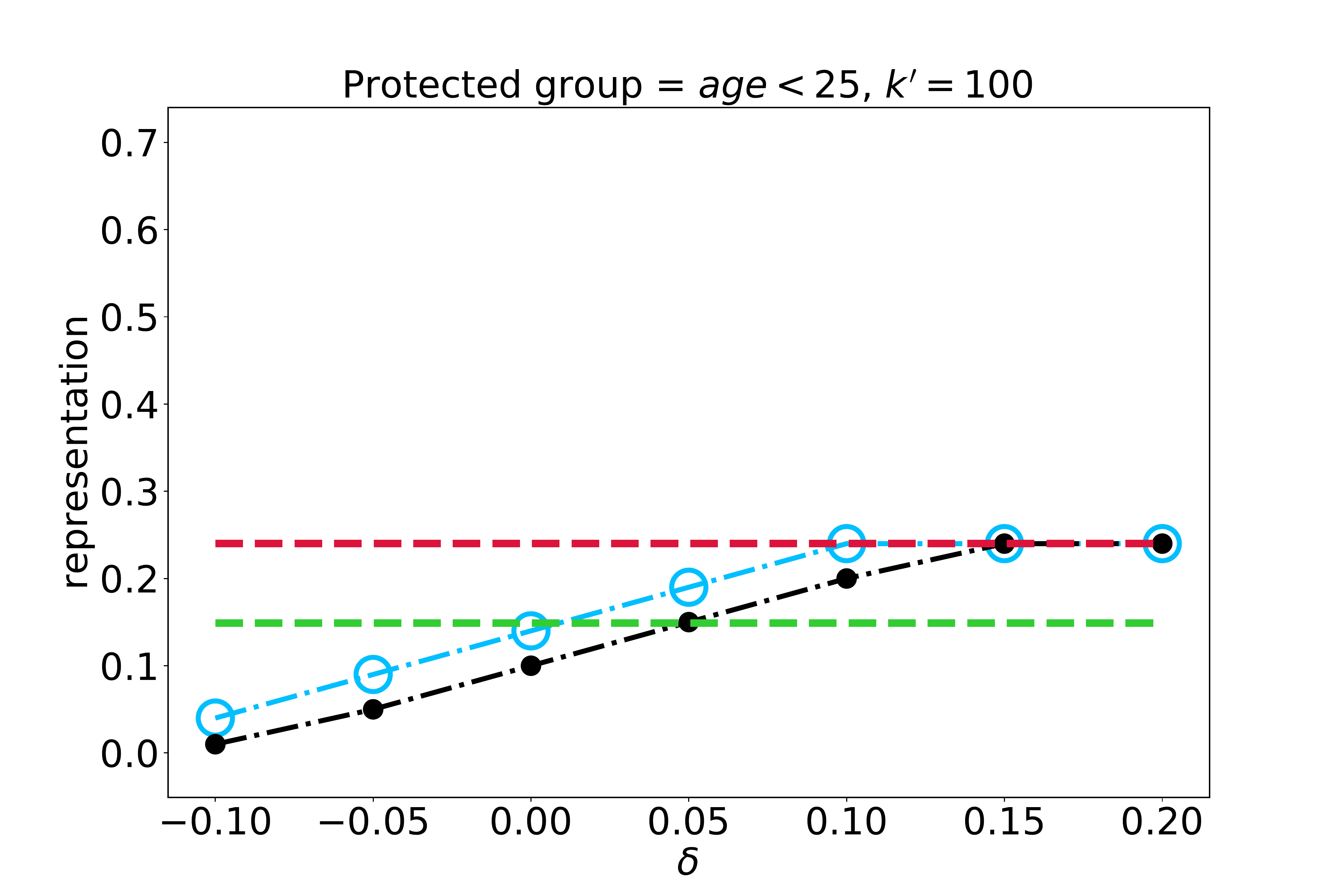} 
		\caption{Representation at top $100$ ranks.}
	\end{subfigure}
	
	\begin{subfigure}[b]{0.33\linewidth}
		\centering
		\includegraphics[scale=0.149]{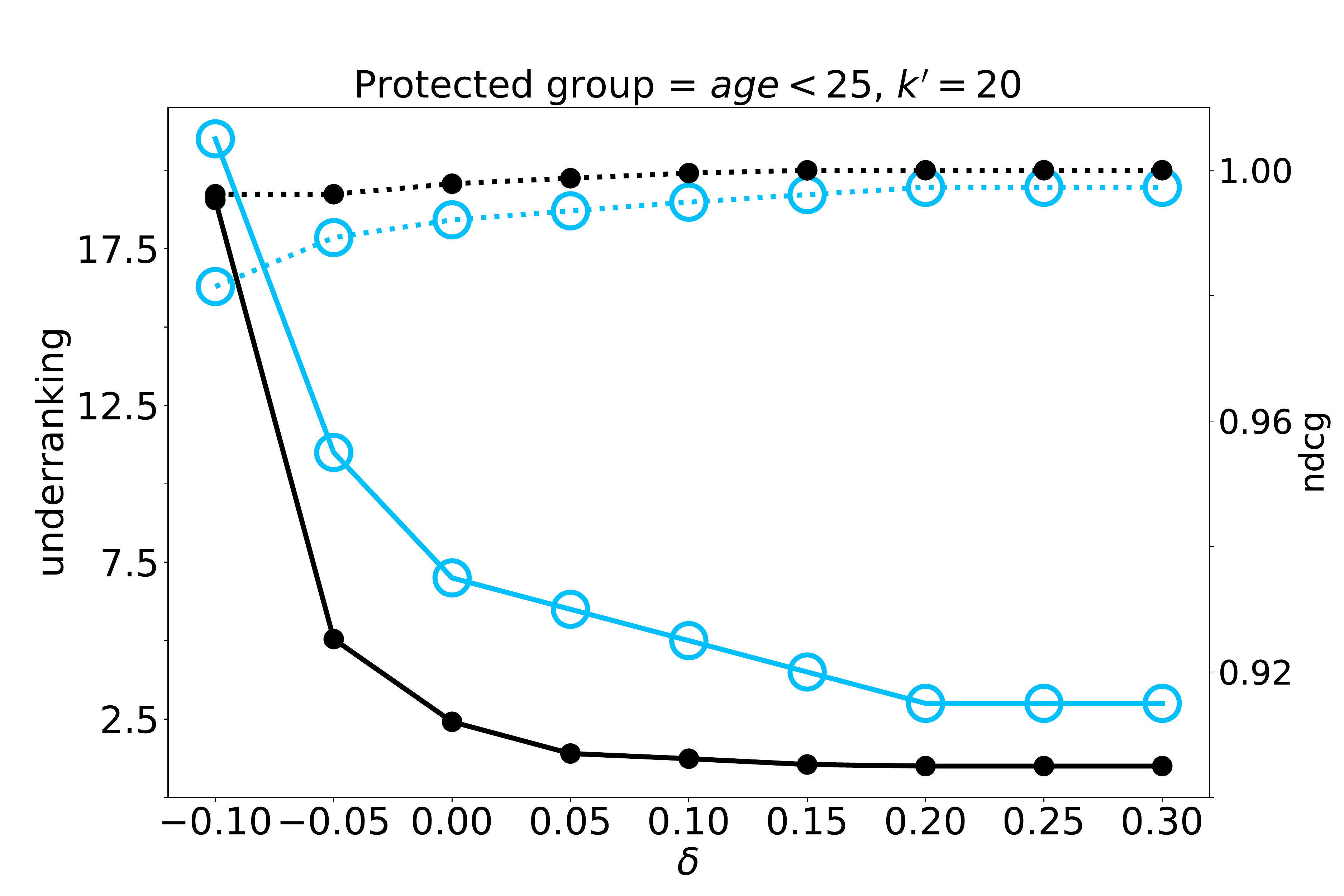} 
		\caption{Underranking, nDCG at top $20$ ranks.}
	\end{subfigure}
	\begin{subfigure}[b]{0.33\linewidth}
		\centering
		\includegraphics[scale=0.149]{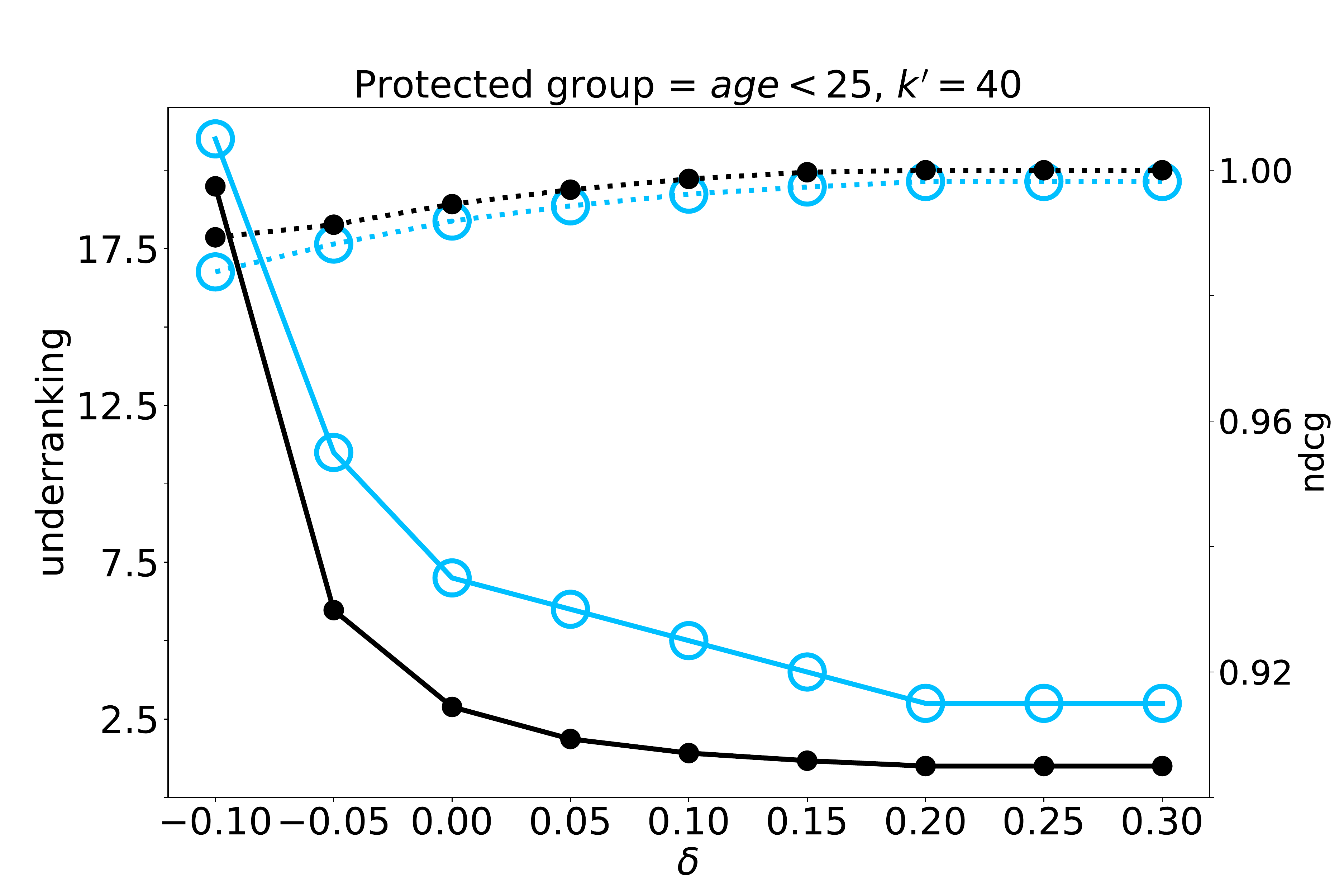} 
		\caption{Underranking, nDCG at top $40$ ranks.}
	\end{subfigure}
	\begin{subfigure}[b]{0.33\linewidth}
		\centering
		\includegraphics[scale=0.149]{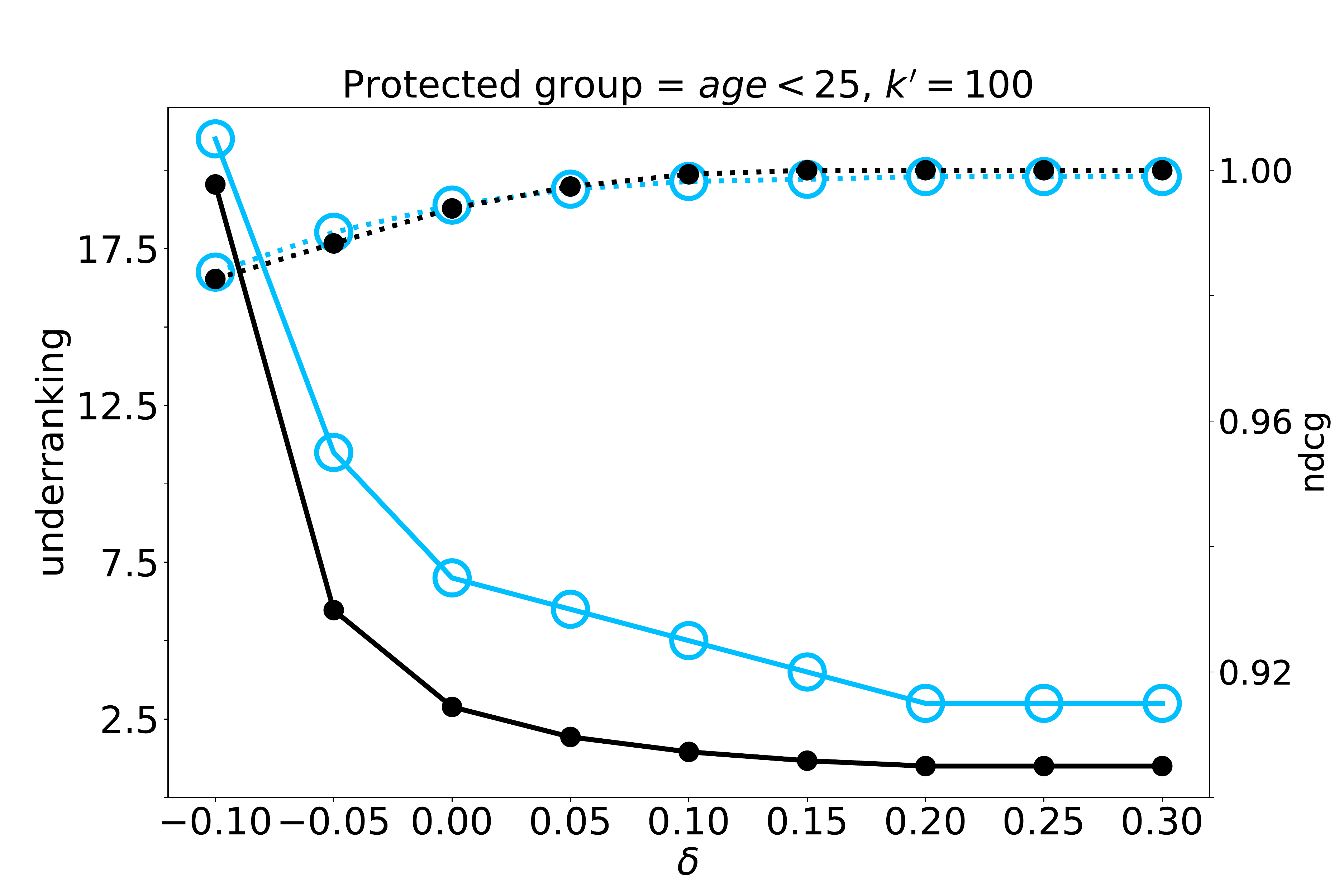} 
		\caption{Underranking, nDCG at top $100$ ranks.}
	\end{subfigure}
	\caption{Results on the German Credit Risk dataset with \textit{age}$<25$ as the protected group.}
	\label{fig:german_25_rev}
\end{figure}

\begin{figure}[H]
	\begin{subfigure}[b]{\linewidth}
		\centering
		\includegraphics[scale=0.2]{results/legend.pdf} 
	\end{subfigure}
	
	\begin{subfigure}[b]{0.33\linewidth}
		\centering
		\includegraphics[scale=0.149]{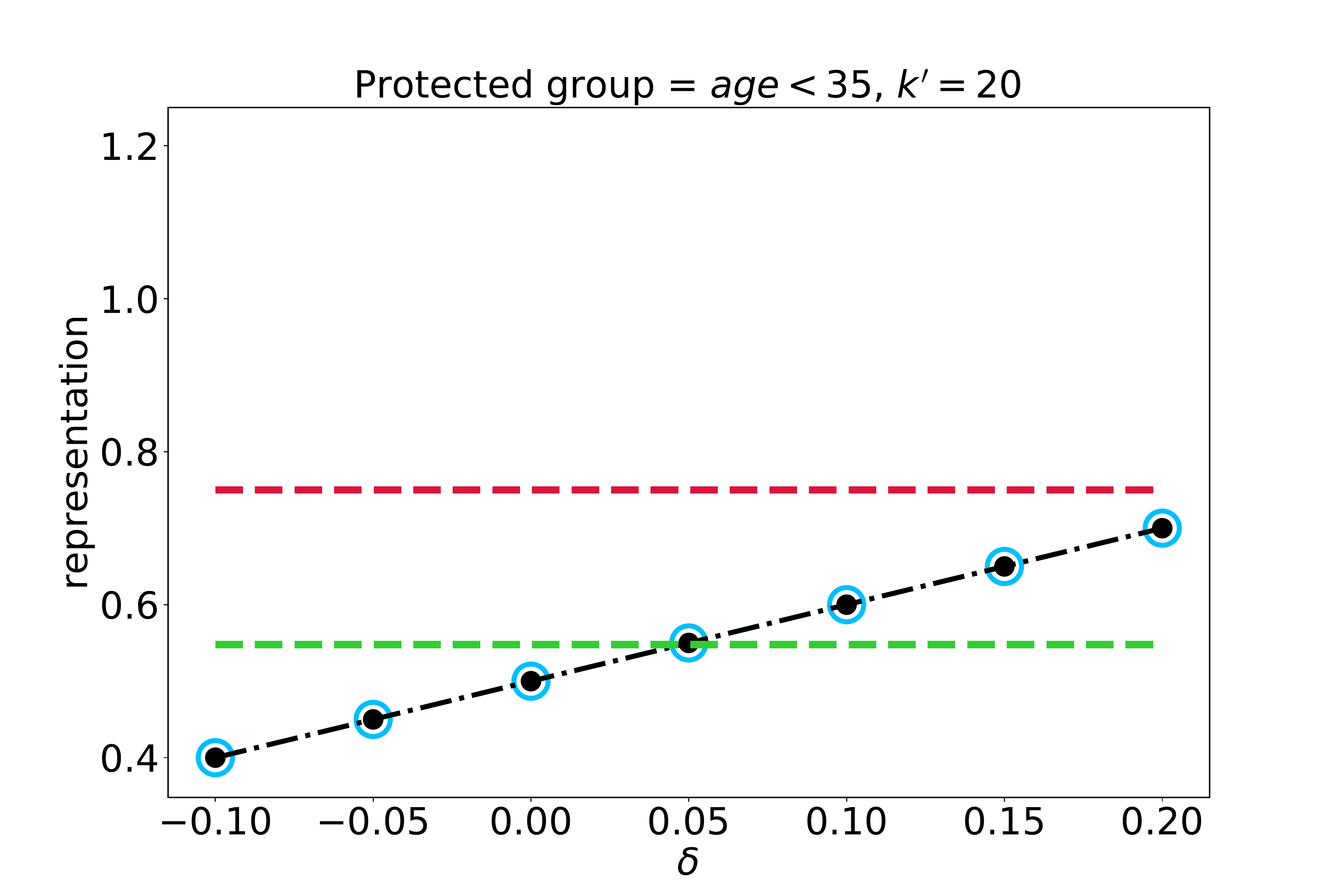} 
		\caption{Representation at top $20$ ranks.}
	\end{subfigure}
	\begin{subfigure}[b]{0.33\linewidth}
		\centering
		\includegraphics[scale=0.149]{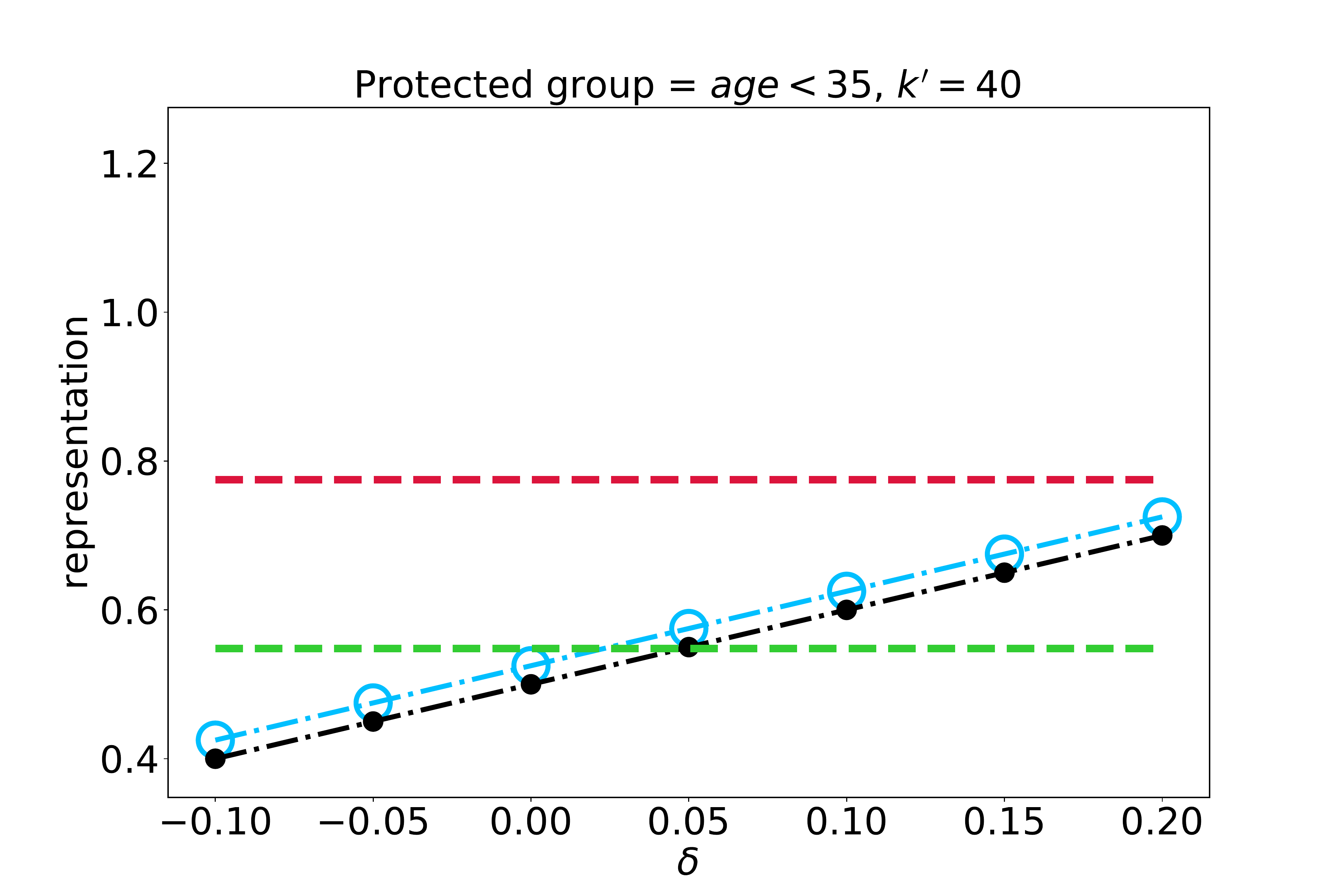} 
		\caption{Representation at top $40$ ranks.}
	\end{subfigure}
	\begin{subfigure}[b]{0.33\linewidth}
		\centering
		\includegraphics[scale=0.149]{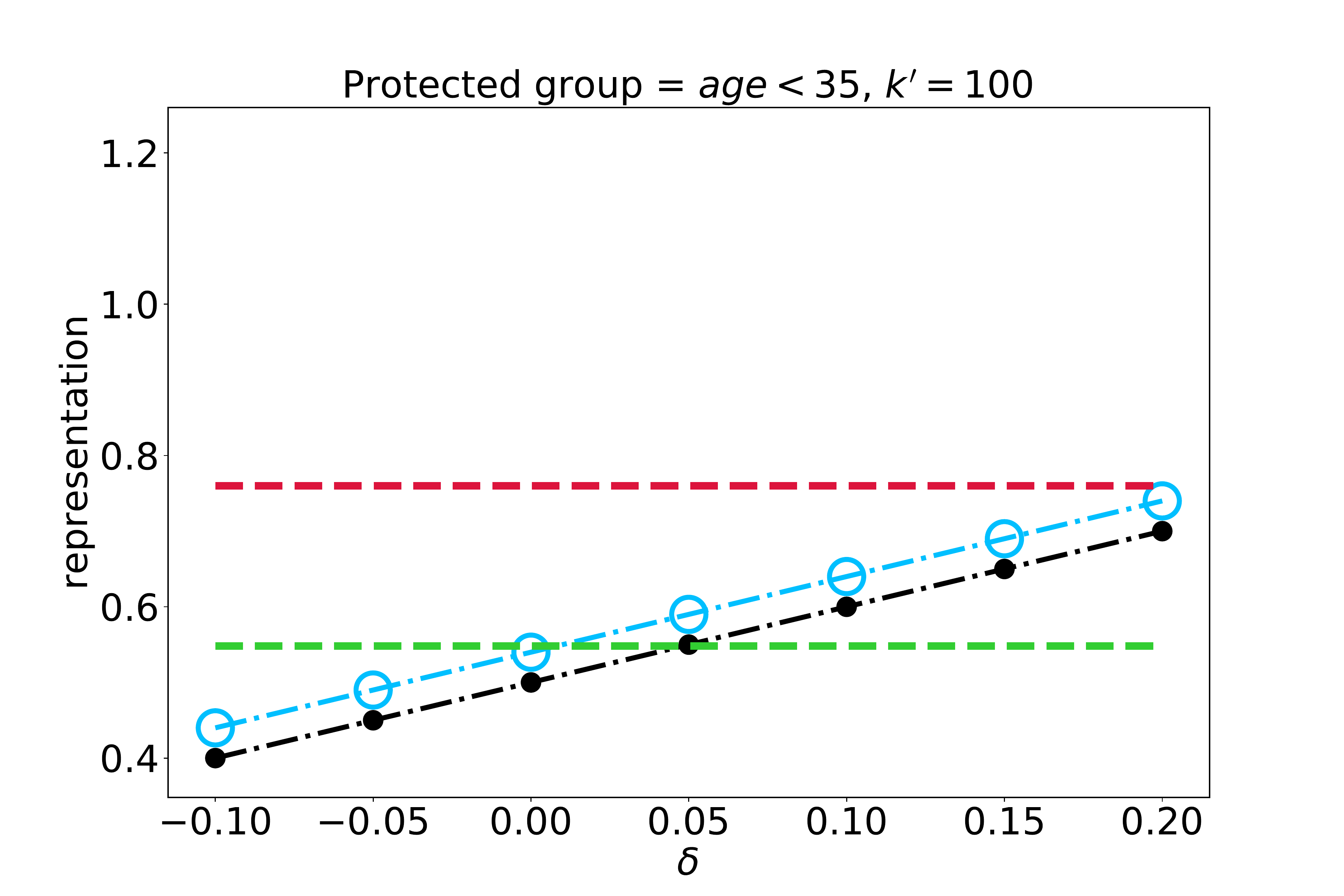} 
		\caption{Representation at top $100$ ranks.}
	\end{subfigure}
	
	\begin{subfigure}[b]{0.33\linewidth}
		\centering
		\includegraphics[scale=0.149]{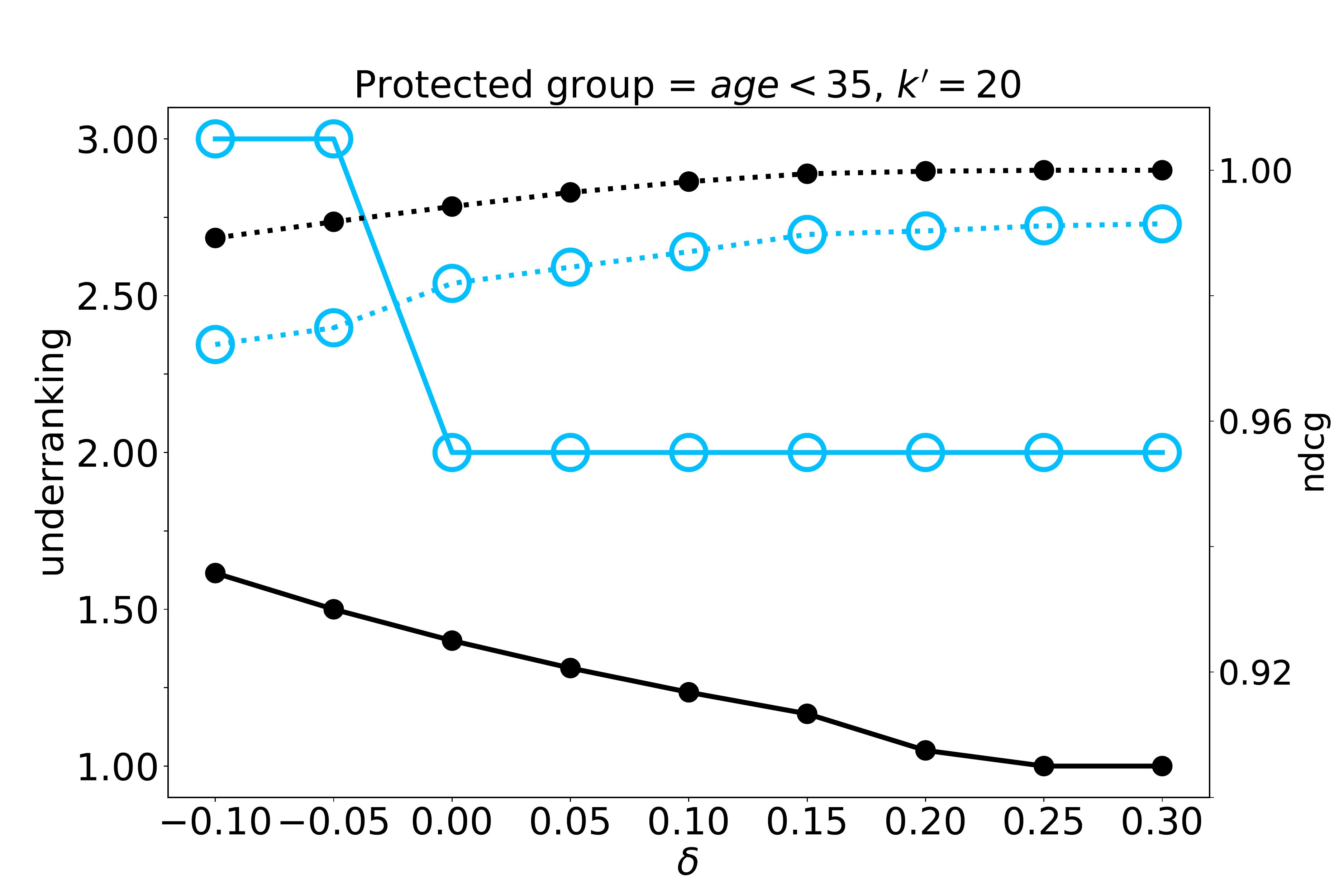} 
		\caption{Underranking, nDCG at top $20$ ranks.}
	\end{subfigure}
	\begin{subfigure}[b]{0.33\linewidth}
		\centering
		\includegraphics[scale=0.149]{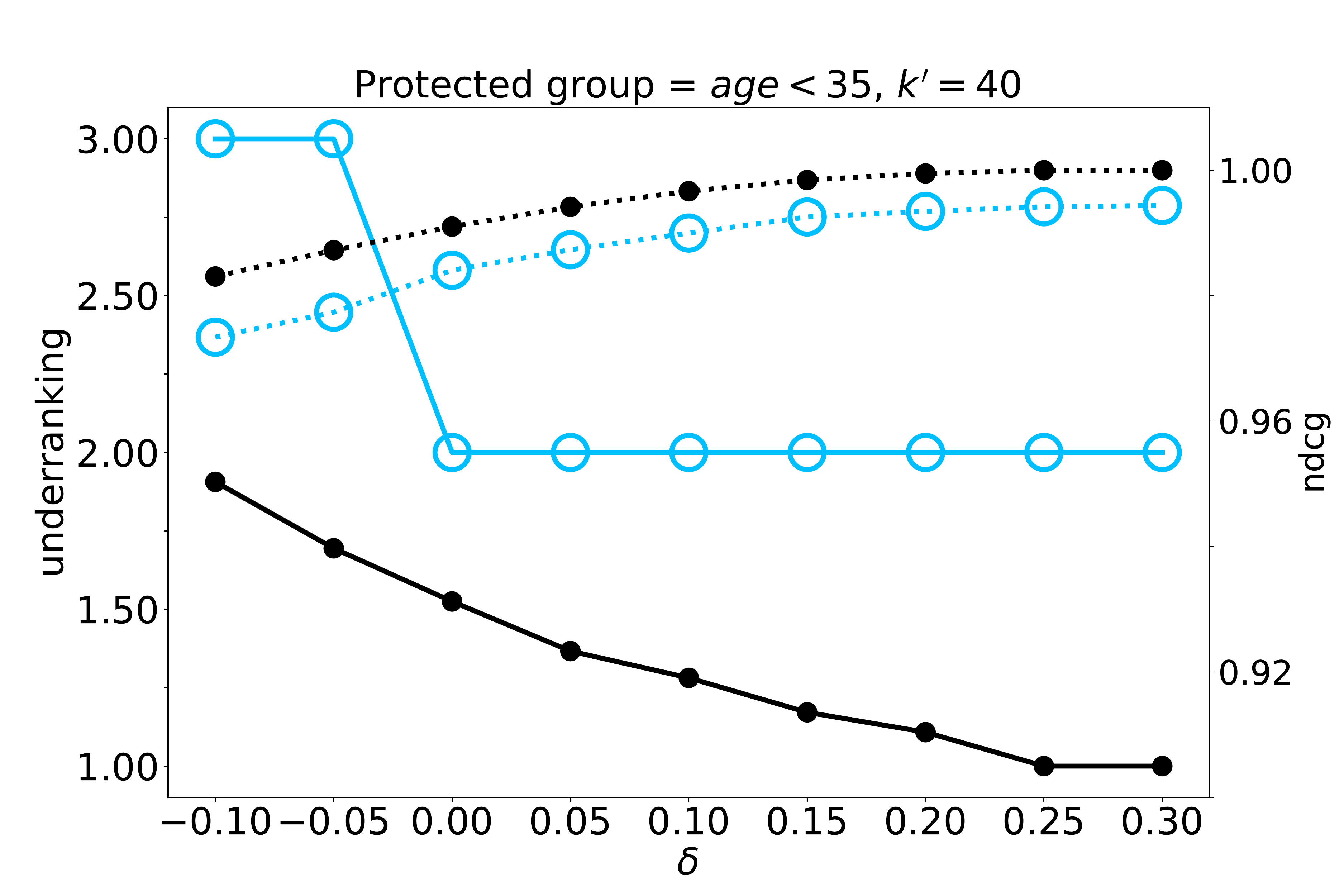} 
		\caption{Underranking, nDCG at top $40$ ranks.}
	\end{subfigure}
	\begin{subfigure}[b]{0.33\linewidth}
		\centering
		\includegraphics[scale=0.149]{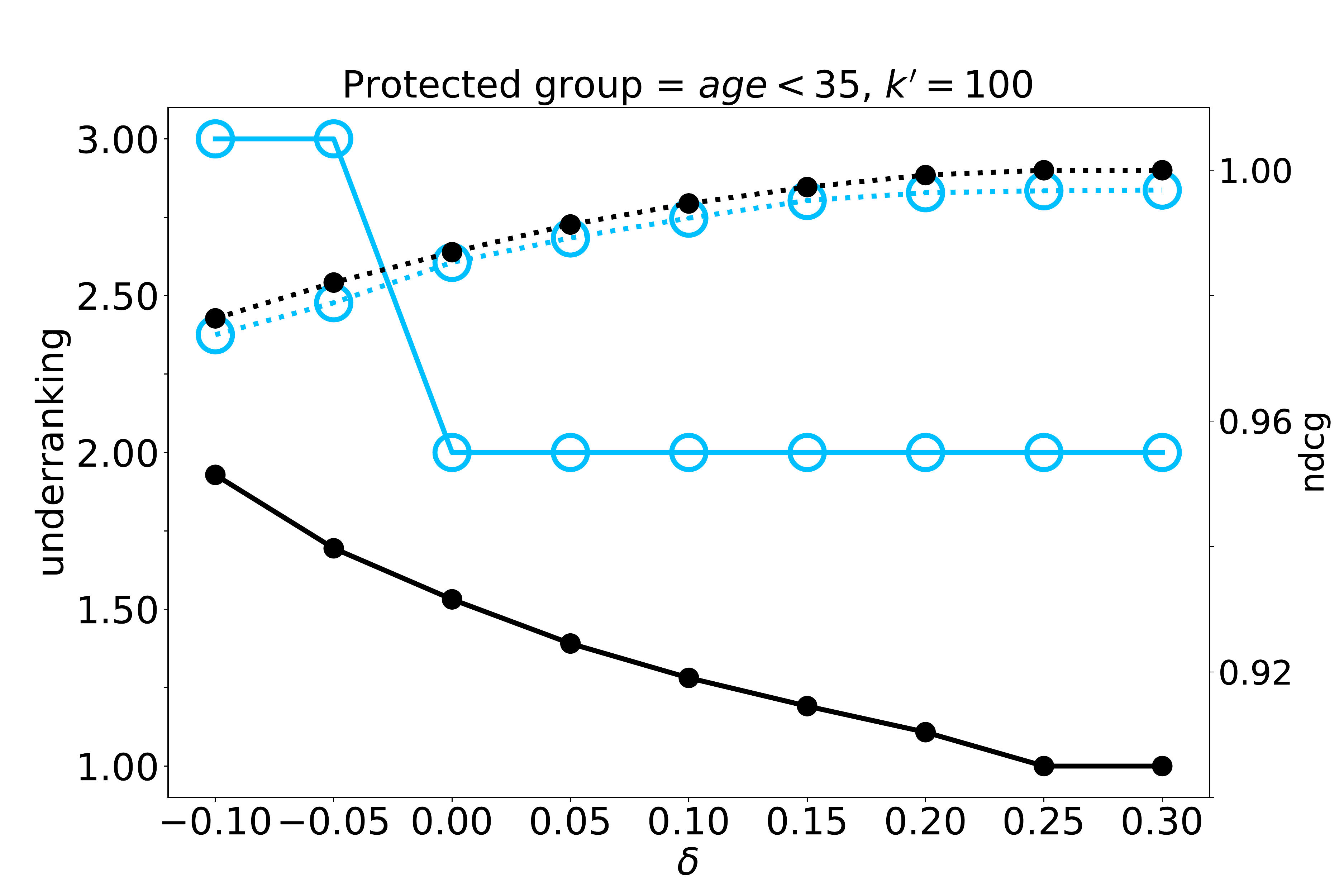} 
		\caption{Underranking, nDCG at top $100$ ranks.}
	\end{subfigure}
	\caption{Results on the German Credit Risk dataset with \textit{age}$<35$ as the protected group.}
	\label{fig:german_35_rev}
\end{figure}

\begin{figure}[H]
	\begin{subfigure}[b]{\linewidth}
		\centering
		\includegraphics[scale=0.2]{results/legend.pdf} 
	\end{subfigure}
	
	\begin{subfigure}[b]{0.33\linewidth}
		\centering
		\includegraphics[scale=0.149]{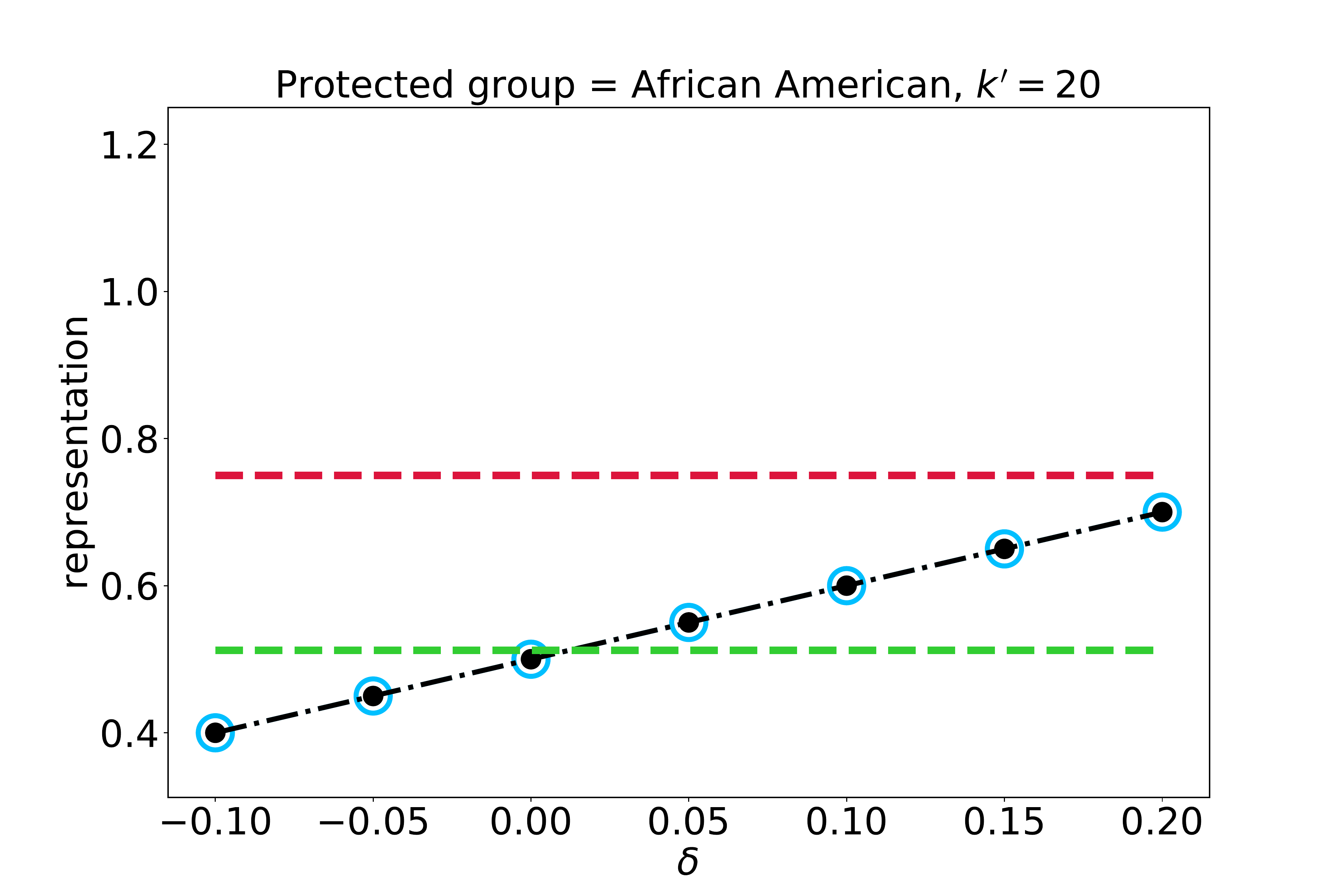} 
		\caption{Representation at top $20$ ranks.}
	\end{subfigure}
	\begin{subfigure}[b]{0.33\linewidth}
		\centering
		\includegraphics[scale=0.149]{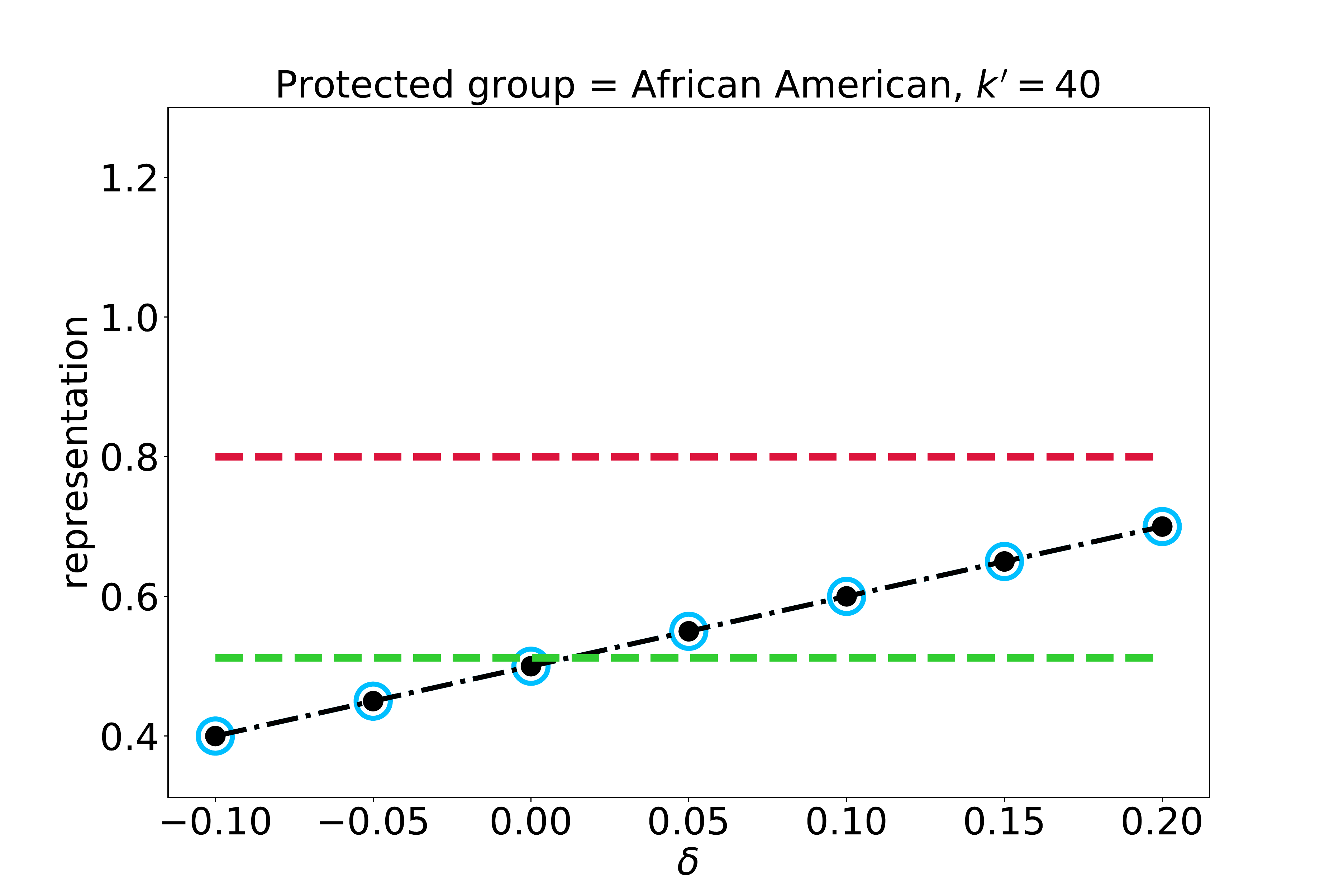} 
		\caption{Representation at top $40$ ranks.}
	\end{subfigure}
	\begin{subfigure}[b]{0.33\linewidth}
		\centering
		\includegraphics[scale=0.149]{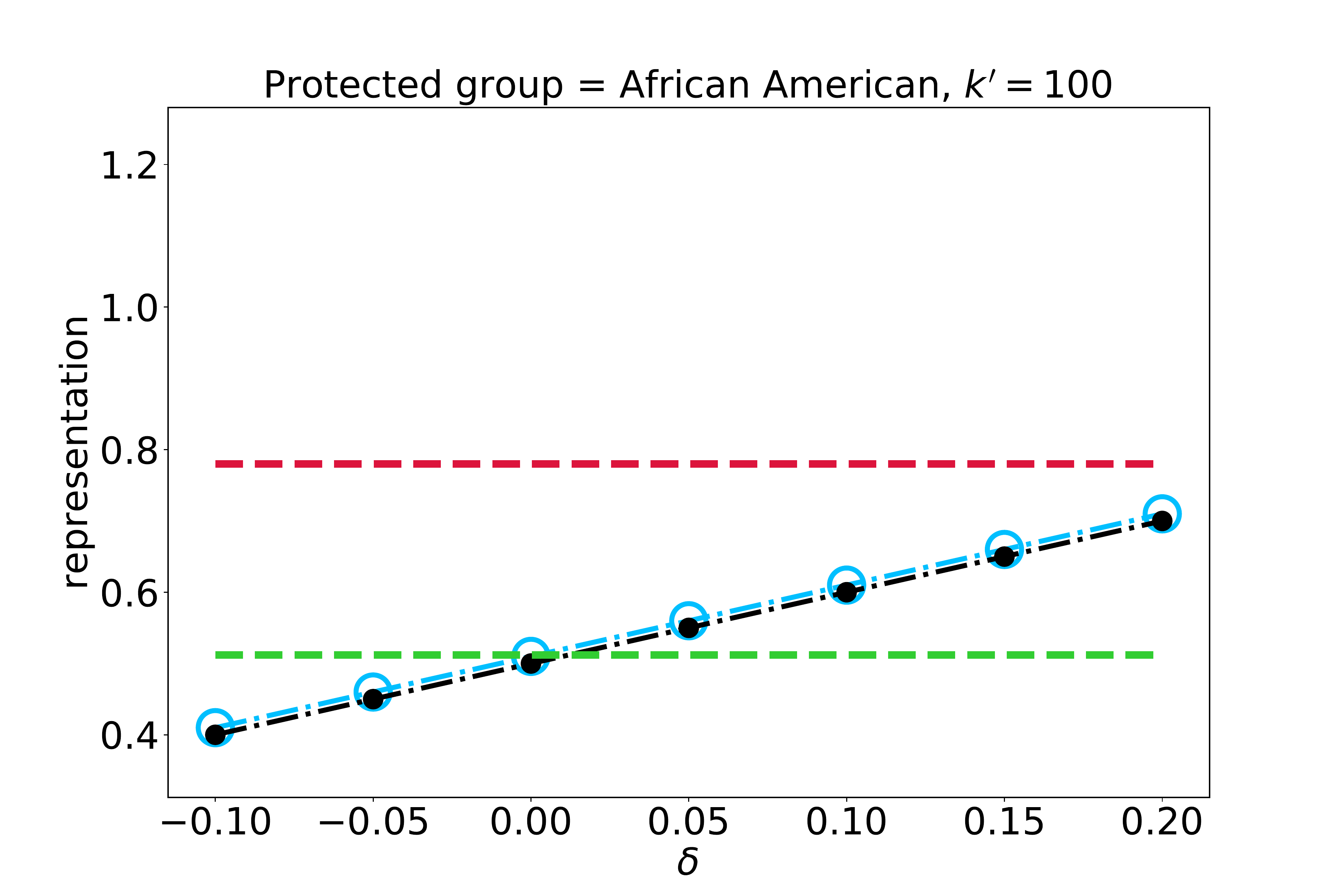} 
		\caption{Representation at top $100$ ranks.}
	\end{subfigure}
	
	\begin{subfigure}[b]{0.33\linewidth}
		\centering
		\includegraphics[scale=0.149]{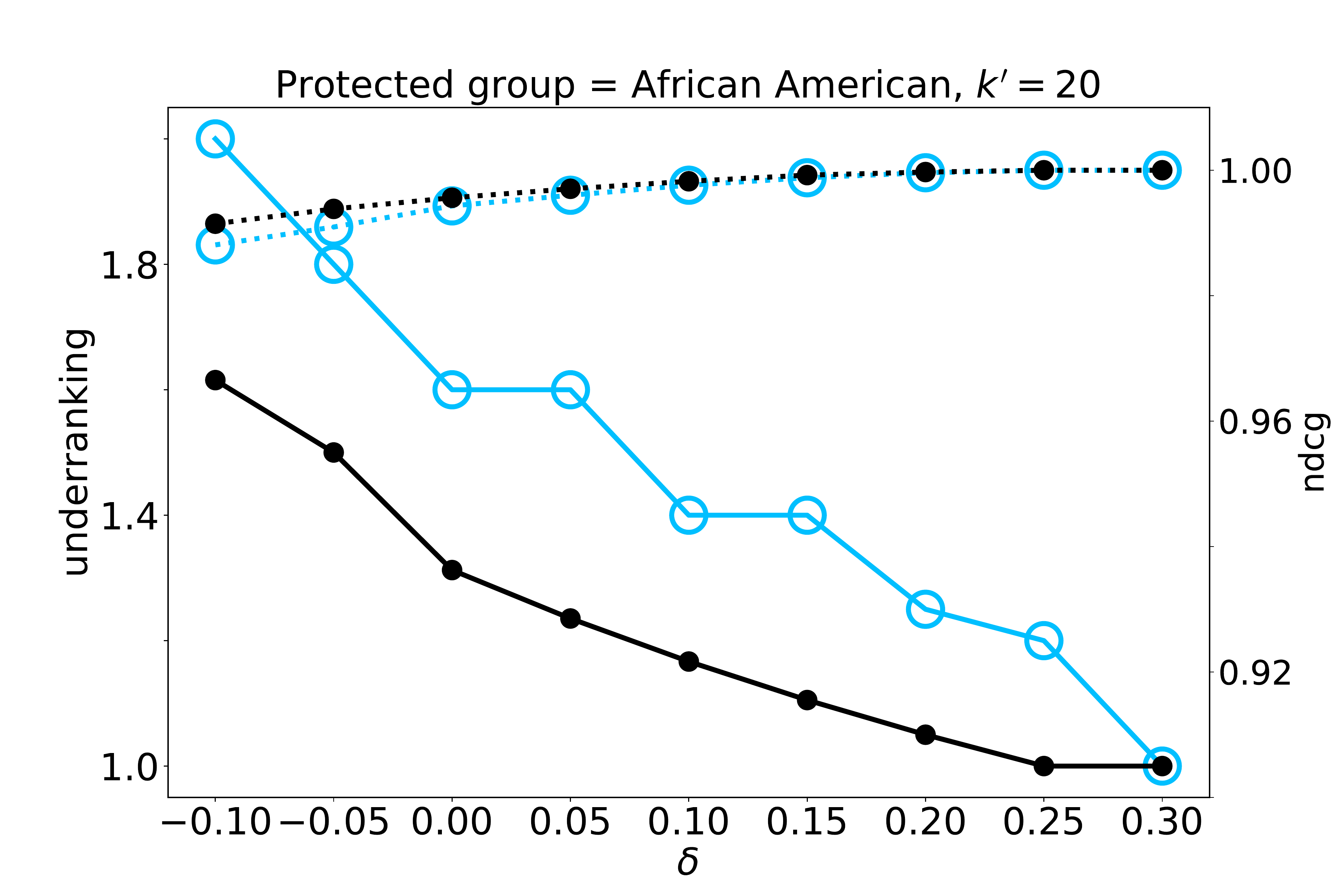} 
		\caption{Underranking, nDCG at top $20$ ranks.}
	\end{subfigure}
	\begin{subfigure}[b]{0.33\linewidth}
		\centering
		\includegraphics[scale=0.149]{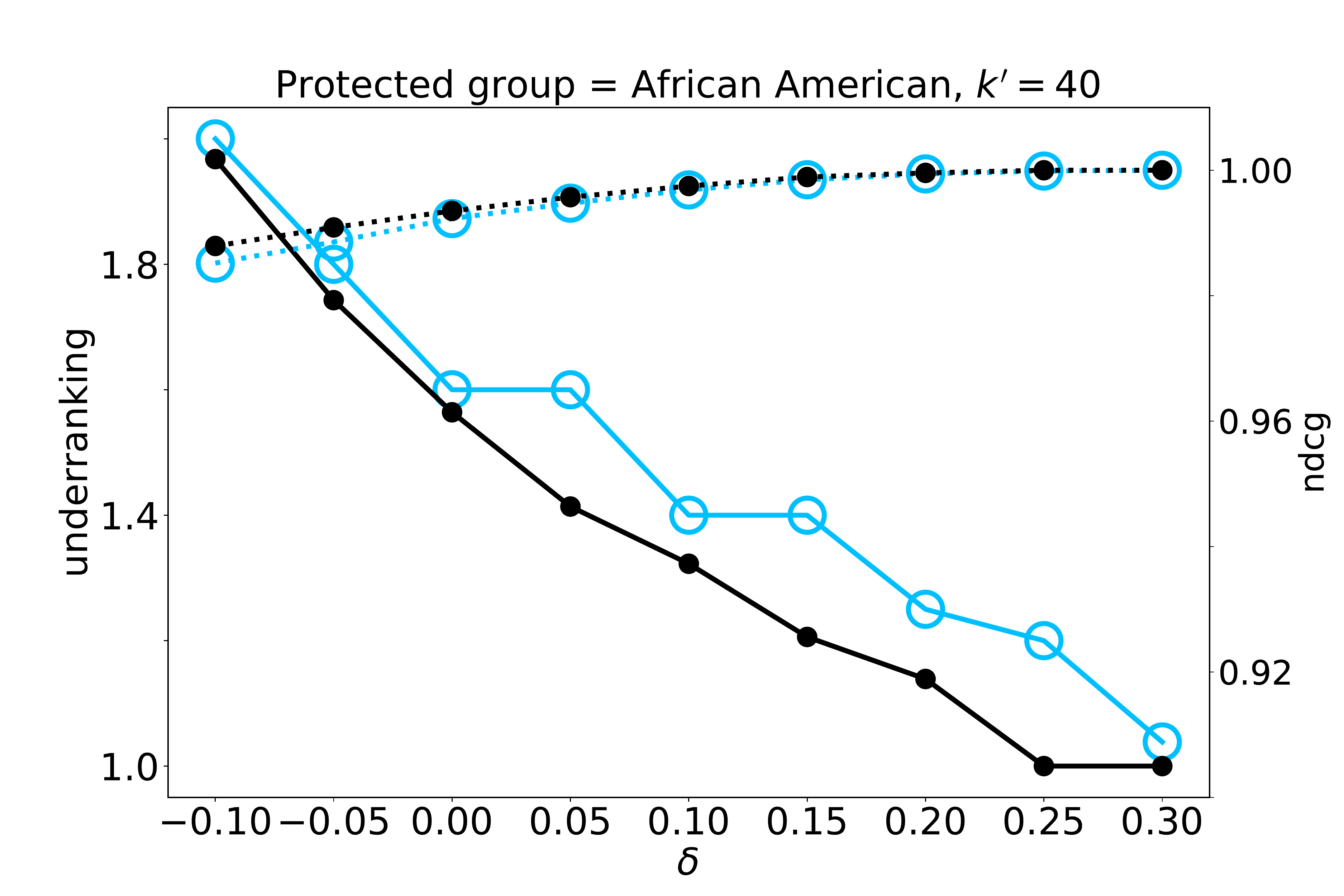} 
		\caption{Underranking, nDCG at top $40$ ranks.}
	\end{subfigure}
	\begin{subfigure}[b]{0.33\linewidth}
		\centering
		\includegraphics[scale=0.149]{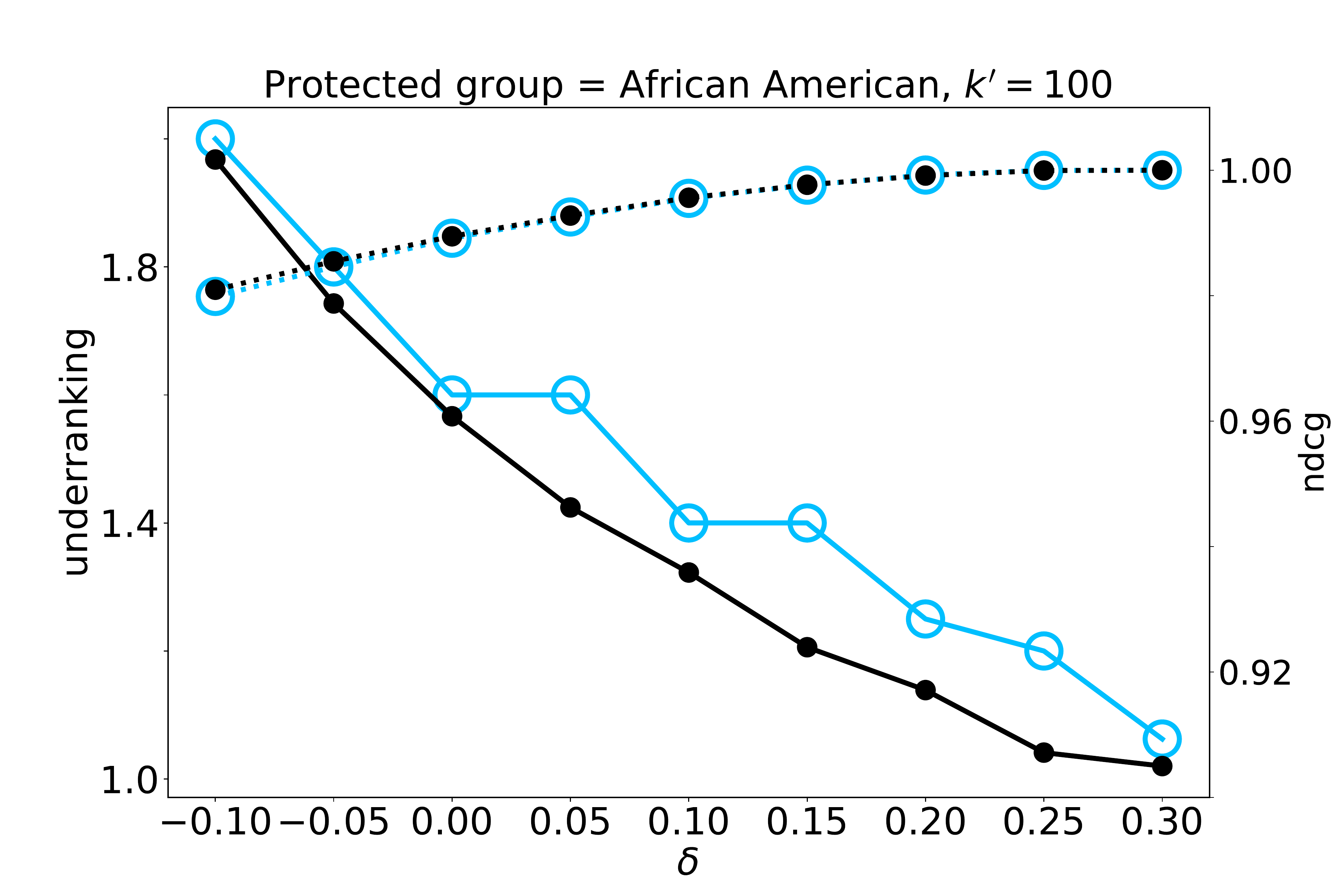} 
		\caption{Underranking, nDCG at top $100$ ranks.}
	\end{subfigure}
	\caption{Results on the COMPAS Recidivism dataset with \textit{African American} as the protected group.}
	\label{fig:compas_race_rev}
\end{figure}

\begin{figure}[H]
	\begin{subfigure}[b]{\linewidth}
		\centering
		\includegraphics[scale=0.2]{results/legend.pdf} 
	\end{subfigure}
	
	\begin{subfigure}[b]{0.33\linewidth}
		\centering
		\includegraphics[scale=0.149]{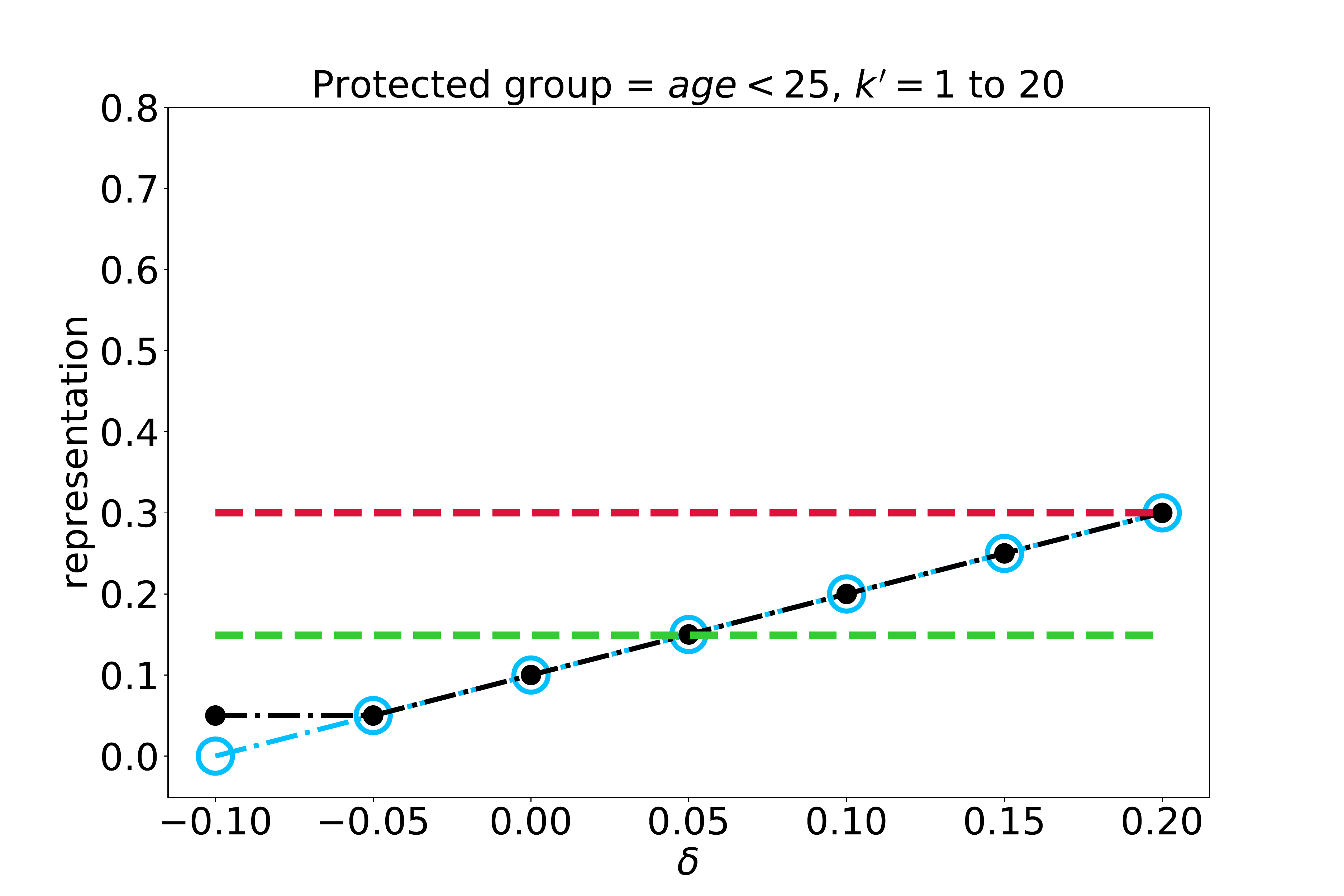} 
		\caption{Representation at ranks $1$ to $20$.}
	\end{subfigure}
	\begin{subfigure}[b]{0.33\linewidth}
		\centering
		\includegraphics[scale=0.149]{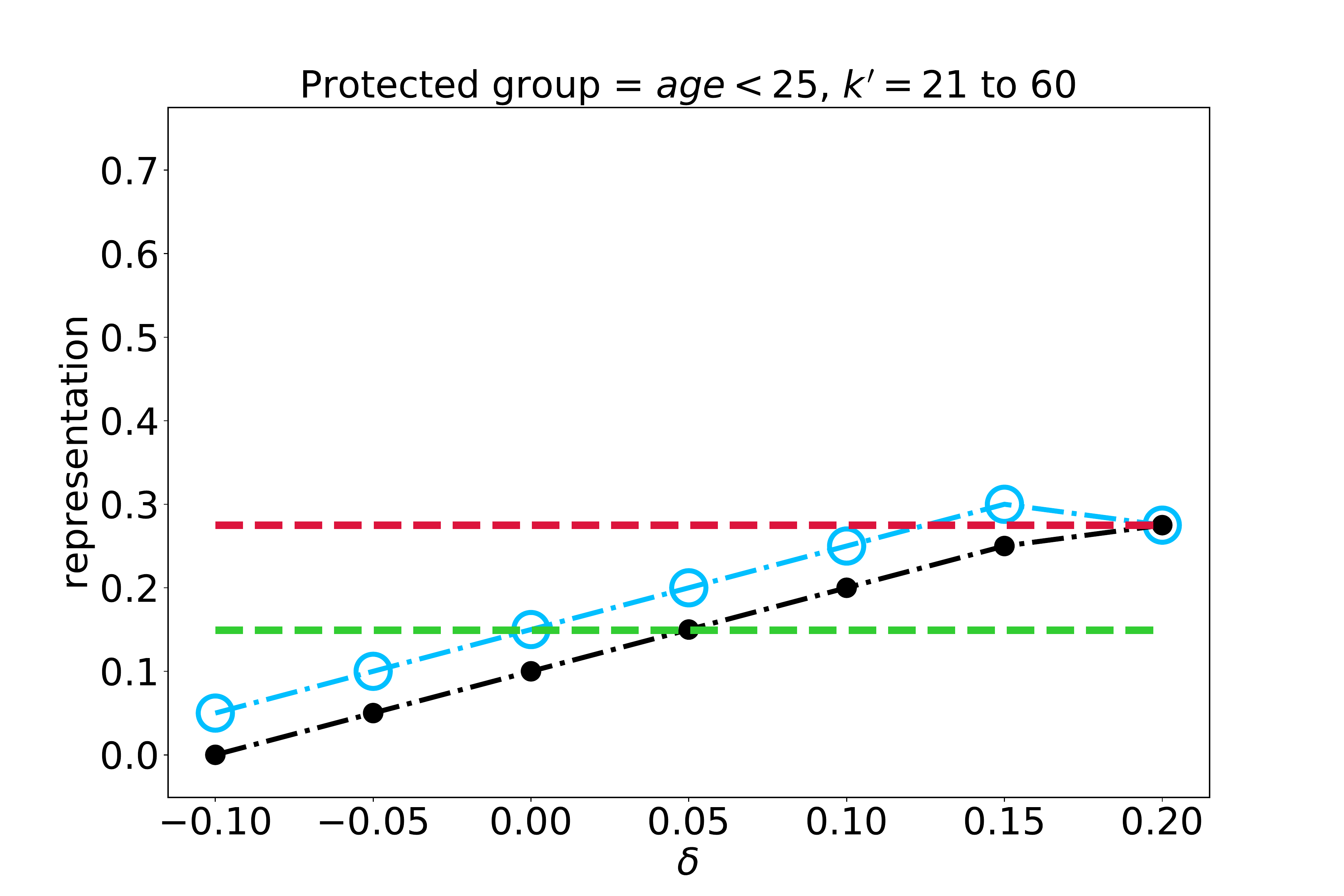} 
		\caption{Representation at ranks $21$ to $60$.}
	\end{subfigure}
	\begin{subfigure}[b]{0.33\linewidth}
		\centering
		\includegraphics[scale=0.149]{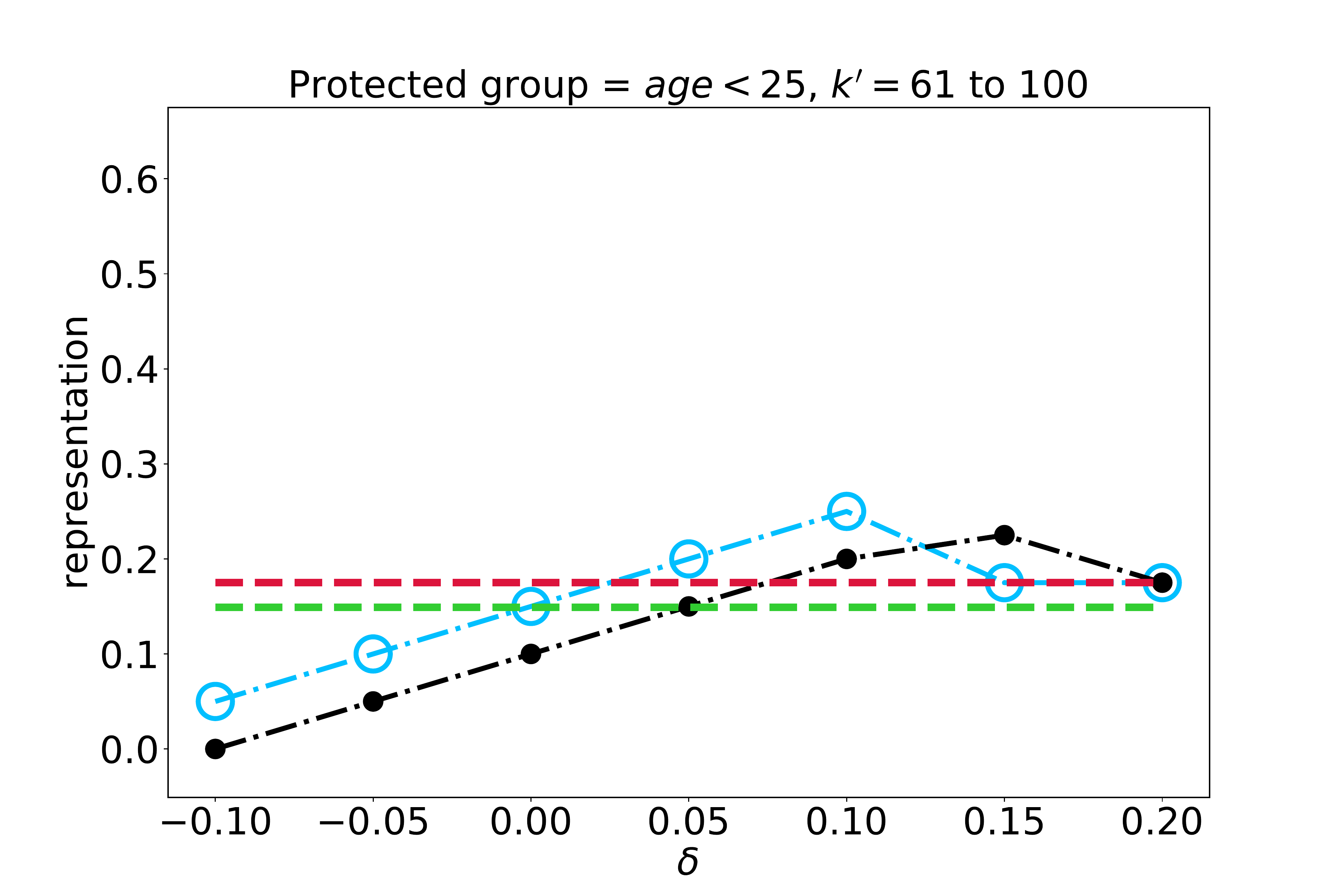} 
		\caption{Representation at ranks $61$ to $100$.}
	\end{subfigure}
	
	\begin{subfigure}[b]{0.33\linewidth}
		\centering
		\includegraphics[scale=0.149]{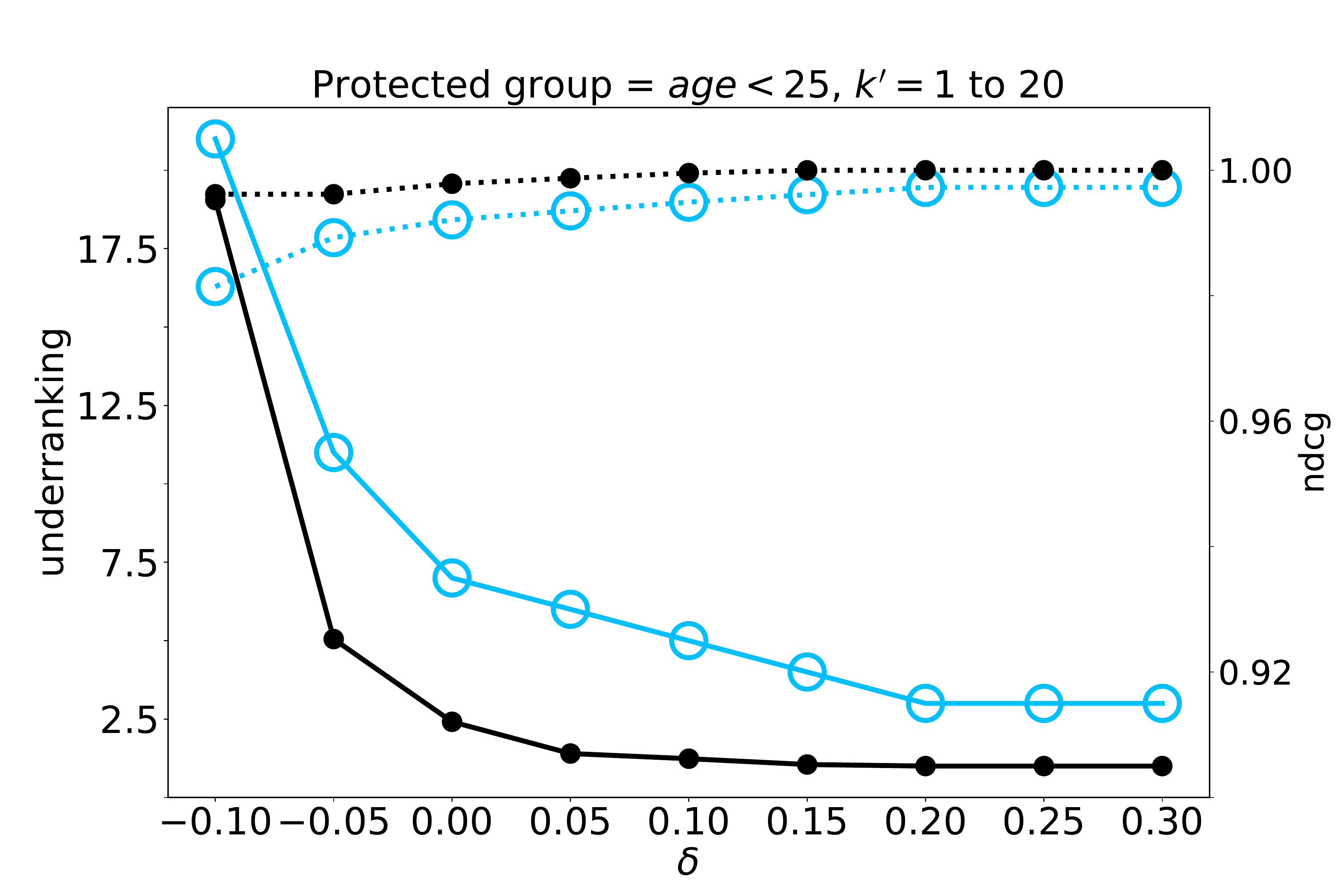} 
		\caption{Underranking, nDCG at top $20$ ranks.}
	\end{subfigure}
	\begin{subfigure}[b]{0.33\linewidth}
		\centering
		\includegraphics[scale=0.149]{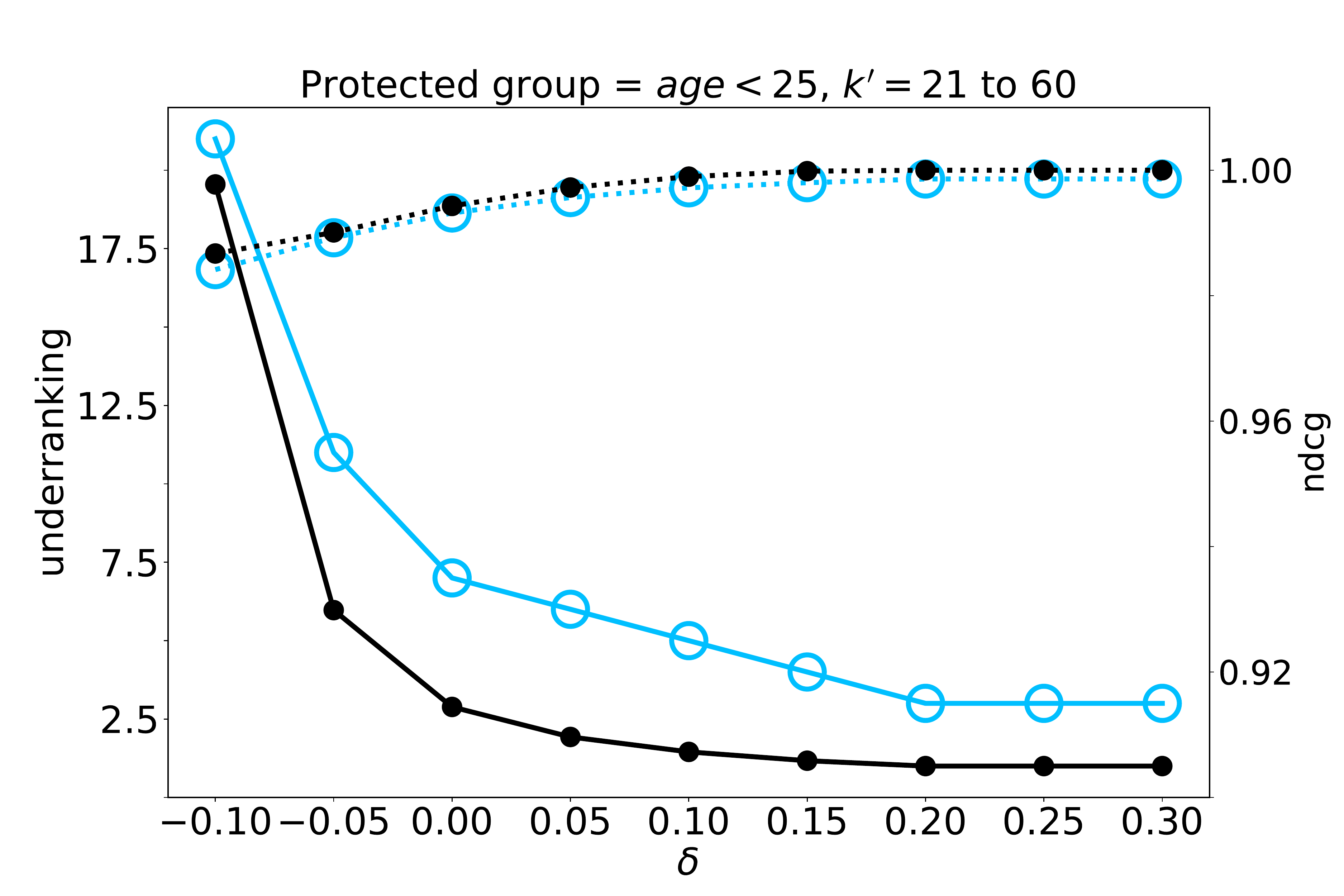} 
		\caption{Underranking, nDCG at top $60$ ranks.}
	\end{subfigure}
	\begin{subfigure}[b]{0.33\linewidth}
		\centering
		\includegraphics[scale=0.149]{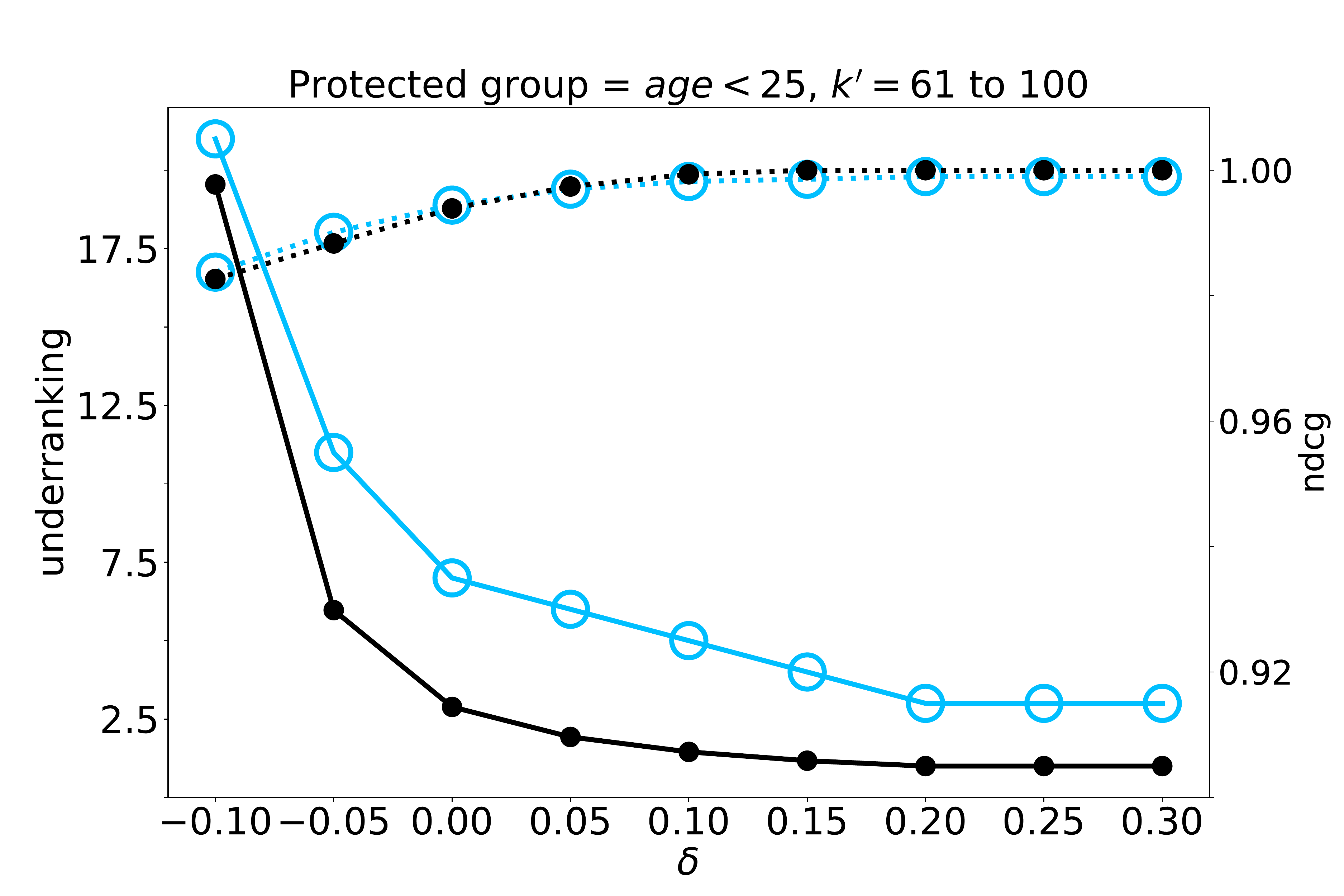} 
		\caption{Underranking, nDCG at top $100$ ranks.}
	\end{subfigure}
	\caption{Results on the German Credit Risk dataset with \textit{age}$<25$ as the protected group.}
	\label{fig:german_25_rev_block}
\end{figure}

\begin{figure}[H]
	\begin{subfigure}[b]{\linewidth}
		\centering
		\includegraphics[scale=0.2]{results/legend.pdf} 
	\end{subfigure}
	
	\begin{subfigure}[b]{0.33\linewidth}
		\centering
		\includegraphics[scale=0.149]{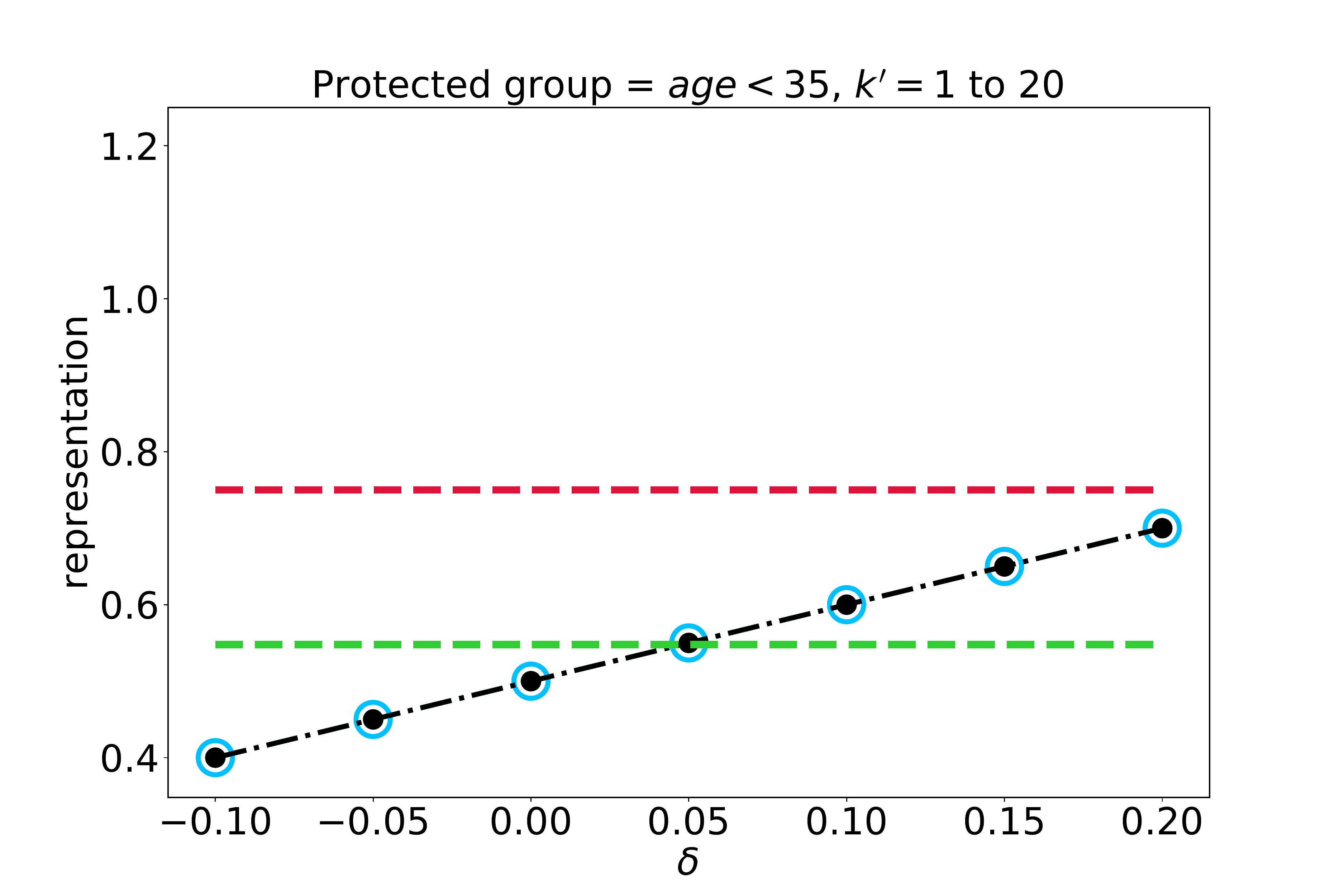} 
		\caption{Representation at ranks $1$ to $20$.}
	\end{subfigure}
	\begin{subfigure}[b]{0.33\linewidth}
		\centering
		\includegraphics[scale=0.149]{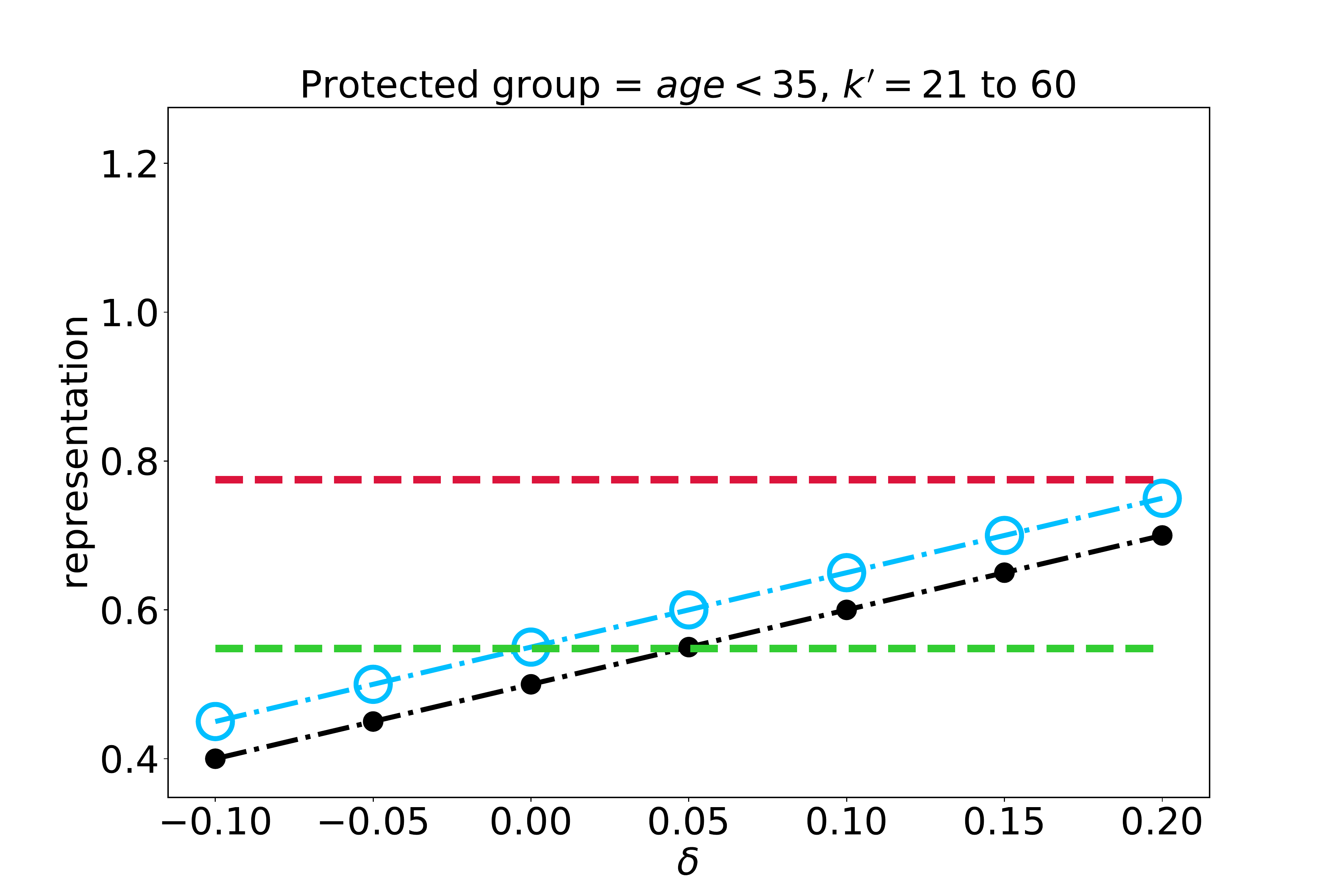} 
		\caption{Representation at ranks $21$ to $60$.}
	\end{subfigure}
	\begin{subfigure}[b]{0.33\linewidth}
		\centering
		\includegraphics[scale=0.149]{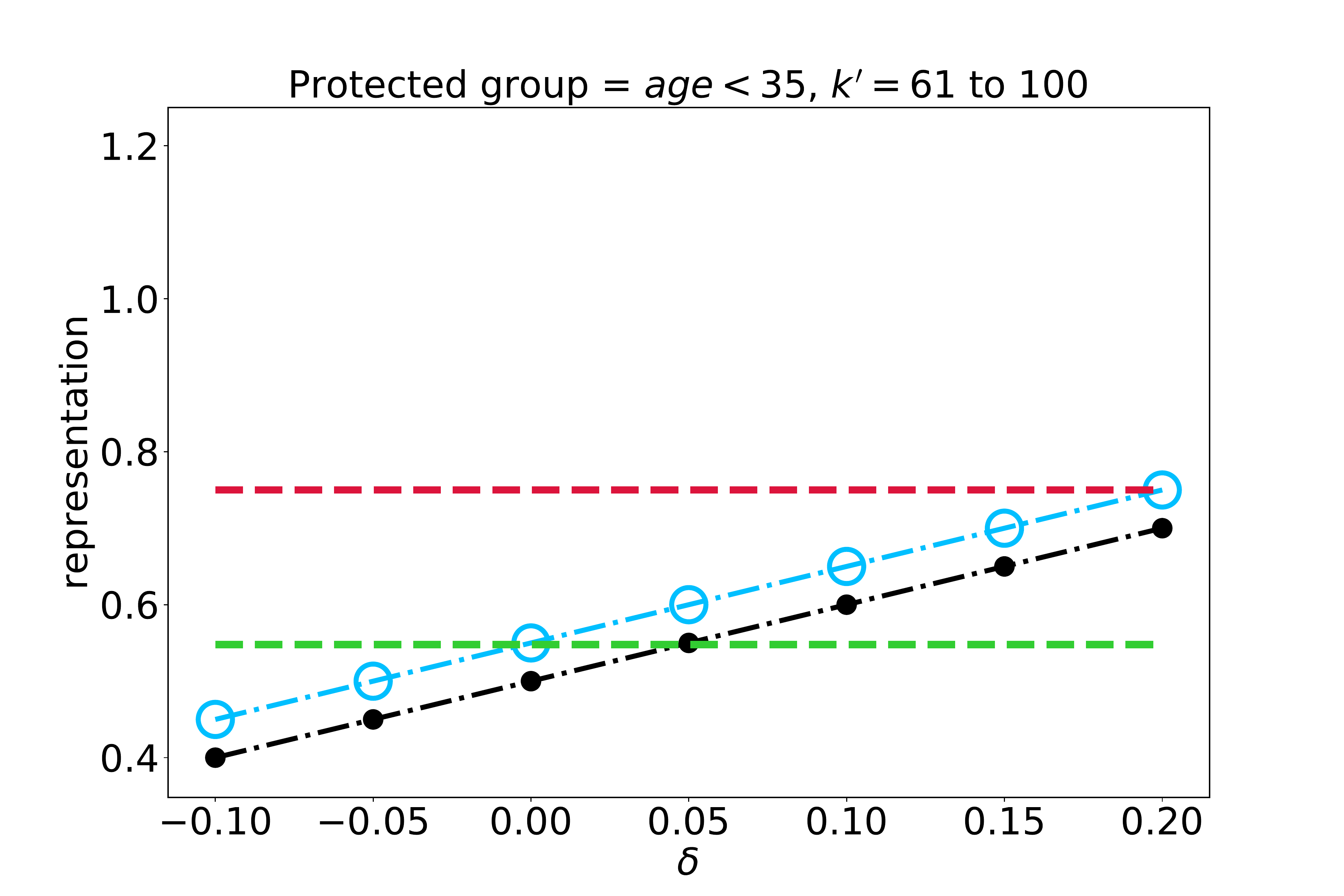} 
		\caption{Representation at ranks $61$ to $100$.}
	\end{subfigure}
	
	\begin{subfigure}[b]{0.33\linewidth}
		\centering
		\includegraphics[scale=0.149]{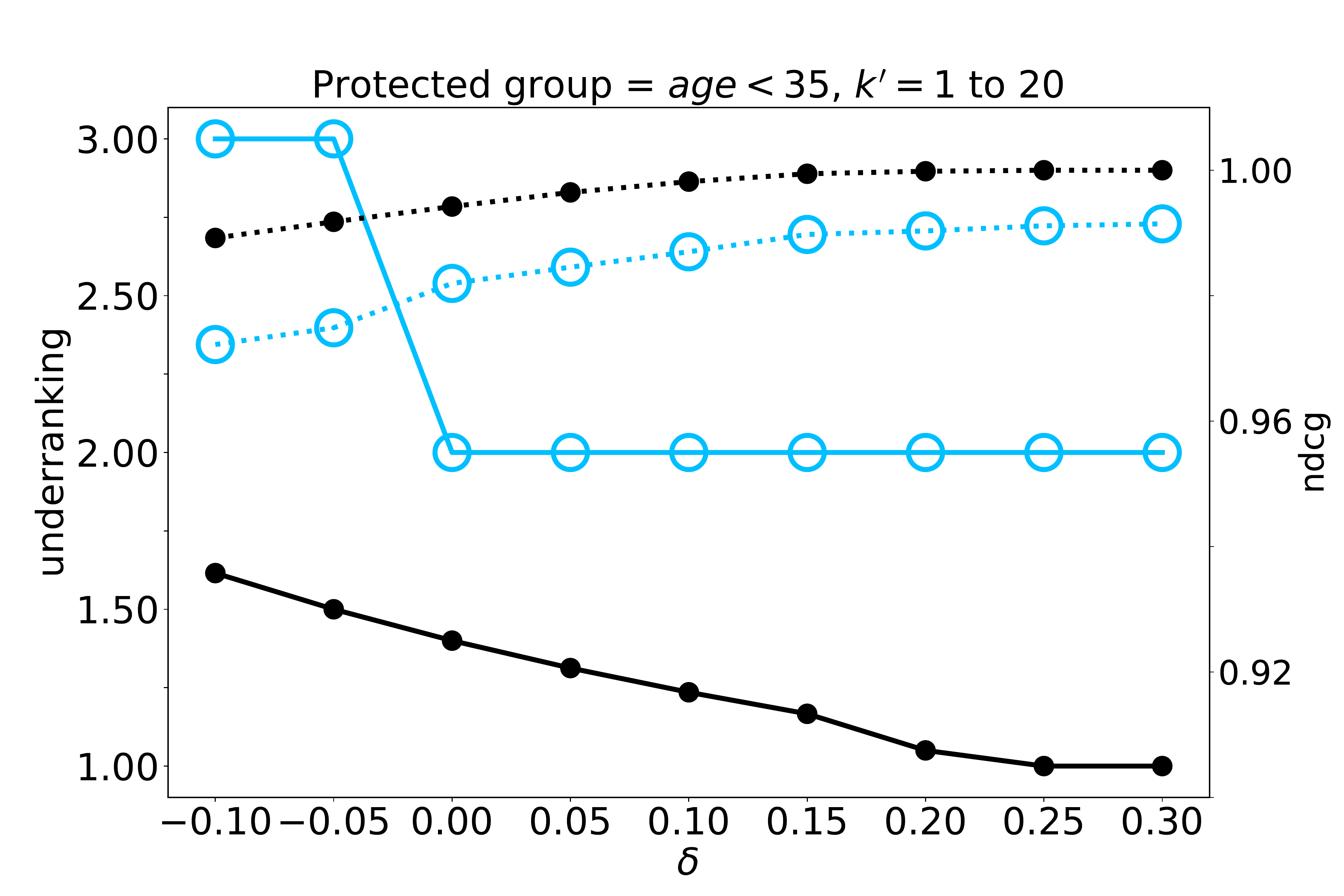} 
		\caption{Underranking, nDCG at top $20$ ranks.}
	\end{subfigure}
	\begin{subfigure}[b]{0.33\linewidth}
		\centering
		\includegraphics[scale=0.149]{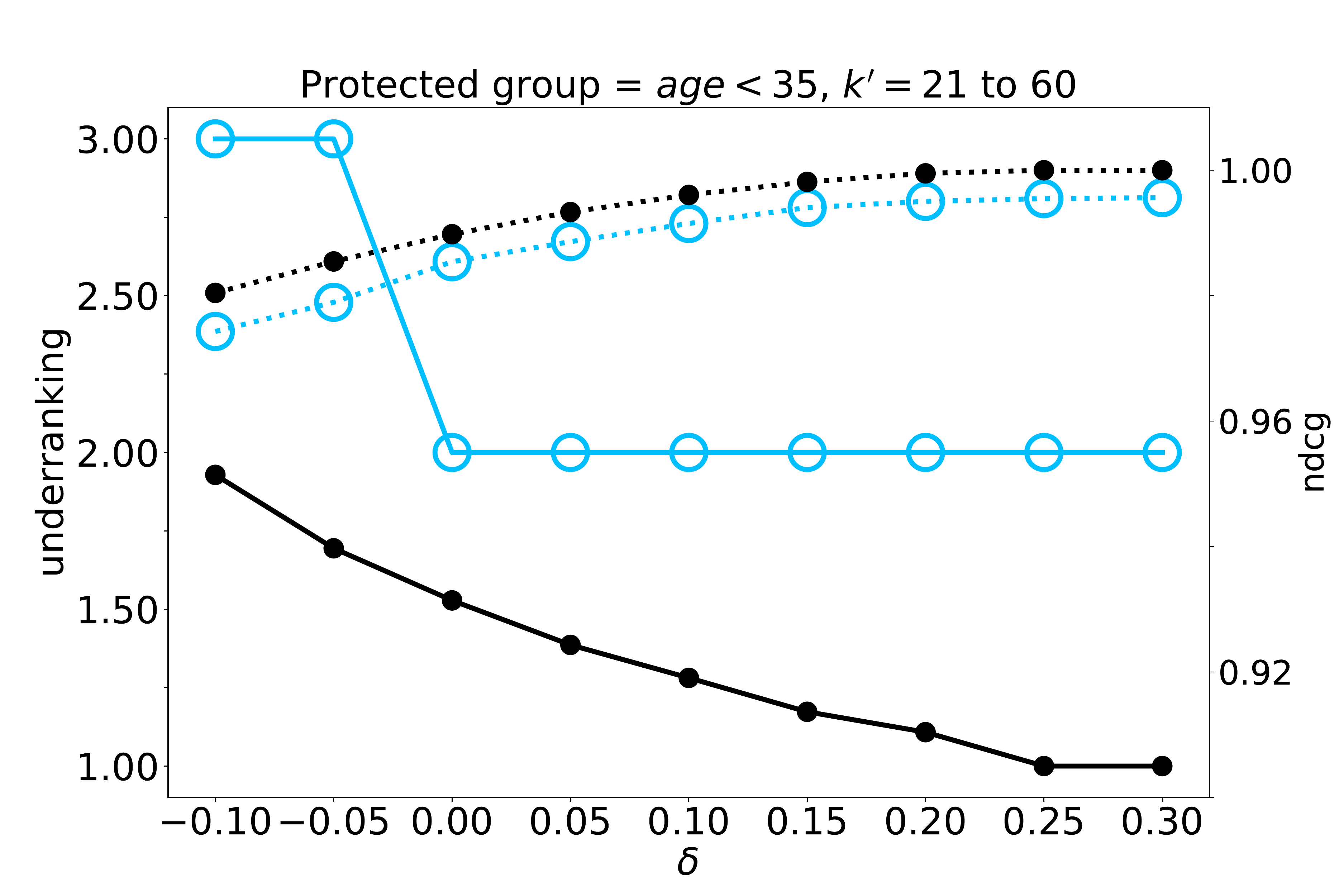} 
		\caption{Underranking, nDCG at top $60$ ranks.}
	\end{subfigure}
	\begin{subfigure}[b]{0.33\linewidth}
		\centering
		\includegraphics[scale=0.149]{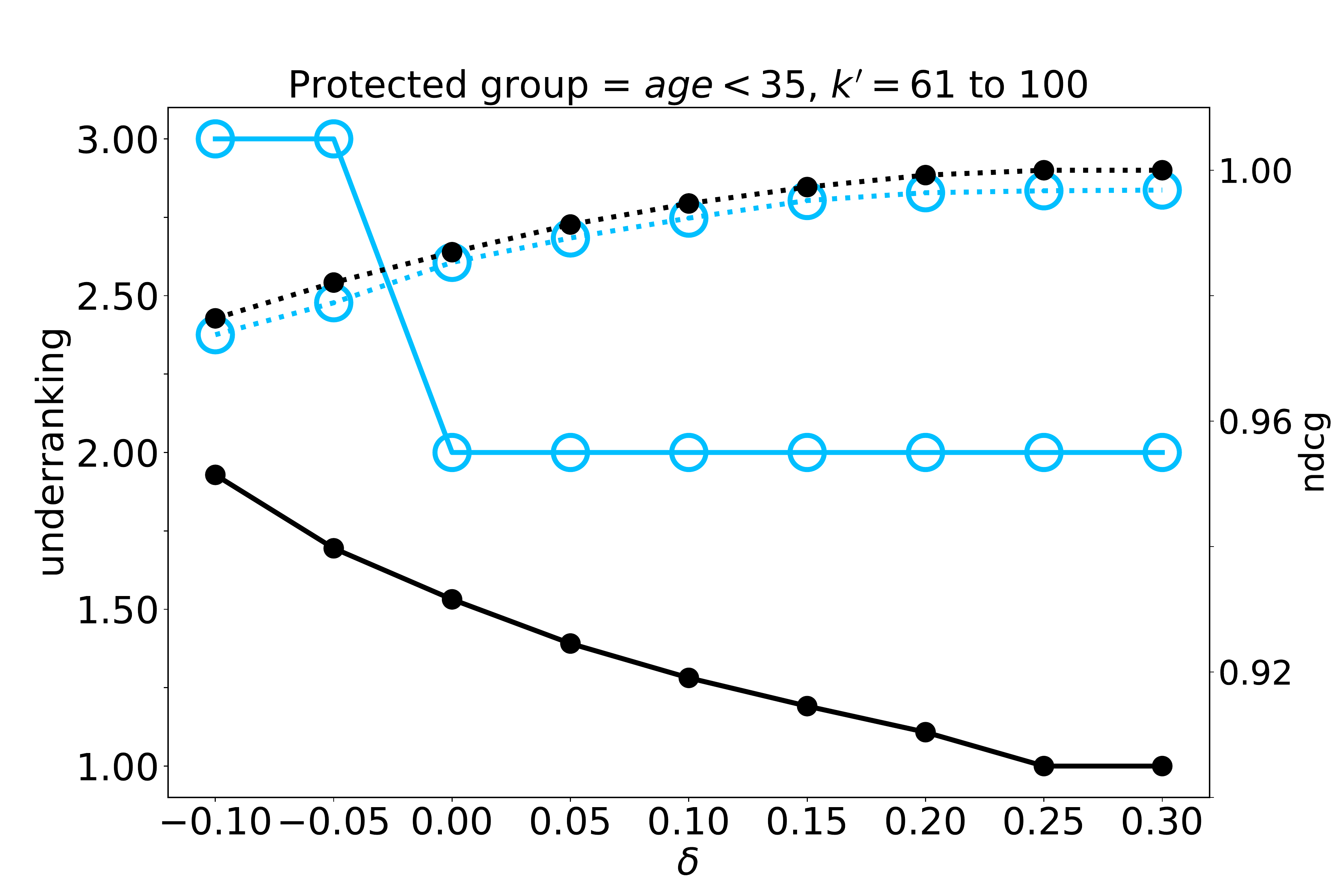} 
		\caption{Underranking, nDCG at top $100$ ranks.}
	\end{subfigure}
	\caption{Results on the German Credit Risk dataset with \textit{age}$<35$ as the protected group.}
	\label{fig:german_35_rev_block}
\end{figure}

\begin{figure}[H]
	\begin{subfigure}[b]{\linewidth}
		\centering
		\includegraphics[scale=0.2]{results/legend.pdf} 
	\end{subfigure}
	
	\begin{subfigure}[b]{0.33\linewidth}
		\centering
		\includegraphics[scale=0.149]{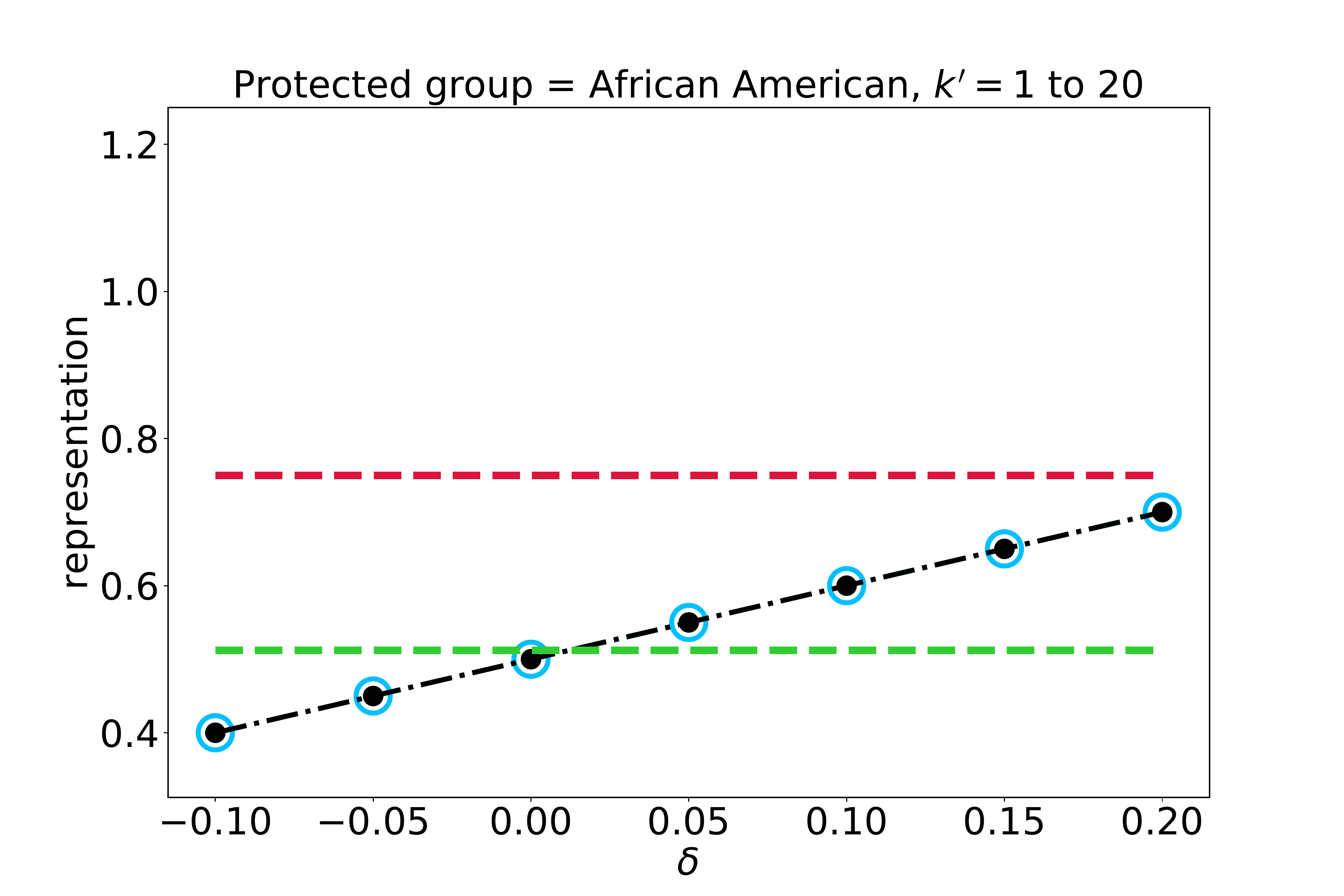} 
		\caption{Representation at ranks $1$ to $20$.}
	\end{subfigure}
	\begin{subfigure}[b]{0.33\linewidth}
		\centering
		\includegraphics[scale=0.149]{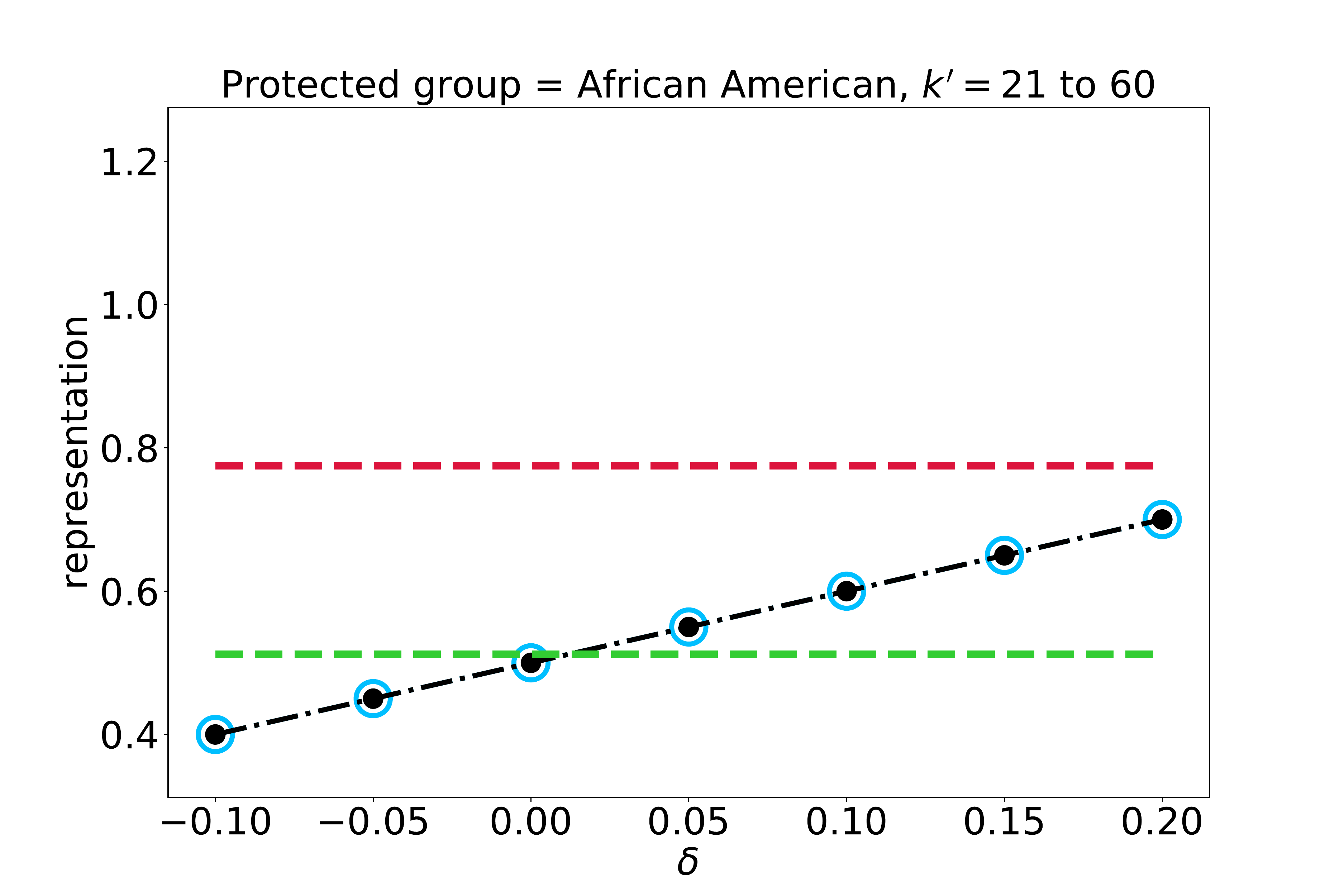} 
		\caption{Representation at ranks $21$ to $60$.}
	\end{subfigure}
	\begin{subfigure}[b]{0.33\linewidth}
		\centering
		\includegraphics[scale=0.149]{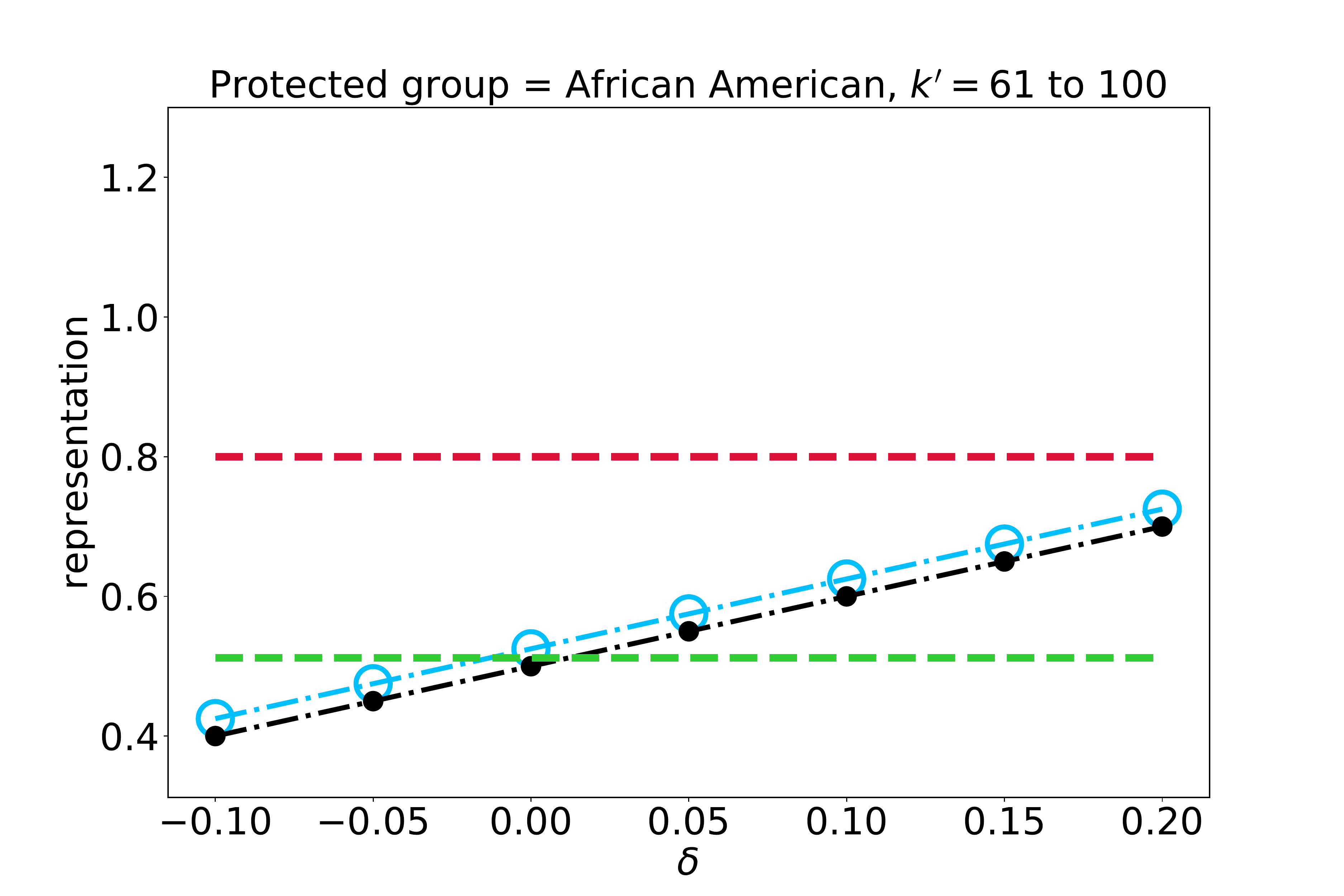} 
		\caption{Representation at ranks $61$ to $100$.}
	\end{subfigure}
	
	\begin{subfigure}[b]{0.33\linewidth}
		\centering
		\includegraphics[scale=0.149]{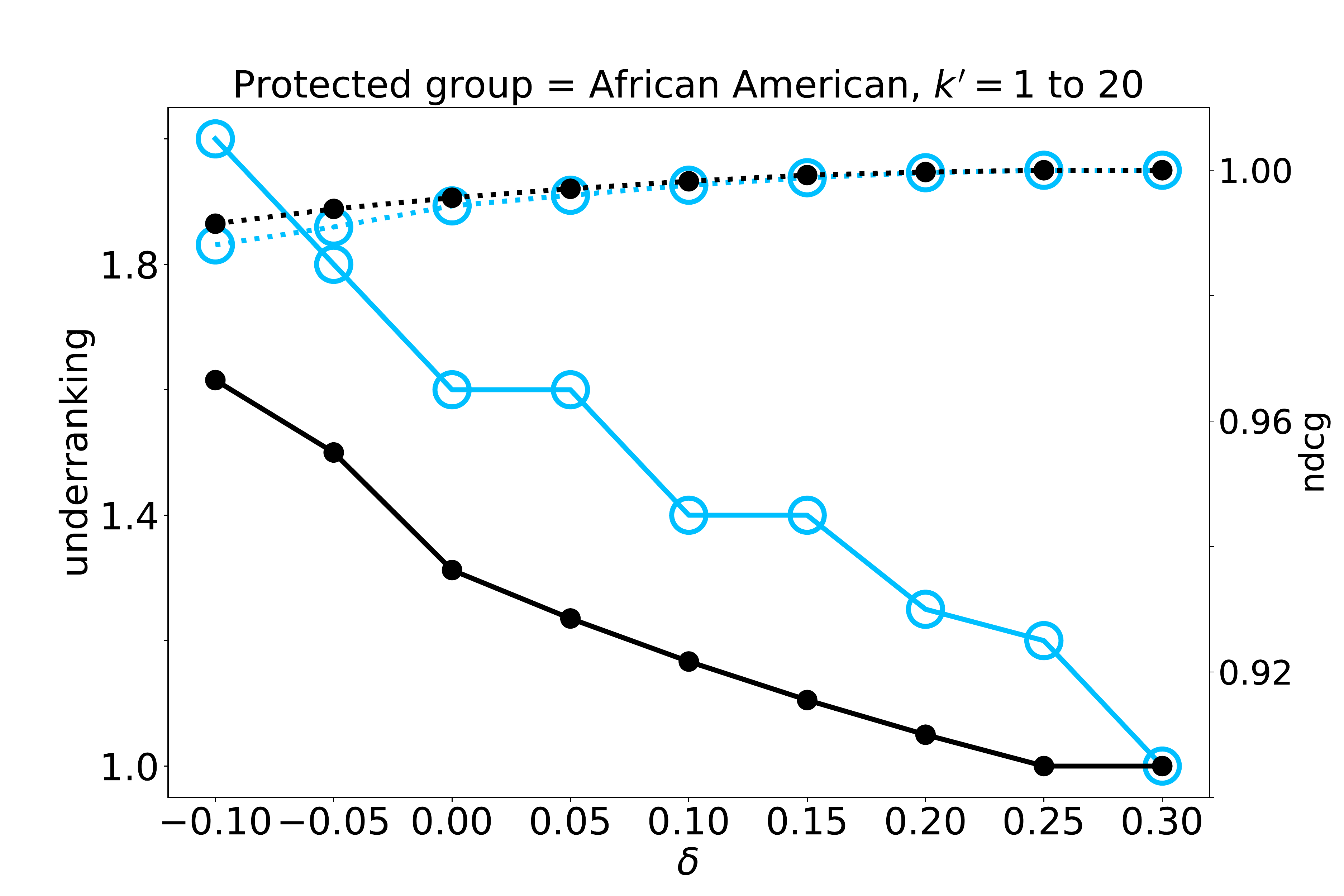} 
		\caption{Underranking, nDCG at top $20$ ranks.}
	\end{subfigure}
	\begin{subfigure}[b]{0.33\linewidth}
		\centering
		\includegraphics[scale=0.149]{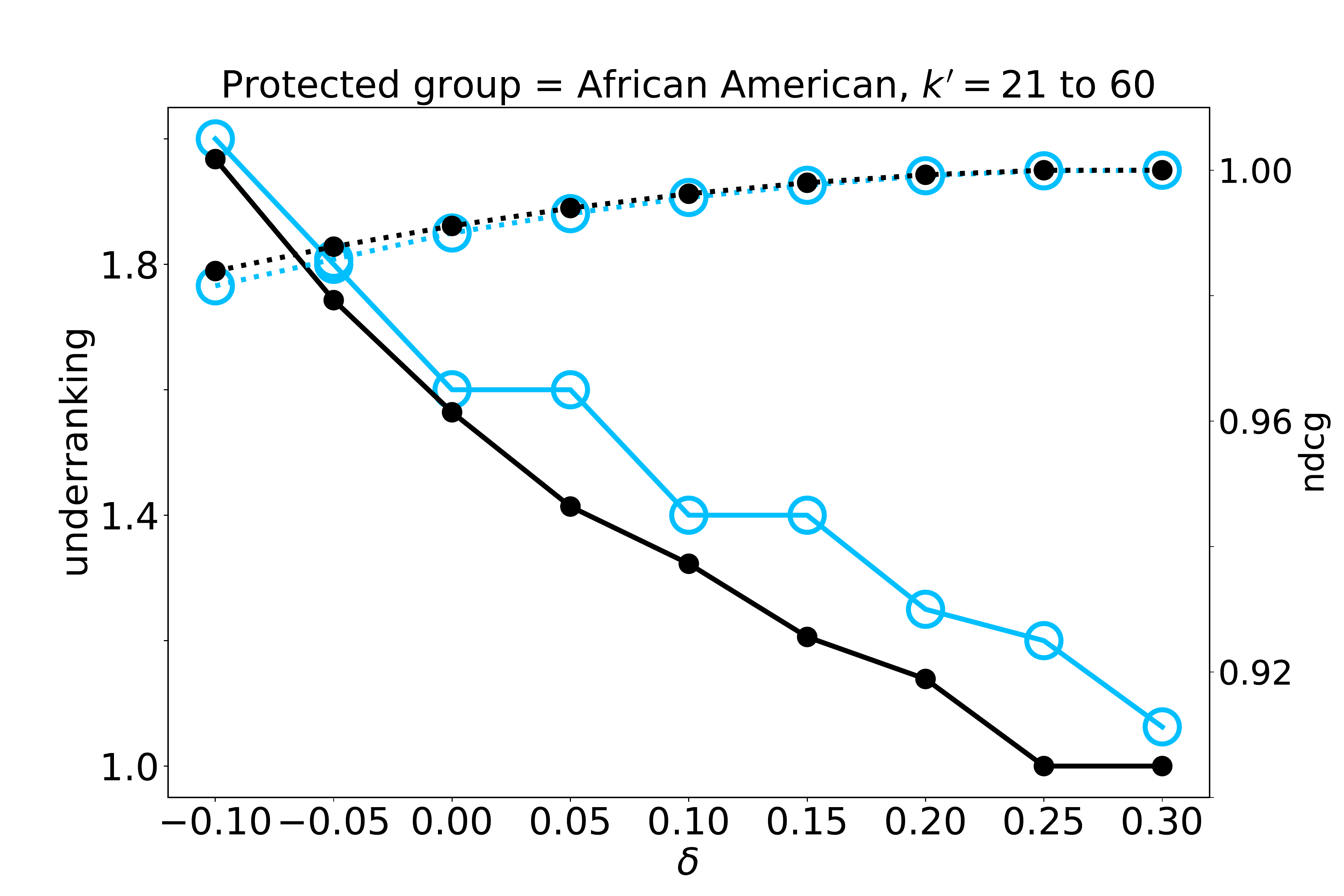} 
		\caption{Underranking, nDCG at top $60$ ranks.}
	\end{subfigure}
	\begin{subfigure}[b]{0.33\linewidth}
		\centering
		\includegraphics[scale=0.149]{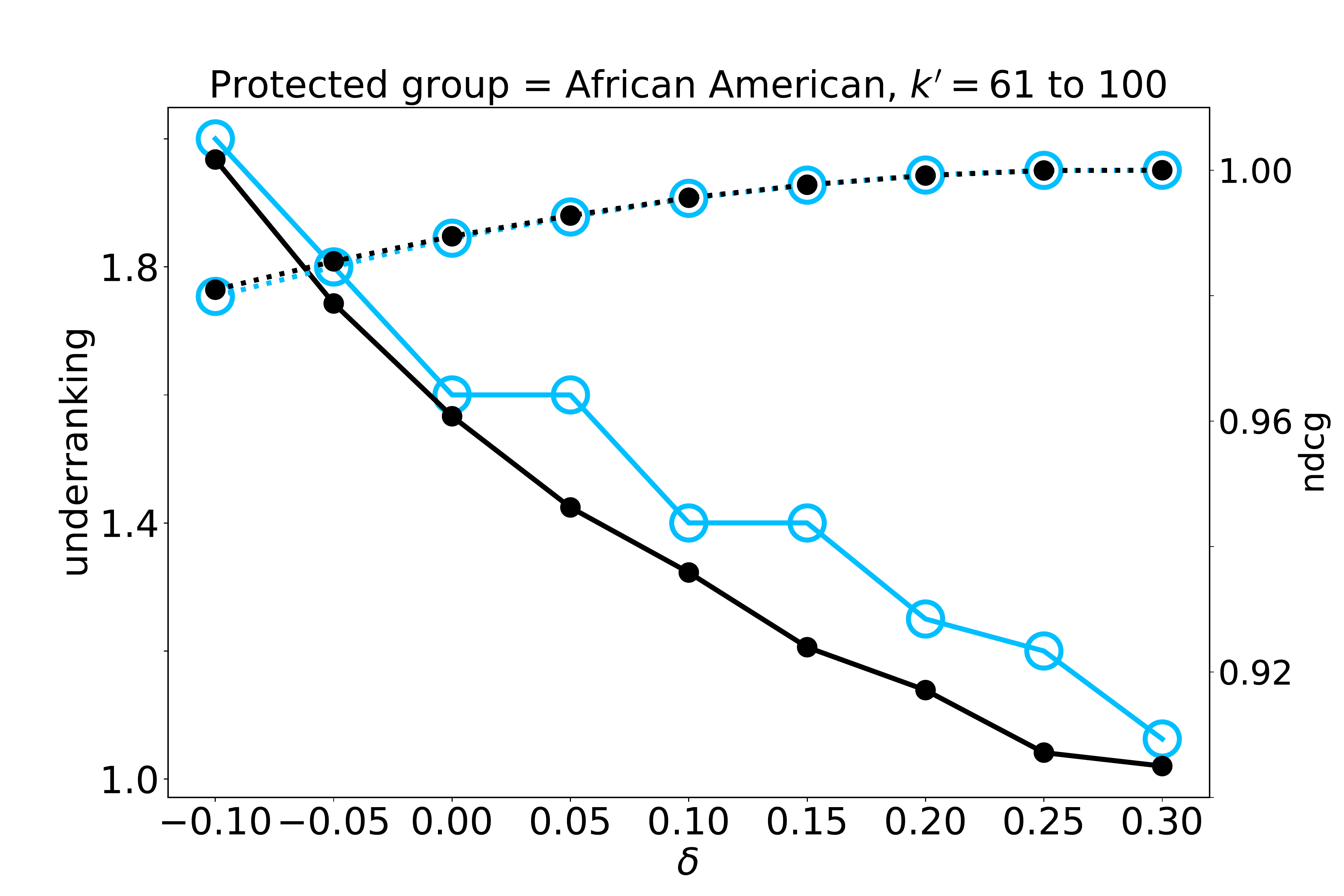} 
		\caption{Underranking, nDCG at top $100$ ranks.}
	\end{subfigure}
	\caption{Results on the COMPAS Recidivism dataset with \textit{African American} as the protected group.}
	\label{fig:compas_race_rev_block}
\end{figure}

\section{Conclusion}
Previous works involving group-fair ranking are mainly focused on the trade-off between its utility and group fairness.
We presented the first (to the best of our knowledge) algorithm that takes a true ranking and outputs another ranking with simultaneous group fairness and underranking guarantees.
Our algorithm achieves the best of both underranking and group fairness compared to the state-of-the-art group-fair ranking algorithms.
It also works in the case of more than two disjoint groups, and with different group fairness constraints for each of these groups.
One limitation of our work (and other re-ranking algorithms) is that it requires the true ranking as input. All our theoretical guarantees are with respect to this true ranking; in practice, a true merit-based ranking may be debatable or unavailable due to incomplete data, unobserved features, legal and ethical considerations behind the downstream application of these rankings, etc.

\paragraph{Acknowledgements.}
AL was supported in part by SERB Award ECR/2017/003296 and a Pratiksha Trust Young Investigator Award. AL is also grateful to Microsoft Research for supporting this collaboration.

\bibliographystyle{alpha}
\bibliography{references}

\newcommand{\etalchar}[1]{$^{#1}$}
\begin{thebibliography}{NCGW20}

\bibitem[ALMK16]{compas_propublica}
Julia Angwin, Jeff Larson, Surya Mattu, and Lauren Kirchner.
\newblock Machine bias, 2016.

\bibitem[AT05]{Adomavicius2005}
G.~{Adomavicius} and A.~{Tuzhilin}.
\newblock Toward the next generation of recommender systems: a survey of the
  state-of-the-art and possible extensions.
\newblock {\em IEEE Transactions on Knowledge and Data Engineering},
  17:734--749, 2005.

\bibitem[BCD{\etalchar{+}}19]{fairness_in_recommendation_ranking_through_pairwise_comparisons}
Alex Beutel, J.~Chen, T.~Doshi, H.~Qian, L.~Wei, Y.~Wu, L.~Heldt, Zhe Zhao,
  L.~Hong, Ed~Huai hsin Chi, and Cristos Goodrow.
\newblock Fairness in recommendation ranking through pairwise comparisons.
\newblock {\em Proceedings of the 25th ACM SIGKDD International Conference on
  Knowledge Discovery \& Data Mining}, 2019.

\bibitem[BHN19]{Barocas2019}
Solon Barocas, Moritz Hardt, and Arvind Narayanan.
\newblock {\em Fairness and Machine Learning}.
\newblock fairmlbook.org, 2019.

\bibitem[BP98]{BrinPage1998}
Sergey Brin and Lawrence Page.
\newblock The anatomy of a large-scale hypertextual web search engine.
\newblock In {\em Proceedings of the Seventh International Conference on World
  Wide Web 7}, page 107–117, NLD, 1998. Elsevier Science Publishers B. V.

\bibitem[BS16]{Barocas2016}
Solon Barocas and Andrew~D. Selbst.
\newblock Big data's disparate impact.
\newblock {\em California Law Review}, 104(3):671--732, 2016.

\bibitem[Cas19]{Castillo2019survey}
Carlos Castillo.
\newblock Fairness and transparency in ranking.
\newblock {\em Special Interest Group on Information Retrieval Forum},
  52(2):64–71, January 2019.

\bibitem[CQL{\etalchar{+}}07]{learning_to_rank_from_pairwise_approach_to_listwise_approach}
Zhe Cao, Tao Qin, Tie-Yan Liu, Ming-Feng Tsai, and Hang Li.
\newblock Learning to rank: From pairwise approach to listwise approach.
\newblock In {\em Proceedings of the 24th International Conference on Machine
  Learning}, ICML ’07, page 129–136, New York, NY, USA, 2007. Association
  for Computing Machinery.

\bibitem[Cro04]{Crosby2004}
F.J. Crosby.
\newblock {\em Affirmative Action is Dead: Long Live Affirmative Action}.
\newblock Current perspectives in psychology. Yale University Press, 2004.

\bibitem[CSV18]{ranking_with_fairness_constraints}
L.~Elisa Celis, Damian Straszak, and Nisheeth~K. Vishnoi.
\newblock Ranking with fairness constraints.
\newblock In {\em 45th International Colloquium on Automata, Languages, and
  Programming, {ICALP} 2018, July 9-13, 2018, Prague, Czech Republic}, volume
  107 of {\em LIPIcs}, pages 28:1--28:15. Schloss Dagstuhl - Leibniz-Zentrum
  f{\"{u}}r Informatik, 2018.

\bibitem[DG17]{german_credit}
Dheeru Dua and Casey Graff.
\newblock {UCI} machine learning repository, 2017.

\bibitem[GAK19]{Geyik2019}
Sahin~Cem Geyik, Stuart Ambler, and Krishnaram Kenthapadi.
\newblock Fairness-aware ranking in search \& recommendation systems with
  application to linkedin talent search.
\newblock In {\em Proceedings of the 25th ACM SIGKDD International Conference
  on Knowledge Discovery \& Data Mining}, KDD '19, page 2221–2231, New York,
  NY, USA, 2019. Association for Computing Machinery.

\bibitem[HPS16]{Hardt2016}
Moritz Hardt, Eric Price, and Nathan Srebro.
\newblock Equality of opportunity in supervised learning.
\newblock In {\em Proceedings of the 30th International Conference on Neural
  Information Processing Systems}, NIPS'16, page 3323–3331, Red Hook, NY,
  USA, 2016. Curran Associates Inc.

\bibitem[JK00]{ir_evaluation_methods}
Kalervo J\"{a}rvelin and Jaana Kek\"{a}l\"{a}inen.
\newblock Ir evaluation methods for retrieving highly relevant documents.
\newblock In {\em Proceedings of the 23rd Annual International ACM SIGIR
  Conference on Research and Development in Information Retrieval}, SIGIR '00,
  page 41–48, New York, NY, USA, 2000. Association for Computing Machinery.

\bibitem[KLH16]{Kofler2016}
Christoph Kofler, Martha Larson, and Alan Hanjalic.
\newblock User intent in multimedia search: A survey of the state of the art
  and future challenges.
\newblock {\em ACM Comput. Surv.}, 49(2), 2016.

\bibitem[MRS08]{introduction_to_information_retrieval}
Christopher~D. Manning, Prabhakar Raghavan, and Hinrich Sch\"{u}tze.
\newblock {\em Introduction to Information Retrieval}.
\newblock Cambridge University Press, USA, 2008.

\bibitem[NCGW20]{pairwise_fairness_for_ranking_and_regression}
Harikrishna Narasimhan, Andrew Cotter, Maya Gupta, and Serena Wang.
\newblock Pairwise fairness for ranking and regression.
\newblock {\em Proceedings of the AAAI Conference on Artificial Intelligence},
  34(04):5248--5255, Apr. 2020.

\bibitem[Nob18]{Noble2018}
Safiya~Umoja Noble.
\newblock {\em Algorithms of Oppression: How Search Engines Reinforce Racism}.
\newblock NYU Press, 2018.

\bibitem[O'N16]{ONeil2016}
Cathy O'Neil.
\newblock {\em Weapons of Math Destruction: How Big Data Increases Inequality
  and Threatens Democracy}.
\newblock Crown Publishing Group, USA, 2016.

\bibitem[Par11]{Pariser2011}
Eli Pariser.
\newblock {\em The Filter Bubble: What the Internet Is Hiding from You}.
\newblock Penguin Group , The, 2011.

\bibitem[PZZ{\etalchar{+}}19]{Pei2019}
Changhua Pei, Yi~Zhang, Yongfeng Zhang, Fei Sun, Xiao Lin, Hanxiao Sun, Jian
  Wu, Peng Jiang, Junfeng Ge, Wenwu Ou, and Dan Pei.
\newblock Personalized re-ranking for recommendation.
\newblock In {\em Proceedings of the 13th ACM Conference on Recommender
  Systems}, RecSys '19, page 3–11, New York, NY, USA, 2019. Association for
  Computing Machinery.

\bibitem[SJ19]{policy_learning_for_fairness_in_ranking}
Ashudeep Singh and Thorsten Joachims.
\newblock Policy learning for fairness in ranking.
\newblock In {\em Advances in Neural Information Processing Systems},
  volume~32, pages 5426--5436. Curran Associates, Inc., 2019.

\bibitem[Tav]{Tavani2020}
Herman Tavani.
\newblock {Search Engines and Ethics}.
\newblock In Edward~N. Zalta, editor, {\em The {Stanford} Encyclopedia of
  Philosophy}. Metaphysics Research Lab, Stanford University, fall 2020
  edition.

\bibitem[YLS20]{causal_intersectionality_for_fair_ranking}
Ke~Yang, Joshua~R. Loftus, and Julia Stoyanovich.
\newblock Causal intersectionality for fair ranking.
\newblock {\em ArXiv}, abs/2006.08688, 2020.

\bibitem[YS17]{measuring_fairness_in_ranked_outputs}
Ke~Yang and Julia Stoyanovich.
\newblock Measuring fairness in ranked outputs.
\newblock In {\em Proceedings of the 29th International Conference on
  Scientific and Statistical Database Management}, SSDBM ’17, New York, NY,
  USA, 2017. Association for Computing Machinery.

\bibitem[ZBC{\etalchar{+}}17]{fa_ir_a_fair_top_k_ranking_algorithm}
Meike Zehlike, Francesco Bonchi, Carlos Castillo, Sara Hajian, Mohamed Megahed,
  and Ricardo Baeza-Yates.
\newblock Fa*ir: A fair top-k ranking algorithm.
\newblock In {\em Proceedings of the 2017 ACM on Conference on Information and
  Knowledge Management}, CIKM ’17, page 1569–1578, New York, NY, USA, 2017.
  Association for Computing Machinery.

\bibitem[ZC20]{reducing_disparate_exposure_in_ranking}
Meike Zehlike and Carlos Castillo.
\newblock {\em Reducing Disparate Exposure in Ranking: A Learning To Rank
  Approach}, page 2849–2855.
\newblock Association for Computing Machinery, New York, NY, USA, 2020.

\bibitem[ZWS{\etalchar{+}}13]{learning_fair_representations}
Rich Zemel, Yu~Wu, Kevin Swersky, Toni Pitassi, and Cynthia Dwork.
\newblock Learning fair representations.
\newblock In Sanjoy Dasgupta and David McAllester, editors, {\em Proceedings of
  the 30th International Conference on Machine Learning}, volume~28 of {\em
  Proceedings of Machine Learning Research}, pages 325--333, Atlanta, Georgia,
  USA, 2013. PMLR.

\end{thebibliography}

\end{document}